\documentclass[12pt]{amsart}
\usepackage{a4wide}
\usepackage{amssymb}
\usepackage{amsmath}
\usepackage{amsfonts}
\usepackage{amsthm}
\usepackage{enumerate}
\usepackage{graphicx}
\usepackage{epsfig}
\usepackage{afterpage}
\usepackage{longtable}
\usepackage[usenames,dvipsnames]{color}
\usepackage{hyperref}
\usepackage[normalem]{ulem}
\graphicspath{{./Figures/}}

\newtheorem{thm}{Theorem}
\newtheorem{lem}[thm]{Lemma}
\newtheorem{cor}[thm]{Corollary}
\newtheorem{prop}[thm]{Proposition}
\newtheorem{conj}{Conjecture}

\theoremstyle{definition}
\newtheorem{defn}[thm]{Definition}

\theoremstyle{remark}
\newtheorem*{rmk}{Remark}

\newcommand{\eps}{\varepsilon}

\newcommand{\DEF}{{\,:=\,}}

\newcommand{\PT}[1]{\mathbf{#1}}

\DeclareMathOperator{\dd}{\mathrm{d}}

\DeclareMathOperator{\atanh}{atanh}

\DeclareMathOperator{\betafcn}{B}
\DeclareMathOperator{\bal}{Bal}
\DeclareMathOperator{\CAP}{cap}

\DeclareMathOperator{\gammafcn}{\Gamma}

\DeclareMathOperator{\GoncharA}{G_1}
\DeclareMathOperator{\GoncharB}{G_2}
\DeclareMathOperator{\GoncharC}{G_3}
\DeclareMathOperator{\GoncharD}{G_4}
\DeclareMathOperator{\kelvin}{K}
\DeclareMathOperator{\kelvinMEAS}{\mathcal{K}}
\DeclareMathOperator{\LOG}{log}

\DeclareMathOperator{\digammafcn}{\psi}

\DeclareMathOperator{\EllipticE}{E}
\DeclareMathOperator{\EllipticK}{K}

\DeclareMathOperator{\supp}{supp}

\DeclareMathOperator{\HyperF}{F}
\DeclareMathOperator{\HyperTildeF}{\tilde{F}}
\newcommand{\Hypergeom}[5]{{\sideset{_#1}{_#2}\HyperF\!\left(\substack{\displaystyle#3\\\displaystyle#4};#5\right)}}
\newcommand{\HypergeomReg}[5]{{\sideset{_#1}{_#2}\HyperTildeF\!\left(\substack{\displaystyle#3\\\displaystyle#4};#5\right)}}

\newcommand{\Pochhsymb}[2]{{\left(#1\right)_{#2}}}

\allowdisplaybreaks[1]

\title[An  Electrostatics Problem on the Sphere]{An  Electrostatics Problem on the Sphere Arising from a Nearby Point Charge}
\author[J. S. Brauchart, P. D. Dragnev, E. B. Saff]{Johann S. Brauchart, Peter D. Dragnev\textdagger, Edward B. Saff\textdaggerdbl} 
\thanks{\noindent The research of this author was supported, in part, by an APART-Fellowship of the Austrian Academy of Sciences, a Grants-in-Aid program of ORESP at IPFW, and, by the Australian Research Council. \\
\textdagger The research of this author was supported, in part, by a Grants-in-Aid program of ORESP at IPFW and by a grant from the Simons Foundation no. 282207. \\
\textdaggerdbl The research of this author was supported, in
part, by the U. S. National Science Foundation under grant DMS-1109266 as well as by an Australian Research Council Discovery grant.}

\date{\today}

\begin{document}

\address{J. S. Brauchart:
School of Mathematics and Statistics, 
University of New South Wales, 
Sydney, NSW, 2052, 
Australia }
\address{P. D. Dragnev:
Department of Mathematical Sciences,
Indiana University - Purdue University,
Fort Wayne, IN 46805,
USA}
\address{E. B. Saff:
Center for Constructive Approximation, 
Department of Mathematics, 
Vanderbilt University, 
Nashville, TN 37240, 
USA}
\email{j.brauchart@unsw.edu.au}
\email{dragnevp@ipfw.edu}
\email{Edward.B.Saff@Vanderbilt.Edu}

\begin{abstract}
For a positively charged insulated  $d$-dimensional sphere we investigate how the distribution of this charge is affected by proximity to a nearby positive or negative point charge when the system is governed by a Riesz $s$-potential $1/r^s, s>0,$ where $r$ denotes Euclidean distance between point charges. Of particular 
interest are those distances from the point charge to the sphere for which the equilibrium charge distribution is no longer supported on the whole of the sphere (i.e. spherical caps of negative charge appear). 
Arising from this problem attributed to A. A. Gonchar are sequences of polynomials of a complex variable that have some fascinating properties regarding their zeros.  
\end{abstract}


\maketitle


\section{Introduction}

For the insulated unit sphere $\mathbb{S}^d$ in $\mathbb{R}^{d+1}$ of total charge +1 on which particles interact according to the Riesz-$s$ potential $1/r^s$, $s > 0$, where $r$ is the Euclidean distance between two particles, the equilibrium distribution of charge is uniform; that is, given  by normalized surface area measure $\sigma_d$ on $\mathbb{S}^d$. However, in the presence of an ``external field'' due to a nearby point charge the equilibrium distribution changes. A positive external field repels charge away from the portion of the sphere near the source and may even clear a spherical cap of charge, whereas a negative external field attracts charge nearer to the source, thus `thinning out' a region on $\mathbb{S}^d$ opposite to the direction of the source. In the Coulomb case $s = 1$ and $d = 2$ this is a well-studied problem in electrostatics (cf., e.g., \cite{Ja1998}). Here we deviate from this classical setting and show that new and sometimes surprising phenomena can be observed.  
The outline of the paper is as follows.
 
In Section~\ref{sec:Gonchar.s.Question} we introduce and discuss a problem (\emph{Gonchar's problem}) concerning the critical distance from a unit point charge to $\mathbb{S}^d$ such that the support of the $s$-equilibrium measure on the sphere (for Riesz $s$-potential $1/r^s$, $s>0$) is no longer all of the sphere when the point charge is at any closer distance. We shall make this more precise below. Our starting point is the solution to the ``signed equilibrium problem'' for a positive charge outside the sphere, which is intimately connected with the solution of this external field energy problem. Using similar methods, we are able to extend the results in \cite{BrDrSa2009} to external fields due to a positive/negative point charge inside/outside the sphere.
We shall analyze Gonchar's problem for different magnitudes and for both positive and negative charges  which, in turn, will provide a more comprehensive picture than presented in \cite{BrDrSa2009} and \cite{BrDrSa2012}, and will reveal some interesting new phenomena. For example, in the case when $d-s$ is an even positive integer, a logarithmic term appears in the formula for the critical distance; and for a negative external field due to a source inside the sphere, a crucial issue is whether the Riesz kernel is strictly subharmonic ($d-1 < s < d$) or strictly superharmonic ($0 < s < d-1$).

For the logarithmic potential, we are able to provide a complete answer for Gonchar's problem. In the classical (harmonic) case $s = d - 1$ (and more generally, when ${d - s}$ is an odd positive integer) the critical distance appears as a zero of a certain family of polynomials indexed by the dimension $d$.
The classical case and the associated polynomials are discussed in some detail in \cite{BrDrSa2012}. By allowing a negative as well as positive charge $q$ and allowing this charge to be inside or outside the sphere, we derive a family of polynomials for each of the four combinations: (i) $0 \leq R < 1$ and $q > 0$, (ii) $R > 1$ and $q > 0$, (iii) $0 \leq R < 1$ and $q < 0$, and (iv) $R > 1$ and $q < 0$. In Section~\ref{sec:polynomial.families}, we investigate these four families of polynomials arising from Gonchar's problem. Figure~\ref{fig2} illustrates the qualitative patterns of their zeros in the classical case $s = d - 1$ and illustrates how these families complement each other. The last two displays of this figure are for $s = d - 3$. Figure~\ref{fig2b} contrasts the cases $s = d - 3$ and $s = d - 5$ for ``weak'' ($q=1/10$), ``canonical'' ($q=1$) and ``strong'' ($q=10$) external fields.

In \cite{BrDrSa2014} we discussed the Riesz external field problem due to a negative point charge above the South Pole of a positively charged unit sphere (total charge $1$) and we derived the extremal and the signed equilibria on spherical caps for $d - 2 < s < d$. In Section~\ref{sec:negatively.charged.external field}, we provide the details of this derivation and consider also the limiting cases when $s \to d - 2$. 

Section~\ref{sec:proofs} contains the remaining proofs. In the Appendix we study the $s$-potential of uniform (normalized) surface area measure on the $d$-sphere in more detail.

\section{Potential Theoretic Setting for Gonchar's Problem}
\label{sec:Gonchar.s.Question}

Let $A$ be a compact subset of $\mathbb{S}^d$. Consider the class $\mathcal{M}(A)$ of unit positive Borel measures supported on $A$. The {\em Riesz-$s$ potential} and {\em Riesz-$s$ energy} of a measure $\mu\in \mathcal{M}(A)$ modeling a (positive) charge distribution of total charge $1$ on $A$ are defined as
\begin{equation*}
U_s^\mu (\PT{x}) \DEF \int |\PT{x}-\PT{y}|^{-s} \dd\mu(\PT{y}), \qquad \mathcal{I}_s [\mu] \DEF \int \int |\PT{x}-\PT{y}|^{-s} \dd\mu(\PT{x}) \dd\mu(\PT{y}).
\end{equation*}
The \emph{Riesz-$s$ energy} and the \emph{$s$-capacity} of $A$ are given by
\begin{equation*}
W_s(A) \DEF \inf\Big\{ \mathcal{I}_s [\mu] : \mu \in \mathcal{M}(A) \Big\}, \qquad \CAP_s(A) = 1 / W_s(A).
\end{equation*}
It is well-known from potential-theory (cf. Landkof~\cite{La1972}) that if $A$ has positive $s$-capacity, then there always exists a unique measure $\mu_{A,s} \in \mathcal{M}(A)$, which is called the {\em $s$-equilibrium measure on~$A$}, such that $W_s(A) = \mathcal{I}_s[\mu_{A,s}]$. For example, $\sigma_d$ is the $s$-equilibrium measure on~$\mathbb{S}^d$ for each $0<s<d$.
A standard argument utilizing the uniqueness of the $s$-equilibrium measure $\mu_{A,s}$ shows that it is the limit distribution (in the weak-star sense) of a sequence of minimal $s$-energy $N$-point configurations on $A$ minimizing the {\em discrete $s$-energy}
\begin{equation*}
E_s(\PT{x}_1, \dots, \PT{x}_N) \DEF \mathop{\sum_{j=1}^N \sum_{k=1}^N}_{j \neq k} \frac{1}{\left| \PT{x}_j - \PT{x}_k \right|^s}, \qquad \PT{x}_1, \dots, \PT{x}_N \in A,
\end{equation*}
over all $N$-point systems on $A$. (For the discrete $s$-energy problem we refer to \cite{HaSa2004}.)

\subsection*{Weighted Energy and External Fields}
We are concerned with the Riesz external field generated by a positive or negative point charge of amount $q$ located at $\PT{a} = (\PT{0}, R)$ on the polar axis with $0 \leq R < 1$ or $R > 1$. Such a field is given by
\begin{equation} \label{externalfield}
Q( \PT{x} ) = Q_{R,q,s}( \PT{x} ) \DEF q / \left| \PT{x} - \PT{a} \right|^s, \qquad \PT{x} \in \mathbb{R}^{d+1}.
\end{equation}
The Riesz-$s$ external field on a compact subset $A \subset \mathbb{S}^d$ with positive $s$-capacity is $Q$ restricted to $A$. (For simplicity, we use the same symbol.) The \emph{weighted $s$-energy $\mathcal{I}_Q[\mu]$} associated with such a continuous external field and its extremal value $V_Q$ are defined as
\begin{equation*}
\mathcal{I}_Q[\mu] \DEF \mathcal{I}_s[\mu] + 2 \int Q(\PT{x}) \dd\mu(\PT{x}), \qquad V_Q(A) \DEF \inf \left\{ \mathcal{I}_Q[\mu] : \mu \in \mathcal{M}(A) \right\}.
\end{equation*}
A measure $\mu_Q\in \mathcal{M}(A)$ such that $\mathcal{I}_Q[\mu] = V_Q(A)$ is called an \emph{$s$-extremal (or positive equilibrium) measure on $A$ associated with $Q$}. This measure is unique and it satisfies the Gauss variational inequalities (cf. \cite{DrSa2007}) \footnote{Note that a continuous negative field can be made into a positive one by adding a fixed constant.}
\begin{align}
U_s^{\mu_Q}(\PT{x}) + Q(\PT{x}) &\geq F_Q(A) \qquad \text{everywhere on $A$,} \label{geqineq} \\
U_s^{\mu_Q}(\PT{x}) + Q(\PT{x}) &\leq F_Q(A) \qquad \text{everywhere on support $\supp(\mu_Q)$ of $\mu_Q$,} \label{leqineq}
\end{align}
where
\begin{equation*}
F_Q(A) \DEF V_Q(A) - \int Q(\PT{x}) \, \dd \mu_Q(\PT{x}).
\end{equation*}
In fact, once $\supp(\mu_Q)$ is known, the equilibrium measure $\mu_Q$ can be recovered by solving the integral equation
\begin{equation*}
U_s^{\mu}(\PT{x}) + Q(\PT{x}) = 1 \qquad \text{everywhere on $A$}
\end{equation*}
for positive measures $\mu$ supported on $A$.
In the absence of an external field ($Q \equiv 0$) and when $\CAP_s(A) > 0$ the measure $\mu_Q$ coincides with $\mu_{A,s}$.

Riesz external fields due to a positive charge on the sphere $\mathbb{S}^d$ were instrumental in the derivation of separation results for minimum Riesz-$s$ energy points on $\mathbb{S}^d$ for $s \in (d-2,d)$ (see \cite{DrSa2007}).
In \cite{BrDrSa2009} we studied Riesz external fields due to a positive charge above $\mathbb{S}^d$ which led to a discussion of a fascinating sequence of polynomials arising from answering Gonchar's problem for the harmonic case in \cite{BrDrSa2012}.
The least separation of minimal energy configurations on $\mathbb{S}^d$ subjected to an external field is investigated in \cite{BrDrSa2014}.
In \cite{BrDrSa2009} we also developed a technique for finding the extremal measure associated with more general axis-supported fields.\footnote{The case $d=1$, $s=0$, where the source is a point on the unit circle, was investigated in \cite{LaSaVa1979}.} For external fields in the most general setting we refer the reader to the work of Zori{\u\i}~\cite{Zo2003,Zo2003a,Zo2004} (also cf.~\cite{HaWeZo2012}).

\subsection*{Signed Equilibrium}
The Gauss variational inequalities \eqref{geqineq} and \eqref{leqineq} for $A=\mathbb{S}^d$ imply that the weighted equilibrium potential is constant everywhere on the support of the measure $\mu_Q$ on $\mathbb{S}^d$. In general, one can not expect that the support of $\mu_Q$ is all of the sphere. A sufficiently strong external field (large $q>0$ or small $R>1$) would thin out the charge distribution around the North Pole and even clear a spherical cap of charge. (In this ``insulated sphere'' setting there is no negative charge that would be attracted to the North Pole.) By enforcing constant weighted potential everywhere on the sphere, in general, one has a signed measure as a solution. (This corresponds to a grounded sphere.) 
In the classical Coulomb case ($d=2$, $s=1$) a standard electrostatic problem is to find the charge density (signed measure) on a charged, insulated, conducting sphere in the presence of a point charge $q$ off the sphere (see \cite[Ch.~2]{Ja1998}). This motivates the following definition (cf. \cite{Dr2007}).

\begin{defn} \label{def:signed.equilibrium} Given a compact subset $A\subset \mathbb{R}^p$ ($p\geq 3$) and an external field $Q$ on $A$, we call a signed measure $\eta_{Q}=\eta_{A,Q,s}$ supported on $A$ and of total charge $\eta_{Q}(A)=1$ {\em a signed $s$-equilibrium on $A$ associated with
$Q$} if its weighted Riesz $s$-potential is constant on $A$; i.e., 
\begin{equation}
U_s^{\eta_{Q}}(\PT{x}) + Q(\PT{x}) = G_{A,Q,s} \qquad \text{everywhere on $A$.} \label{signedeq}
\end{equation}
\end{defn}
Physicists usually prefer neutral charge $\eta_{Q}(A)=0$. However, for the applications here it is more convenient to have the normalization $\eta_{Q}(A)=1$. It can be shown that if a signed equilibrium $\eta_{Q}$ on $A$ exists, then it is unique (see \cite{DrSa2007}). We remark that the determination of signed equilibria is a substantially easier problem than that of finding non-negative extremal measures. However, the solution to the former problem is useful in solving the latter problem. In \cite{BrDrSa2009} it is shown for $q>0$ and $R>1$ that the signed $s$-equilibrium $\eta_{Q}$ on $\mathbb{S}^d$ associated with the Riesz external field \eqref{externalfield} is absolutely continuous with respect to the normalized surface area measure on $\mathbb{S}^d$. 
Using ``Imaginary inversion'' (cf. Landkof~\cite{La1972}), the proof can be extended to hold for the class of fields considered here. We remark that in the Coulomb case ($s = 1$ and $d = 2$) this result is well-known from elementary physics (cf. \cite[p.~61]{Ja1998}).

\begin{thm} \label{thm:signed.equilibrium.sphere}
Let $0 < s <  d$. The signed $s$-equilibrium $\eta_Q$ on $\mathbb{S}^d$ associated with the external field $Q$ of \eqref{externalfield} is absolutely continuous with respect to the normalized surface area measure on $\mathbb{S}^d$; that is, $\dd \eta_Q( \PT{x} ) = \eta_{R,q,s}^\prime( \PT{x} ) \, \dd \sigma_d( \PT{x} )$, and its density is given by
\begin{equation} \label{eta.Q}
\eta_{R,q,s}^\prime(\PT{x}) = 1 + \frac{q U_s^{\sigma_d} (\PT{a})}{W_s(\mathbb{S}^d)} - \frac{q\left|R^2-1\right|^{d-s}}{W_s(\mathbb{S}^d)\left|\PT{x}-\PT{a}\right|^{2d-s}}, \qquad \PT{x} \in \mathbb{S}^d.
\end{equation}
(The charge $q$ can be positive or negative and the distance of the charge to the sphere center satisfies $0 \leq R < 1$ or $R > 1$.) Moreover, the weighted $s$-potential of $\eta_Q$ on $\mathbb{S}^d$ equals
\begin{equation} \label{eq:G.sphere.Q.s}
G_{\mathbb{S}^d,Q,s} = W_s( \mathbb{S}^d ) + q U_s^{\sigma_d}( \PT{a} ).
\end{equation}
\end{thm}

In the following we shall use the Pochhammer symbol 
\begin{equation*}
\Pochhsymb{a}{0} \DEF 1, \qquad \Pochhsymb{a}{n} \DEF a (a+1) \cdots (a+n-1), \quad n \geq 1,
\end{equation*}
which can be expressed in terms of the Gamma function $\gammafcn$ by means of $\Pochhsymb{a}{n} = \gammafcn( n + a ) / \gammafcn( a )$ whenever $n + a$ is not an integer $\leq 0$, and the Gauss hypergeometric function and its regularized form with series expansions
\begin{equation} \label{eq:HypergeomSeries}
\Hypergeom{2}{1}{a,b}{c}{z} \DEF \sum_{n=0}^\infty \frac{\Pochhsymb{a}{n}\Pochhsymb{b}{n}}{\Pochhsymb{c}{n}} \frac{z^n}{n!}, \quad  \HypergeomReg{2}{1}{a,b}{c}{z} \DEF \sum_{n=0}^\infty \frac{\Pochhsymb{a}{n}\Pochhsymb{b}{n}}{\gammafcn(n+c)} \frac{z^n}{n!}, \qquad  |z|<1. 
\end{equation}
We shall also use the incomplete Beta function and the Beta function, 
\begin{equation} \label{eq:betafnc}
\betafcn(x;\alpha,\beta) \DEF \int_{0}^x v^{\alpha-1} \left( 1 - v \right)^{\beta-1} \dd v, \qquad \betafcn(\alpha,\beta) \DEF \betafcn(1; \alpha,\beta),
\end{equation}
and the regularized incomplete Beta function
\begin{equation}
\mathrm{I}(x;a,b) \DEF \betafcn(x;a,b) \big/ \betafcn(a,b). \label{regbetafnc}
\end{equation}

The density $\eta_{R,q,s}^\prime$ of \eqref{eta.Q} is given in terms of the $s$-energy of $\mathbb{S}^d$ \footnote{$W_s(\mathbb{S}^d)$ can be obtained using the Funk-Hecke formula \cite{Mu1966}. Also, cf. Landkof~\cite{La1972}.}, 
\begin{equation} \label{eq:W.s.S.d}
W_s(\mathbb{S}^d) = \int \int \frac{1}{\left| \PT{x} - \PT{y} \right|^s} \dd \sigma_d( \PT{x} ) \dd \sigma_d( \PT{y} ) = \frac{\gammafcn(d)\gammafcn((d-s)/2)}{2^s\gammafcn(d/2)\gammafcn(d-s/2)}, 
\end{equation}
and the Riesz-$s$ potential of the uniform normalized surface area measure $\sigma_d$ evaluated at the location of the source of the external field (cf. \cite[Theorem~2]{BrDrSa2009}),  
\begin{equation} \label{eq:s.potential}
U_s^{\sigma_d}(\PT{a}) = \left( R + 1 \right)^{-s} \Hypergeom{2}{1}{s/2,d/2}{d}{4R \big/ \left(R+1\right)^2}, \qquad R = | \PT{a} |.
\end{equation}
By abuse of notation we shall also write $U_s^{\sigma_d}( R )$. From formula~\eqref{eta.Q} we observe that the minimum value of the density $\eta_{R,q,s}^\prime$ is attained at the North Pole $\PT{p}$ if $q>0$,
\begin{equation*}
\eta_{R,q,s}^\prime(\PT{p}) = 1 + \frac{q U_s^{\sigma_d} (\PT{a})}{W_s(\mathbb{S}^d)} - \frac{q \left( R + 1 \right)^{d-s}}{W_s(\mathbb{S}^d) \left| R - 1 \right|^{d}},
\end{equation*}
and at the South Pole $-\PT{p}$ if $q<0$,
\begin{equation*}
\eta_{R,q,s}^\prime(-\PT{p}) = 1 + \frac{q U_s^{\sigma_d} (\PT{a})}{W_s(\mathbb{S}^d)} - \frac{q \left| R - 1 \right|^{d-s}}{W_s(\mathbb{S}^d) \left( R + 1 \right)^{d}}.
\end{equation*}

\begin{prop} \label{prop}
Let $0 < s < d$. For the external field $Q$ of \eqref{externalfield} with $q \neq 0$ and $0 < R < 1$ or $R > 1$, the signed $s$-equilibrium is a positive measure on all of $\mathbb{S}^d$ if and only if 
\begin{enumerate}[\bf (a)]
\item Positive external field ($q > 0$):
\begin{equation} \label{eqsigned}
\frac{W_s(\mathbb{S}^d)}{q} \geq \frac{\left(R+1\right)^{d-s}}{\left|R-1\right|^d} - U_s^{\sigma_d}(R). 
\end{equation}
\item Negative external field ($q < 0$):
\begin{equation} \label{neg.eqsigned} 
\frac{W_s(\mathbb{S}^d)}{q} \leq \frac{\left|R-1\right|^{d-s}}{\left(R+1\right)^d} - U_s^{\sigma_d}(R).
\end{equation}
\end{enumerate}
In such a case, $\mu_{Q}=\eta_{Q}$.
\end{prop}

\begin{proof}
The arguments given in \cite{BrDrSa2009} for $q>0$ and $R>1$ apply. We provide the proof for \eqref{neg.eqsigned}. If $\supp(\mu_Q) = \mathbb{S}^d$, then $\mu_Q$ is a signed equilibrium on $\mathbb{S}^d$ (the Gauss variational inequalities \eqref{geqineq} and \eqref{leqineq} hold everywhere on $\mathbb{S}^d$). By uniqueness of $\eta_Q$, $\eta_Q = \mu_Q$; hence, it is non-negative and \eqref{neg.eqsigned} holds. If \eqref{neg.eqsigned} holds, then $\eta_Q$ is a non-negative measure on $\mathbb{S}^d$ whose weighted $s$-potential is constant everywhere on $\mathbb{S}^d$; that is, $\eta_Q$ satisfies the Gauss variational inequalities with $F_Q(\mathbb{S}^d) = G_{\mathbb{S}^d, Q, s}$. By uniqueness of $\mu_Q$, $\mu_Q = \eta_Q$ and $\supp(\mu_Q) = \mathbb{S}^d$. 
\end{proof}

Note that $R = 0$ satisfies \eqref{eqsigned} and \eqref{neg.eqsigned} with strict inequality for any choice of $q \neq 0$. Given a charge $q \neq 0$, any $R \in ( 0, 1 ) \cup ( 1, \infty )$ for which equality holds in \eqref{eqsigned} or in \eqref{neg.eqsigned} is called \emph{critical distance}. At a critical distance $R^*$, the density $\eta_{R^*,q,s}^\prime$ of \eqref{eta.Q} assumes the value $0$ at one point on $\mathbb{S}^d$ (and $\eta_{R^*,q,s}^\prime$ is strictly positive away from this unique minimum). Interestingly, in the case of negative external fields due to a source inside the sphere, there can be more than one critical distance as discussed below. The critical distance(s) anchor the subintervals of radii $R$ in $( 0, 1 ) \cup ( 1, \infty )$ for which $\eta_{R,q,s}^\prime > 0$ everywhere on $\mathbb{S}^d$.

\begin{rmk}[Positive external fields]
For every fixed positive charge $q$, there is a unique critical distance $R_q$ such that for $R \geq R_q > 1$ ($0 \leq R \leq R_q < 1$) the signed $s$-equilibrium is a positive measure on the whole sphere $\mathbb{S}^d$. This follows from the fact that the right-most part of \eqref{eqsigned} is a strictly decreasing (increasing) function of~$R$ if $R>1$ ($0 < R < 1$).
\end{rmk}

The technical details for this and the next remark will be postponed until Section~\ref{sec:proofs}.

\begin{rmk}[Negative external fields]
The subtleties of the right-hand side of \eqref{neg.eqsigned}, 
\begin{equation*}
f( R ) \DEF \frac{\left|R-1\right|^{d-s}}{\left(R+1\right)^d} - U_s^{\sigma_d}(R), 
\end{equation*}
gives rise to a multitude of different, even surprising, cases (cf. Theorems~\ref{thm:Gonchar.C}, \ref{thm:Gonchar.D.superharmonic}, \ref{thm:Gonchar.D.harmonic}, and~\ref{thm:Gonchar.D.subharmonic}).  
The function $f$ is continuous on $[0,\infty)$, negative on $(0,\infty)$, and bounded from below. Therefore, \eqref{neg.eqsigned} is trivially satisfied for all  charges $q$ with $W_s( \mathbb{S}^d ) / q < \min_{R \geq 0} f(R)$; otherwise, at least one critical distance exists. 
For an exterior field source (that is, on $( 1, \infty )$) the function $f$ has the same qualitative behavior for all $0 < s < d$ in the sense that it is strictly monotonically increasing with lower bound $f(1^+) = - W_s(\mathbb{S}^d)$ and a horizontal asymptote at level~$0$. Consequently, there is a unique critical distance if $q < -1$, which is the least distance $R$ such that $\eta_Q \geq 0$ on $\mathbb{S}^d$, and none if $q \geq -1$. For an interior field source (that is, on $( 0, 1 )$) the qualitative behavior of $f$ changes with the potential-theoretic regime (superharmonic, harmonic, subharmonic $s$). Figure~\ref{fig:typical.f.R} illustrates the typical form of $f$. 
\begin{figure}[ht]
\begin{center}
\includegraphics[scale=.8]{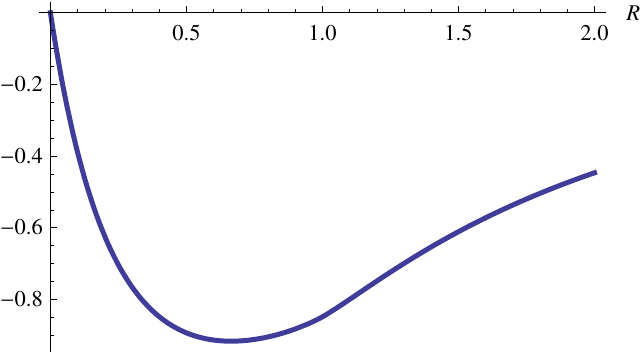}
\includegraphics[scale=.8]{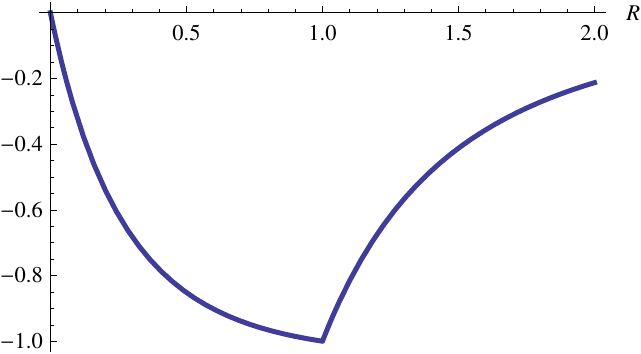}
\includegraphics[scale=.8]{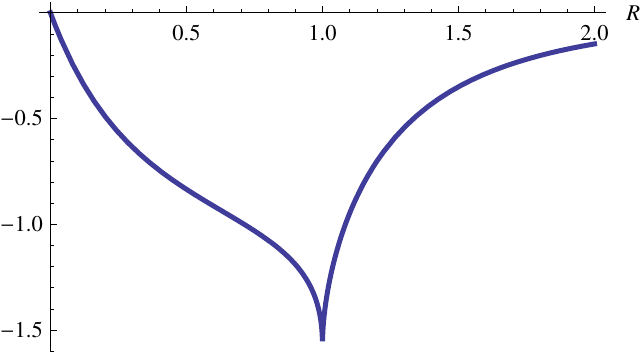}
\caption{\label{fig:typical.f.R} The typical behavior of the function $f(R)$ (right-hand side of~\eqref{neg.eqsigned}) in the strictly superharmonic, harmonic, and strictly subharmonic case shown here for $d = 3$ and $s = (d - 1)/2$, $s = d - 1$, and $s = d - 1/2$.}
\end{center}
\end{figure}
In particular, one can have more than one critical value of $R$ for which equality is assumed in \eqref{neg.eqsigned} as demonstrated for the case $d = 4$ and $s = 1$ when \eqref{neg.eqsigned} reduces to
\begin{equation} \label{demonstration}
W_s(\mathbb{S}^d) / q \leq \left( 1 - R \right)^3 \left( 1 + R \right)^{-4} + R^2 / 5 - 1, \qquad 0 \leq R < 1.
\end{equation}
The right-hand side above is convex in $(0,1)$ with a unique minimum at $R\approx 0.507122392\dots$, so that equality will hold in above relation at two radii $R_{q,1}$ and $R_{q,2}$ in $(0,1)$ near this minimum and the above relation will hold on $[0,R_{q,1}]$ and $[R_{q,2},1)$.
\end{rmk}

We remark that the complementary problem of fixing the distance $R$ and asking for the critical charge $q_R$ is trivial. A simple manipulation yields the unique $q_R$ for which equality holds in \eqref{eqsigned} or \eqref{neg.eqsigned} such that for $q > q_R > 0$ (positive external field) or for $q < \min\{ q_R, 0\}$ (negative external field) the signed $s$-equilibrium measure is a positive measure on all of~$\mathbb{S}^d$.

\subsection*{Gonchar's Problem}

A. A. Gonchar asked the following question (cf. \cite{LoMaNeEtal2013}): {\em A positive unit point charge approaching the insulated unit sphere carrying the total charge $+1$ will eventually cause a spherical cap free of charge to appear. \footnote{On a grounded sphere a negatively charged spherical cap will appear.} What is the smallest distance from the point charge to the sphere where still all of the sphere is positively charged?} 

For the classical harmonic Newtonian potential ($s=d-1$) we answer this question in~\cite{BrDrSa2009} and discuss it in detail in~\cite{BrDrSa2012}. In this particular case the $s$-energy $W_s(\mathbb{S}^d)$ of the $d$-sphere equals $1$ (see \eqref{eq:W.s.S.d}), and the mean-value property for harmonic functions implies that the $s$-potential of the normalized surface area measure $\sigma_d$ at $\PT{a}$ simplifies to $U_s^{\sigma_d}(R) = R^{1-d}$ ($R \geq 1$), cf. \eqref{eq:s.potential.A} below. 
So, when requiring that $R > 1$ and $q > 0$, Proposition~\ref{prop} yields that $\supp(\mu_Q) = \mathbb{S}^d$ if and only if the following rational relation is satisfied:
\begin{equation} \label{fundamental.relation}
1 / q \geq \left( R + 1 \right) \left( R - 1 \right)^{-d} - R^{1-d}.
\end{equation}
The critical distance $R_q$ (from the center of $\mathbb{S}^d$) is assumed when equality holds in \eqref{fundamental.relation}. Curiously, for $d=2$ (classical Coulomb case) the answer to Gonchar's problem is the Golden ratio $\phi$, that is $R_1-1$ (the distance from the unit sphere) equals $( 1 + \sqrt{5} ) / 2$ \footnote{An elementary physics argument would also show that $R-1 \geq \phi$ (for $q=1$) implies $\supp(\mu_Q) = \mathbb{S}^2$.}; and for $d=4$, the answer is the {\em Plastic constant} $P$ (defined in Eq.~\eqref{Plastic.constant} below). For general dimension $d\geq2$, the critical distance $R_q = R_q(d)$ is for positive exterior external fields a solution of the following algebraic equation
\begin{equation} \label{equation}
\GoncharA( d, q; R ) \DEF \left[ \left( R - 1 \right)^d / q - R - 1 \right] R^{d-1} + \left( R - 1 \right)^d = 0, 
\end{equation}
which follows from \eqref{fundamental.relation} and gives rise to the family of polynomials studied in \cite{BrDrSa2012}. In fact, it is shown in \cite{BrDrSa2012} that $R_q$ is the uniqe (real) zero in $(1, +\infty)$ of the Gonchar polynomial $\GoncharA(d,q; z)$.
Asymptotical analysis (see \cite[Appendix~A]{BrDrSa2012}) shows that
\begin{equation*}
R_q = 2 + \left[ \log (3 q) \right] / d + \mathcal{O}(1/d^2) \qquad \text{as $d\to\infty$.}
\end{equation*}

The answer to Gonchar's problem for the external field of \eqref{externalfield} for general parameters $0<s<d$, $R  > 1$, and $q > 0$ relies on solving a, in general, highly non-algebraic equation for the critical distance~$R_q$, namely the characteristic equation (cf. relation \eqref{eqsigned})
\begin{equation} \label{eq:characteristic.equation.exterior.positive}
\frac{W_s(\mathbb{S}^d)}{q} = \frac{\left( R + 1 \right)^{d-s}}{\left( R - 1 \right)^d} - U_s^{\sigma_d}( R ).
\end{equation}
Figure~\ref{fig3} displays the graphical solution to Gonchar's problem for the dimensions ${d = 2, 4}$.
\begin{figure}[ht]
\begin{center}
\hfill
\includegraphics[scale=.085]{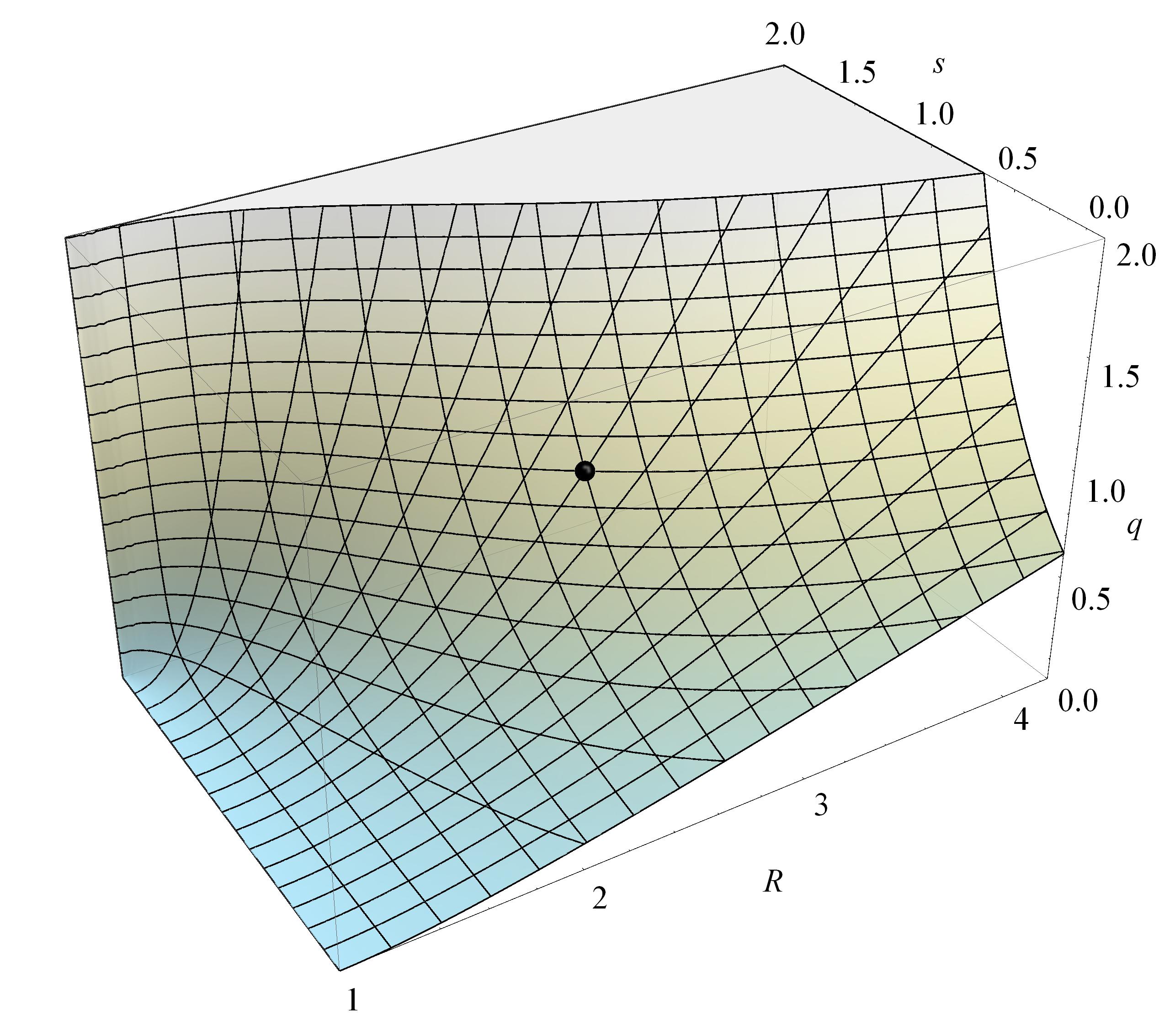} \hfill
\includegraphics[scale=.085]{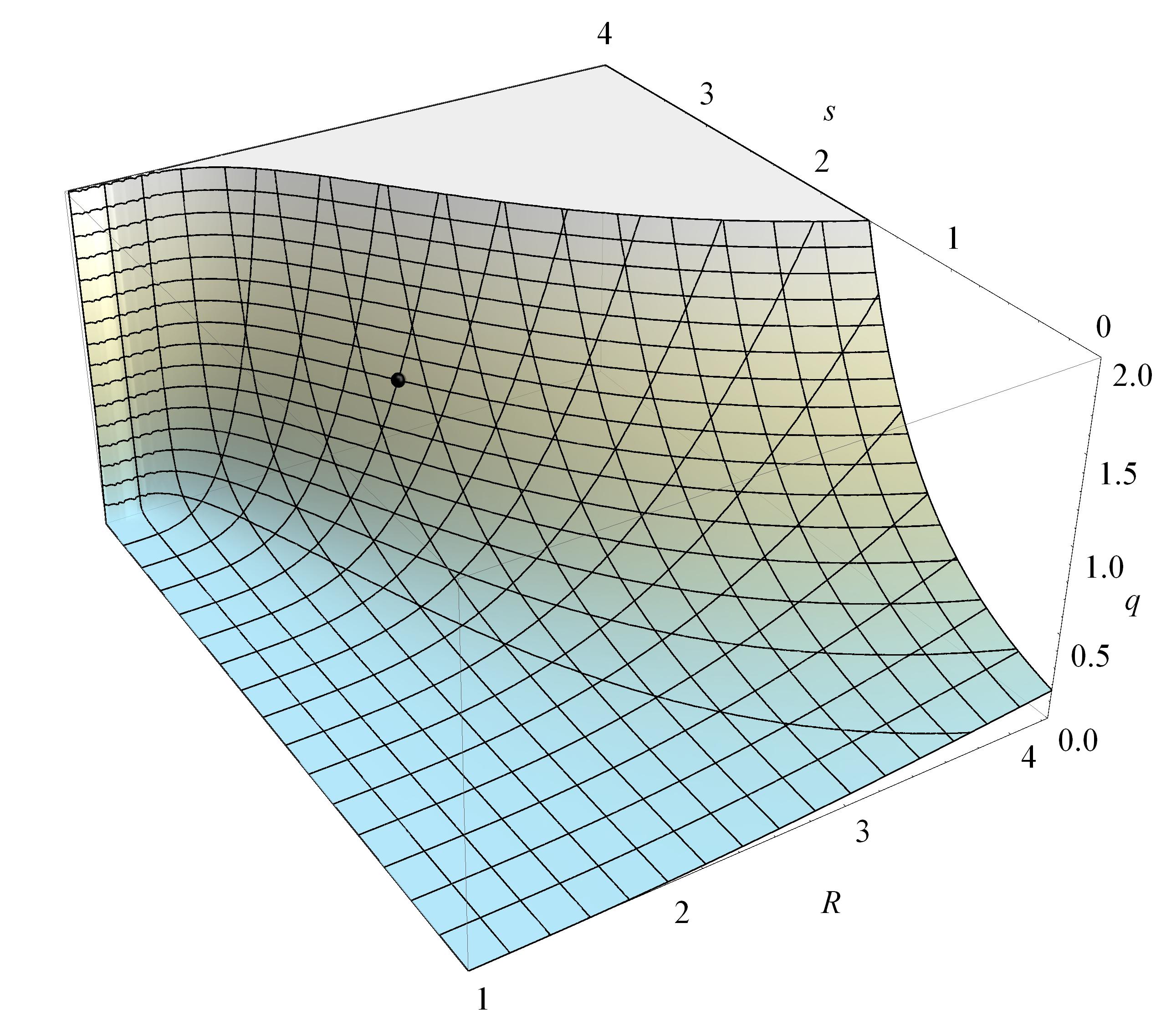} 
\caption{\label{fig3} Surfaces (equality in \eqref{eqsigned}) representing the answer to Gonchar question for $d=2,4$ for the selected ranges for $R$, $q$, and $s$. The dots indicate the solution to Eq.~\eqref{equation} (Newtonian case, $q=1$).} 
\end{center}
\end{figure}
Taking into account that $R > 1$, we can rewrite \eqref{eq:characteristic.equation.exterior.positive} as 
\begin{equation}
\GoncharA( d, s, q; R ) = 0,
\end{equation}
where we define the \emph{Gonchar function of the first kind}
\begin{equation} \label{eq:Gonchar.function.1st.kind}
\GoncharA( d, s, q; R ) \DEF \left( \frac{W_s(\mathbb{S}^d)}{q} \left( R - 1 \right)^d - \left( R + 1 \right)^{d-s} \right) R^{d-1} + R^{d-1} \left( R - 1 \right)^d U_s^{\sigma_d}( R ).
\end{equation}
Answering Gonchar's problem for $0 < s < d$, $q > 0$, and $R > 1$ amounts to finding the unique (cf. remark after Proposition~\ref{prop}) (real)\footnote{Depending on the formula used for $U_s^{\sigma_d}( R )$ (e.g., \eqref{eq:s.potential}), $\GoncharA( d, s, q; R )$ can be analytically continued to the complex $R$-plane.} zero in $(1,\infty)$ of the function $\GoncharA( d, s, q; R )$.
Moreover, the density $\eta_{R,q,s}^\prime$ of Theorem~\ref{thm:signed.equilibrium.sphere} evaluated at the North Pole can be expressed as
\begin{equation*}
\eta_{R,q,s}^\prime( \PT{p} ) = \frac{q}{W_s( \mathbb{S}^d )} \, \frac{\GoncharA( d, s, q; R )}{R^{d-1} \left( R - 1 \right)^{d}}, \qquad R > 1, q > 0, 0 < s < d.
\end{equation*}

The quadratic transformation formula for Gauss hypergeometric functions \cite[Eq.~15.3.17]{AbSt1992} applied to the formula of $U_s^{\sigma_d}( R )$ of \eqref{eq:s.potential} yields
\begin{equation} \label{eq:s.potential.A}
U_s^{\sigma_d}( R ) = \left( 1 / R \right)^{s}  \Hypergeom{2}{1}{-(d-1-s)/2,s/2}{(d+1)/2}{\left( 1 / R \right)^{2}}, \qquad R > 1.
\end{equation}
The hypergeometric function simplifies to $1$ if $s = d - 1$ (harmonic case) and reduces to a polynomial of degree $m$ in $1/R^2$ for $d - 1 - s = 2m$. Thus $R_q$ will be an algebraic number if $W_s(\mathbb{S}^d) / q$ is algebraic, which is interesting from a number-theoretic point of view. 

The following result generalizes \cite[Theorem~5]{BrDrSa2012}.

\begin{thm} \label{thm:Gonchar.A}
For the external field $Q$ of \eqref{externalfield} with $0 < s < d$, $q > 0$, and $R > 1$ the signed $s$-equilibrium is a positive measure on all of~$\mathbb{S}^d$ if and only if $R \geq R_q$, where $R_q$ is the unique (real) zero in $(1,\infty)$ of the Gonchar function $\GoncharA(d,s,q;z)$. If $s = d - 1 - 2m$ for $m$ a non-negative integer, then $\GoncharA(d,d-1-2m,q;z)$ is a polynomial. 

In particular, the solution to Gonchar's problem is given by $\rho( d, s ) = R_q - 1$. 
\end{thm}

We remark that integer values of $s$ give rise to special forms of the Gonchar function $\GoncharA$. Indeed, for $s = d - 1, d - 3, \dots$, the $s$-potential $U_s^{\sigma_d}$ and therefore the Gonchar function $\GoncharA$ reduces to a polynomial. 
If $(d-s)/2$ is a positive integer and $d$ is an even dimension, then successive application of the contiguous function relations in \cite[\S~15.5(ii)]{NIST:DLMF} to \eqref{eq:s.potential.A} lead to a linear combination of (cf. \cite[Eq.~15.4.2, 15.4.6]{NIST:DLMF})
\begin{equation*}
\Hypergeom{2}{1}{1/2,1}{3/2}{z^2} = \frac{1}{2z} \, \log \frac{1+z}{1-z} = \frac{1}{z} \, \atanh z, \qquad \Hypergeom{2}{1}{1/2,1}{1/2}{z^2} = \left( 1 - z^2 \right)^{-1} 
\end{equation*}
with unique coefficients that are rational functions of $z^2$, where here $z = 1/R$.
For the convenience of the reader, we record here that for $d = 4$ and $s = 2$, the $s$-potential in \eqref{eq:s.potential.A} reduces to (which can be verified directly by using, for example, MATHEMATICA)
\begin{equation*}
U_2^{\sigma_4}(R) = \frac{3}{8} \frac{R^2+1}{R^2} + \frac{3}{32} \frac{\left( R^2 - 1 \right)^2}{R^3} \log \frac{\left( R-1 \right)^2}{\left( R+1 \right)^2}
\end{equation*}
which yields the Gonchar function of the first kind,
\begin{equation*}
\begin{split}
\GoncharA( 4, 2, q; R ) 
&= \left( \frac{3/4}{q} \left( R - 1 \right)^4 - \left( R + 1 \right)^2 \right) R^3 + \frac{12}{32} \left( R - 1 \right)^4 R \left( R^2 + 1 \right) \\
&\phantom{=}+ \frac{3}{32} \left( R - 1 \right)^4 \left( R^2 - 1 \right)^2 \log \frac{(R - 1)^2}{(R + 1)^2},
\end{split}
\end{equation*}
and for $d = 6$ and $s = 4$ or $s = 2$ one has
\begin{align*}
U_4^{\sigma_6}(R) &= \frac{15}{32} \, \frac{R^4 - \frac{2}{3} R^2 + 1}{R^4} + \frac{15}{128} \, \frac{\left(R^2 + 1\right) \left( R^2 - 1 \right)^2}{R^5} \, \log \frac{\left( R-1 \right)^2}{\left( R+1 \right)^2}, \\
U_2^{\sigma_6}(R) &= - \frac{15}{128} \, \frac{\left(R^2 + 1\right) \left(R^4 - \frac{14}{3} R^2 + 1\right)}{R^4} - \frac{15}{512} \, \frac{\left(R^2-1\right)^4}{R^5} \, \log \frac{\left( R-1 \right)^2}{\left( R+1 \right)^2}.
\end{align*}
(This representations hold, in fact, for all $R \in (0,1) \cup (1, \infty)$.) 
Further analysis shows that for even $d$ and $s = d - 2, d - 4, \dots$, the Gonchar function of the first kind reduces to
\begin{equation*}
\GoncharA( d, d - 2m, q; R ) = \frac{W_s(\mathbb{S}^d)}{q} \left( R - 1 \right)^d R^{d-1} + R \, P( R ) + Q( R ) \, \log \frac{\left( R-1 \right)^2}{\left( R+1 \right)^2}
\end{equation*}
for some polynomials $P$ and $Q$.
A curious fact is that for odd dimension $d \geq 3$ and $(d-s)/2$ a positive integer (that is, $s = d - 2m$), the $s$-potential $U_{d-2m}^{\sigma_d}( R )$ is a linear combination of a \emph{complete elliptic integral of the first kind},
\begin{equation*}
\EllipticK( m ) \DEF \int_0^{\pi/2} \frac{\dd \theta}{\sqrt{1 - m \left( \sin \theta \right)^2}} = \frac{\pi}{2} \, \Hypergeom{2}{1}{1/2,1/2}{1}{m},
\end{equation*}
and a \emph{complete elliptic integral of the second kind},
\begin{equation*}
\EllipticE( m ) \DEF \int_0^{\pi/2} \sqrt{1 - m \left( \sin \theta \right)^2} \, \dd \theta = \frac{\pi}{2} \, \Hypergeom{2}{1}{-1/2,1/2}{1}{m},
\end{equation*}
with $m = 1 / R^2$ and with coefficients that are rational functions of $1/R^2$. \footnote{A similar relation holds if $0 < R < 1$. Then $m = R^2$.}
This follows by applying contiguous function relations for hypergeometric functions to \eqref{eq:s.potential.A} (cf. \cite[\S~15.5(ii)]{NIST:DLMF}). 
For example, for $d = 3$ and $s = d - 2 = 1$, one has
\begin{equation*}
U_1^{\sigma_d}( R ) = \frac{4}{3\pi} \, \frac{1 + R^2}{R^2} \, \EllipticE( R^2 ) - \frac{4}{3\pi} \, \frac{1 - R^2}{R^2} \, \EllipticK( R^2 ), \qquad 0 \leq R < 1.
\end{equation*}

If $(d-s)/2$ is not an integer, then the linear transformation \cite[Eq.s~15.8.4]{NIST:DLMF} applied to~\eqref{eq:s.potential} followed by the linear transformation \cite[last in Eq.~15.8.1]{NIST:DLMF} gives the following formula valid for all $R \in (0,1) \cup (1, \infty)$,
\begin{equation} \label{eq:s.potential.C}
\begin{split}
U_s^{\sigma_d}(R) 
&= W_s(\mathbb{S}^d) R^{1-d} \left( \frac{R+1}{2} \right)^{2d-s-2} \Hypergeom{2}{1}{1-d/2,1-d+s/2}{1-(d-s)/2}{\frac{\left( R-1 \right)^2}{\left( R+1 \right)^2}} \\
&\phantom{=}+ \frac{\gammafcn((d+1)/2) \gammafcn((s-d)/2)}{2 \sqrt{\pi} \gammafcn(s/2)} \left| R - 1 \right|^{d-s} R^{1-d} \left( \frac{R+1}{2} \right)^{d-2} \\
&\phantom{=\pm}\times \Hypergeom{2}{1}{1-d/2,1-s/2}{1+(d-s)/2}{\frac{\left( R-1 \right)^2}{\left( R+1 \right)^2}}.
\end{split}
\end{equation}
Both hypergeometric functions reduce to a polynomial if $d$ is an even positive integer. In this case the Gonchar function reduces to
\begin{equation*}
\GoncharA( d, s, q; R ) = \frac{W_s(\mathbb{S}^d)}{q} \left( R - 1 \right)^d R^{d-1} + W_s( \mathbb{S}^d ) \left( R + 1 \right)^{d-s} P( R ) + C(d,s) \left( R - 1 \right)^{d-s} Q(  R )
\end{equation*}
for some polynomials $P$ and $Q$.

\subsection*{Gonchar's Problem for Interior Sources} 
{\em A positive (unit) point charge is placed inside the insulated unit sphere with total charge $+1$. What is the smallest distance from the point charge to the sphere so that the support of the extremal measure associated with the external field due to this interior source is just the entire sphere?} 

The trivial solution is to put the field source at the center of the sphere. 
Then the signed $s$-equilibrium $\eta_Q$ on $\mathbb{S}^d$ and the $s$-extremal measure $\mu_Q$ on $\mathbb{S}^d$ associated with the external field $Q( \PT{x} ) = q / | \PT{x} |^s$ coincide with the $s$-equilibrium measure $\sigma_d$ on $\mathbb{S}^d$. 
We are interested in non-trivial solutions. 
 
First, we answer this question for the classical Newtonian case ($s=d-1$). The maximum principle for harmonic functions implies that the $s$-potential of the $s$-equilibrium measure~$\sigma_d$ is constant on $\mathbb{S}^d$ and this extends to the whole unit ball ({\em Faraday cage effect}); that is, $U_{d-1}^{\sigma_d}(\PT{a}) = W_{d-1}(\mathbb{S}^d) = 1$ for all $\PT{a} \in \mathbb{R}^{d+1}$ with $| \PT{a} | \leq 1$. Assuming $0 < R < 1$ and $q > 0$, by Proposition~\ref{prop}, $\supp(\mu_{Q}) = \mathbb{S}^d$ if and only if
\begin{equation} \label{fundamental.relation.smaller.1}
1 / q \geq \left( 1 + R \right) \left( 1 - R \right)^{-d} - 1.
\end{equation}
It follows that for a positive external field ($q > 0$) induced by an interior source ($0 < R < 1$) the critical distance $R_q = R_q(d)$ (to the center of $\mathbb{S}^d$) is a solution of the following algebraic equation
\begin{equation} \label{equation.interior}
\GoncharB( d, q; R ) \DEF \left[ 1 + ( 1 / q ) \right] \left( 1 - R \right)^d - R - 1 = 0.
\end{equation}
For $d=2$ (classical Coulomb case) the answer to Gonchar's problem is 
\begin{equation*}
1 - R_q = \frac{ \sqrt{ 9 + 8 / q } - 1 }{ 2 ( 1 + 1 / q ) },
\end{equation*}
which reduces to $( \sqrt{17} - 1 ) / 4$ for $q = 1$. This number seems to have no special meaning. \footnote{Trivia: The digit sequence of $2 - R_1 = ( 3 + \sqrt{17} ) / 4$ is sequence A188485 of Sloane's OEIS~\cite{OEIS2013}. One feature is the periodic continued fraction expansion $[\overline{1,1,3}]$.} 

The answer to Gonchar's problem for the external field of \eqref{externalfield} for general parameters $0<s<d$, $0 < R < 1$, and $q > 0$ relies on solving the characteristic equation (cf. \eqref{eqsigned})
\begin{equation} \label{eq:characteristic.equation.interior.positive}
\frac{W_s(\mathbb{S}^d)}{q} = \frac{\left( 1 + R \right)^{d-s}}{\left( 1 - R \right)^d} - U_s^{\sigma_d}( R ).
\end{equation}
Taking into account that $0 \leq R < 1$, we can rewrite \eqref{eq:characteristic.equation.interior.positive} as
\begin{equation}
\GoncharB( d, s, q; R ) = 0,
\end{equation}
where we define the \emph{Gonchar function of the second kind},
\begin{equation} \label{eq:Gonchar.function.2nd.kind}
\GoncharB( d, s, q; R ) \DEF \frac{W_s(\mathbb{S}^d)}{q} \left( 1 - R \right)^d - \left( 1 + R \right)^{d-s} + \left( 1 - R \right)^d U_s^{\sigma_d}( R ).
\end{equation}
Answering Gonchar's problem for $0 < s < d$, $q > 0$, and $0 < R < 1$ amounts to finding the unique (cf. remark after Proposition~\ref{prop}) (real) zero in $(0,1)$ of the function $\GoncharB( d, s, q; R )$.
Moreover, the density $\eta_{R,q,s}^\prime$ of Theorem~\ref{thm:signed.equilibrium.sphere} evaluated at the North Pole can be expressed as
\begin{equation*}
\eta_{R,q,s}^\prime( \PT{p} ) = \frac{q}{W_s( \mathbb{S}^d )} \, \frac{\GoncharB( d, s, q; R )}{\left( 1 - R \right)^{d}}, \qquad 0 < R < 1, q > 0, 0 < s < d.
\end{equation*}

The quadratic transformation formula for Gauss hypergeometric functions \cite[Eq.~15.3.17]{AbSt1992} applied to the formula of $U_s^{\sigma_d}( R )$ of \eqref{eq:s.potential} yields
\begin{equation} \label{eq:s.potential.B}
U_s^{\sigma_d}( R ) = \Hypergeom{2}{1}{-(d-1-s)/2,s/2}{(d+1)/2}{R^2}, \qquad 0 \leq R < 1.
\end{equation}
The hypergeometric function simplifies to $1$ if $s = d - 1$ (harmonic case) and reduces to a polynomial of degree $m$ in $R^2$ for $d - 1 - s = 2m$. Thus $R_q$ will be an algebraic number if $W_s(\mathbb{S}^d) / q$ is algebraic. 

\begin{thm} \label{thm:Gonchar.B}
For the external field $Q$ of \eqref{externalfield} with $0 < s < d$, $q > 0$, and $0 < R < 1$ the signed $s$-equilibrium is a positive measure on all of~$\mathbb{S}^d$ if and only if $0 \leq R \leq R_q$, where $R_q$ is the unique (real) zero in $(0,1)$ of the Gonchar function $\GoncharB(d,s,q;z)$. If $s = d - 1 - 2m$ for $m$ a non-negative integer, then $\GoncharB(d,d-1-2m,q;z)$ is a polynomial. 

In particular, the solution to Gonchar's problem is given by $\rho( d, s ) = 1 - R_q$ (and the trivial solution $1$). 
\end{thm}

We remark that in the case of $s \neq d - 1 - 2m$, $m$ a non-negative integer, one can use alternative representations of $U_s^{\sigma_d}( R )$ similar to those derived after Theorem~\ref{thm:Gonchar.A}.

\subsection*{Connecting Interior and Exterior External Fields}

The Riesz-$s$ external fields of the form \eqref{externalfield} induced by an interior ($0 < R^\prime < 1$, $q^\prime > 0$) and an exterior ($R > 1$, $q > 0$) point source giving rise to signed $s$-equilibria $\eta_{Q^\prime}$ and $\eta_Q$ on $\mathbb{S}^d$ with the same weighted $s$-potential,
\begin{equation*}
U_s^{\eta_{Q^\prime}}( \PT{x} ) + Q^\prime( \PT{x} ) = U_s^{\eta_{Q}}( \PT{x} ) + Q( \PT{x} ) \qquad \text{everywhere on $\mathbb{S}^d$,}
\end{equation*}
are connected by the following necessary and sufficient condition (cf. \eqref{eq:G.sphere.Q.s})
\begin{equation} \label{basic.relation}
q^\prime \, U_s^{\sigma_d}( R^\prime ) = q \, U_s^{\sigma_d}( R ). 
\end{equation}

One way to realize this condition is known as the {\em principle of inversion for $\mathbb{S}^d$}. For a fixed $s \in (0, d)$, the principle states that to a charge $q > 0$ at distance $R > 1$ from the center of the sphere there corresponds a charge $q^\prime = q R^{-s}$ at distance $R^\prime = 1 / R$ so that the $s$-equilibria $\eta_Q$ and $\eta_{Q^\prime}$ coincide and thus have the same weighted $s$-potential on $\mathbb{S}^d$. (Indeed, the hypergeometric function in \eqref{eq:s.potential} is invariant under inversion $R \mapsto 1 / R$. The adjustment of the charge follows from \eqref{basic.relation}. An inspection of \eqref{eta.Q} shows that $\eta_{Q^\prime} = \eta_Q$.) 

Now, if we require \eqref{basic.relation} but do not assume that $\eta_{Q^\prime} = \eta_Q$, then new phenomena emerge. 
The potential-theoretic regime (superharmonic, harmonic, subharmonic) determines for what ratios $q^\prime / q$, radii $R^\prime$ and $R$ exist so that \eqref{basic.relation} can be satisfied; cf. Figure~\ref{fig:typical.s.potentials}.
\begin{figure}[tb]
\begin{center}
\includegraphics[scale=.8]{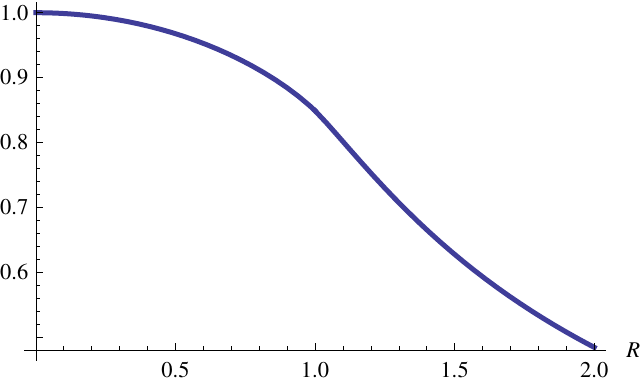}
\includegraphics[scale=.8]{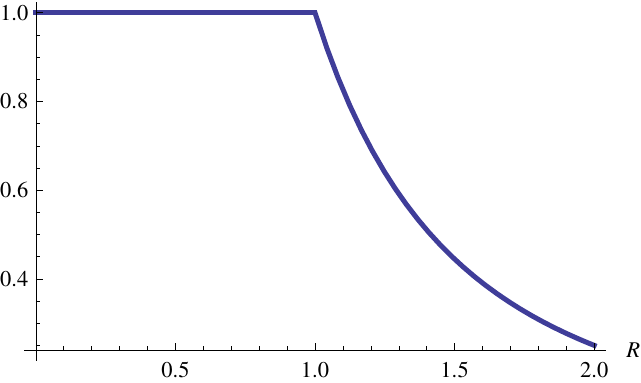}
\includegraphics[scale=.8]{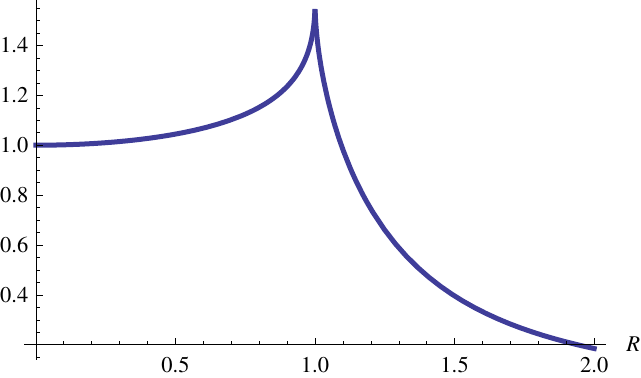}
\caption{\label{fig:typical.s.potentials} The typical behavior of the $s$-potential $U_s^{\sigma_d}( R )$ (for a formula see \eqref{eq:s.potential}) in the strictly superharmonic, harmonic, and strictly subharmonic case shown here for $d = 3$ and $s = (d - 1)/2$, $s = d - 1$, and $s = d - 1/2$.}
\end{center}
\end{figure}
Basic calculus\footnote{Using in \eqref{eq:s.potential.B} (if $0 \leq R < 1$) and \eqref{eq:s.potential} (if $R > 1$).} shows that in the strictly superharmonic case ($0 < s < d - 1$) the continuous $s$-potential $U_s^{\sigma_d}( R )$ is strictly monotonically decreasing on $(0,\infty)$. This implies that for any positive charges $q^\prime$ and $q$ with $0 < q^\prime / q < 1$, each $R^\prime \in (0,1)$ uniquely determines an $R > 1$ such that \eqref{basic.relation} holds. 
In the harmonic case ($s = d - 1$), \eqref{basic.relation} reduces to $q^\prime = q \, R^{1-d}$, since $U_{d-1}^{\sigma_d}( R^\prime ) = 1$ on $[0,1]$ (cf. \eqref{eq:s.potential.B}), each $R^\prime \in (0, 1)$ is mapped to $R = ( q^\prime / q )^{1 / ( 1 - d )}$ provided $q^\prime < q$. 
In the strictly subharmonic case ($d - 1 < s < d$) the $s$-potential $U_s^{\sigma_d}( R )$ is strictly monotonically increasing on $(0,1)$ and strictly monotonically decreasing on $(1,\infty)$. This implies that for any positive charges $q^\prime$ and $q$ with $0 < q^\prime / q < W_s( \mathbb{S}^d )$, each $R^\prime \in (0,1)$ uniquely determines an $R > 1$ such that \eqref{basic.relation} holds. 

As an example we consider the case $d = 2$ and $s = 3/2$ and assume that $q^\prime = q$. Then each $R^\prime \in (0, 1)$ determines a unique $R > 1$ satisfying $U_{3/2}^{\sigma_2}( R^\prime ) = U_{3/2}^{\sigma_2}( R )$, or equivalently, 
\begin{equation*}
\frac{2}{\sqrt{1+R^\prime} + \sqrt{1-R^\prime}} = \frac{\sqrt{R+1} - \sqrt{R-1}}{R}.
\end{equation*}
Curiously, a positive charge at the center of the sphere (that is, $R^\prime=0$ in above relation) would require that $R$ is a solution of the equation ${(\sqrt{R+1} - \sqrt{R-1}) / R = 1}$; that is, one can show that $R$ is the smaller of the two positive zeros of the minimal polynomial $x^4 - 4 x^3 + 4$.

\begin{rmk}
The equations \eqref{equation} and \eqref{equation.interior} characterizing the corresponding critical distance $R_q$ are related by the principle of inversion; that is, application of the transformations $R \mapsto 1 / R$ and $q \mapsto q R^{1-d}$ to either of the two reproduces the other one.
\end{rmk}

\subsection*{Negative External Fields}

A Gonchar type question can also be asked for Riesz external fields $Q$ of the form \eqref{externalfield} induced by a negative point source. The density $\eta_{R,q,s}^\prime$ of the signed equilibrium on $\mathbb{S}^d$ associated with $Q$, given in \eqref{eta.Q}, assumes its minimum value at the South Pole provided $R \in (0,1) \cup (1, \infty)$. Thus, Gonchar's problem concerns the distance from the North Pole such that $\eta_{R,q,s}^\prime$ is zero at the South Pole.
Answering this question for negative Riesz external fields is more subtle than for positive external fields. As discussed in the second remark after Proposition~\ref{prop}, a solution of the characteristic equation (cf. \eqref{neg.eqsigned})
\begin{equation} \label{eq:characteristic.equation.neg.charge}
\frac{W_s(\mathbb{S}^d)}{q} = f( R ) \DEF \frac{\left|R-1\right|^{d-s}}{\left(R+1\right)^d} - U_s^{\sigma_d}(R)
\end{equation}
exists for $q$ satisfying $W_s( \mathbb{S}^d ) / q \geq \min_{R \geq 0} f(R)$ (cf. Figure~\ref{fig:typical.f.R}). 

For exterior negative point sources ($R >1$, $0 < s < d$), we have $f( R ) > f( 1 ) = - W_s( \mathbb{S}^d )$ on $( 1, \infty )$. Thus, a critical distance $R_q$ (and therefore an answer to Gonchar's problem) exists if and only if $q < -1$. (If $R_q$ exists, then it is unique.) Defining the \emph{Gonchar function of the third kind},
\begin{equation} \label{eq:Gonchar.function.3rd.kind}
\GoncharC( d, s, q; R ) \DEF \left( \frac{W_s(\mathbb{S}^d)}{q} \left( R + 1 \right)^d - \left( R - 1 \right)^{d-s} \right) R^{d-1} + R^{d-1} \left( R + 1 \right)^d U_s^{\sigma_d}( R ),
\end{equation}
we can rewrite \eqref{eq:characteristic.equation.neg.charge} as
\begin{equation*}
\GoncharC( d, s, q; R ) = 0.
\end{equation*}
Given $q < -1$, the critical distance $R_q$ is the unique (real) zero of $\GoncharC( d, s, q; R )$ in $( 1, \infty )$. Moreover, the density $\eta_{R,q,s}^\prime$ of Theorem~\ref{thm:signed.equilibrium.sphere} evaluated at the South Pole can be expressed as
\begin{equation*}
\eta_{R,q,s}^\prime( -\PT{p} ) = \frac{q}{W_s( \mathbb{S}^d )} \, \frac{\GoncharC( d, s, q; R )}{R^{d-1} \left( R + 1 \right)^{d}}, \qquad R > 1, q < 0, 0 < s < d.
\end{equation*}

The following result generalizes \cite[Theorem~5]{BrDrSa2012} to exterior negative Riesz external fields.

\begin{thm} \label{thm:Gonchar.C}
For the external field $Q$ of \eqref{externalfield} with $0 < s < d$, $q < 0$, and $R > 1$ the signed $s$-equilibrium is a positive measure on all of~$\mathbb{S}^d$ if and only if one of the following conditions holds
\begin{enumerate}[\bf (i)]
\item $q \in [-1, 0)$ and $R > 1$. Gonchar's problem has no solution (no critical distance).
\item $q < -1$ and $R \geq R_q$, where $R_q$ is the unique (real) zero in $(1,\infty)$ of the Gonchar function $\GoncharC(d,s,q;z)$. The solution to Gonchar's problem is $R_q - 1$.
\end{enumerate}

If $s = d - 1 - 2m$ for $m$ a non-negative integer, then $\GoncharC(d,d-1-2m,q;z)$ is a polynomial.
\end{thm}

For an external field with negative point source inside the sphere ($0 < R < 1$), one needs to differentiate between the (i) strictly superharmonic ($0 < s < d - 1$), (ii) harmonic ($s = d - 1$), or (iii) strictly subharmonic ($d - 1 < s < d$) case; cf. Figure~\ref{fig:typical.f.R}. We define the \emph{Gonchar function of the fourth kind},
\begin{equation} \label{eq:Gonchar.function.4th.kind}
\GoncharD( d, s, q; R ) \DEF \frac{W_s(\mathbb{S}^d)}{q} \left( 1 + R \right)^d - \left( 1 - R \right)^{d-s} + \left( 1 + R \right)^d U_s^{\sigma_d}( R )
\end{equation}
and have
\begin{equation*}
\eta_{R,q,s}^\prime( -\PT{p} ) = \frac{q}{W_s( \mathbb{S}^d )} \, \frac{\GoncharD( d, s, q; R )}{\left( R + 1 \right)^{d}}, \qquad 0 < R < 1, q < 0, 0 < s < d.
\end{equation*}

In case (i), the function $f(R)$ in \eqref{eq:characteristic.equation.neg.charge} has a single minimum in $(0,1)$, say, at $R^*$ with $f(R^*) < - W_s(\mathbb{S}^d)$. Set $q^* \DEF W_s(\mathbb{S}^d) / f(R^*)$. Then \eqref{eq:characteristic.equation.neg.charge} has no solution in $(0,1)$ if ${q \in (q^*, 0)}$, two solutions in $(0,1)$ if $q \in ( -1, q^* )$ (as demonstrated for the special case $d = 4$ and $s = 1$ in \eqref{demonstration}) which degenerate to one as $q \to q^*$, and one solution in $(0,1)$ if $q \leq -1$.

\begin{thm} \label{thm:Gonchar.D.superharmonic}
For the external field $Q$ of \eqref{externalfield} with $0 < s < d-1$, $q < 0$, and $0 < R < 1$ the signed $s$-equilibrium is a positive measure on all of~$\mathbb{S}^d$ if and only if one of the following conditions holds
\begin{enumerate}[\bf (i)]
\item $q \in (q^*, 0)$ and $R \in (0, 1)$. Gonchar's problem has no solution as there exists no critical distance.
\item $q = q^*$ and $R \in (0, 1)$. The solution to Gonchar's problem is $1 - R^*$ as $\eta_{R,q,s}^\prime( -\PT{p} ) = 0$ at $R = R^*$ but $\eta_{R,q,s}^\prime( -\PT{p} ) > 0$ for $R \in (0, R^*) \cup (R^*, 1)$.
\item $q \in ( -1, q^* )$ and $R \in ( 0, R_{q,1} ) \cup ( R_{q,2}, 1 )$, where $R_{q,1}$ and $R_{q,2}$ are the two only (real) zeros in $(0,1)$ of the Gonchar function $\GoncharD(d,s,q;z)$. Gonchar's problem has two solutions $1 - R_{q,1}$ and $1 - R_{q,2}$.
\item $q \leq -1$ and $R \in (0, R_q]$, where $R_q$ is the unique (real) zero in $(0,1)$ of the Gonchar function $\GoncharD(d,s,q;z)$. The solution to Gonchar's problem is $1 - R_q$. 
\item $q < 0$ and $R = 0$ (trivial solution). Gonchar's problem has no solution.
\end{enumerate}

If $s = d - 1 - 2m$ for $m$ a non-negative integer, then $\GoncharD(d,d-1-2m,q;z)$ is a polynomial.
\end{thm}

In case (ii) the function $f(R)$ is strictly monotonically decreasing and convex on $(0,1)$ with $-1 = -W_{d-1}(\mathbb{S}^d) \leq f(R) \leq 0$. Hence, \eqref{eq:characteristic.equation.neg.charge} has no solution in $(0,1)$ if $q \in [-1, 0)$ and one solution in $(0,1)$ if $q < -1$.

\begin{thm} \label{thm:Gonchar.D.harmonic}
For the external field $Q$ of \eqref{externalfield} with $s = d-1$, $q < 0$, and $0 < R < 1$ the signed $s$-equilibrium is a positive measure on all of~$\mathbb{S}^d$ if and only if one of the following conditions holds
\begin{enumerate}[\bf (i)]
\item $q \in [-1, 0)$ and $R \in (0, 1)$. Gonchar's problem has no solution (no critical distance).
\item $q < -1$ and $R \in (0, R_q]$, where $R_q$ is the unique (real) zero in $(0,1)$ of the Gonchar function $\GoncharD(d,s,q;z)$. The solution to Gonchar's problem is $1 - R_q$.
\item $q < 0$ and $R = 0$ (trivial solution). Gonchar's problem has no solution.
\end{enumerate}

If $s = d - 1 - 2m$ for $m$ a non-negative integer, then $\GoncharD(d,d-1-2m,q;z)$ is a polynomial.
\end{thm}

In case (iii) the function $f(R)$ is strictly monotonically decreasing on $(0,1)$ with $-W_{s}(\mathbb{S}^d) \leq f(R) \leq 0$ like in case (ii) but neither convex nor concave on all of $(0,1)$, since $f^{\prime\prime}(0^+) > 0$ and $f^{\prime\prime}(R) \to - \infty$ as $R \to 1^-$. Equation~\eqref{eq:characteristic.equation.neg.charge} has a solution in $(0,1)$ if and only if $q < -1$.

\begin{thm} \label{thm:Gonchar.D.subharmonic}
For the external field $Q$ of \eqref{externalfield} with $d - 1 < s < d$, $q < 0$, and $0 < R < 1$ the signed $s$-equilibrium is a positive measure on all of~$\mathbb{S}^d$ if and only if one of the following conditions holds
\begin{enumerate}[\bf (i)]
\item $q \in [-1, 0)$ and $R \in (0, 1)$. Gonchar's problem has no solution (no critical distance).
\item $q < -1$ and $R \in (0, R_q]$, where $R_q$ is the unique (real) zero in $(0,1)$ of the Gonchar function $\GoncharD(d,s,q;z)$. The solution to Gonchar's problem is $1 - R_q$.
\item $q < 0$ and $R = 0$ (trivial solution). Gonchar's problem has no solution.
\end{enumerate}

If $s = d - 1 - 2m$ for $m$ a non-negative integer, then $\GoncharD(d,d-1-2m,q;z)$ is a polynomial.
\end{thm}

We remark that in the harmonic case, the analogues of the algebraic equations \eqref{equation} and~\eqref{equation.interior} characterizing the critical distance(s) either in $(0,1)$ or $(1,\infty)$ are given by 
\begin{align}
\GoncharC(d,q^\prime; R) \DEF \left[ \left( R + 1 \right)^d / q^\prime - R + 1 \right] R^{d-1} + \left( R + 1 \right)^d &= 0, \qquad R > 1, q^\prime < 0, \label{equation.exterior.neg} \\
\GoncharD(d,q^\prime; R) \DEF \left[ 1 + ( 1 / q^\prime ) \right] \left( 1 + R \right)^d + R - 1 &= 0, \qquad 0 \leq R < 1, q^\prime < 0. \label{equation.interior.neg}
\end{align}
Both equations are related by the principle of inversion: simultaneous application of the transformations $R \mapsto 1 / R$ and $q^\prime \mapsto q^\prime R^{1-d}$ changes one equation into the other.

\subsection*{Gonchar's Problem for the Logarithmic Potential}
The logarithmic potential $\log(1/r)$ follows from the Riesz $s$-potential $1/r^s$ ($s \neq 0$) by means of a limit process:
\begin{equation*}
\log(1/r) = \frac{\dd r^{-s}}{\dd s} \Big|_{s\to 0} = \lim_{s\to0} \frac{r^{-s}-1}{s}, \qquad r > 0.
\end{equation*}
This connection allows us to completely answer Gonchar's problem for the logarithmic potential and the logarithmic external field
\begin{equation} \label{eq:log.external.field}
Q_{\LOG,\PT{a},q}(\PT{x}) \DEF - q \log | \PT{x} - \PT{a} |, \qquad \PT{x} \in \mathbb{R}^{d+1},
\end{equation}
where $\PT{a} \not\in \mathbb{S}^d$. We remark that the logarithmic potential with external field in the plane is treated in \cite{SaTo1997} and \cite{BrDrSa2009,BrDrSa2014} deal with the logarithmic case on $\mathbb{S}^2$ which can be reduced to an external field problem in the plane using stereographic projection as demonstrated in \cite{Dr2007a,Si2005}. The general case for higher-dimensional spheres seems to have not been considered yet.

\begin{thm} \label{thm:signed.equilibrium.log.case}
The signed logarithmic equilibrium $\eta_{Q_{\LOG}}$ on $\mathbb{S}^d$ associated with the external field $Q_{\LOG} = Q_{\LOG,\PT{a},q}$ of \eqref{eq:log.external.field} is absolutely continuous with respect to the (normalized) surface area measure on $\mathbb{S}^d$; that is, $\dd \eta_{Q_{\LOG}}( \PT{x} ) = \eta_{\LOG}^\prime( \PT{x} ) \dd \sigma_d( \PT{x} )$, and its density is given by 
\begin{equation}
\eta_{\LOG}^\prime( \PT{x} ) = 1 + q - q \frac{\left| R^2 - 1 \right|^d}{\left| \PT{x} - \PT{a} \right|^{2d}}, \qquad \PT{x} \in \mathbb{S}^d.
\end{equation}
\end{thm}

\begin{proof}
Differentiating the signed equilibrium relation \eqref{signedeq} for $A = \mathbb{S}^d$ with respect to $s$, i.e.
\begin{equation*}
\int \frac{\dd}{\dd s} \left\{ \frac{\eta_{R,q,s}^\prime (\PT{y})}{\left| \PT{x} - \PT{y} \right|^s} \right\} \dd \sigma_d(\PT{y}) + \frac{\dd }{\dd s} \left\{ \frac{q}{\left| \PT{x} - \PT{a} \right|^s} \right\} = \frac{\dd U_s^{\eta_{Q}}(\PT{x})}{\dd s} + \frac{\dd Q(\PT{x})}{\dd s} = \frac{\dd G_{\mathbb{S}^d,Q,s}}{\dd s},
\end{equation*}
and letting $s$ go to zero yields
\begin{equation}
\begin{split} \label{eq:log.case.master.equ.A}
&\int \log \frac{1}{\left| \PT{x} - \PT{y} \right|} \left[ \eta_{R,q,s}^\prime (\PT{y}) \right]_{s\to0} \dd \sigma_d(\PT{y}) + q \log \frac{1}{\left| \PT{x} - \PT{a} \right|} \\
&\phantom{equalsequals}= \frac{\dd G_{\mathbb{S}^d,Q,s}}{\dd s} - \int \left[ \frac{\dd \eta_{R,q,s}^\prime (\PT{y})}{\dd s}\right]_{s\to0} \dd \sigma_d(\PT{y}).
\end{split}
\end{equation}
Using~\eqref{eq:W.s.S.d} and~\eqref{eq:s.potential}, one can easily verify that $\eta_{\LOG}^\prime( \PT{y} ) = \lim_{s\to0} \eta_{R,q,s}^\prime (\PT{y})$. As the right-hand side above does not depend on $\PT{x} \in \mathbb{S}^d$, the measure $\mu$ with $\dd \mu( \PT{y} ) = \eta_{\LOG}^\prime( \PT{y} ) \dd \sigma_d( \PT{y} )$ has constant weighted logarithmic potential on $\mathbb{S}^d$; that is, by the uniqueness of the signed equilibrium (cf. \cite[Lemma~23]{BrDrSa2009}), the measure $\mu$ is the signed logarithmic equilibrium on $\mathbb{S}^d$ associated with $Q_{\LOG} = Q_{\LOG,\PT{a},q}$.
\end{proof}

From \eqref{eq:log.case.master.equ.A} and \eqref{eq:G.sphere.Q.s} it follows that the weighted logarithmic potential of $\eta_{Q_{\LOG}}$ equals everywhere on $\mathbb{S}^d$ the constant
\begin{equation*}
G_{\mathbb{S}^d, Q_{\LOG}, \LOG} = W_{\LOG}(\mathbb{S}^d) + q U_{\LOG}^{\sigma_d}( \PT{a} ) - \int \left[ \frac{\dd \eta_{R,q,s}^\prime (\PT{y})}{\dd s}\right]_{s\to0} \dd \sigma_d(\PT{y}),
\end{equation*}
where the logarithmic energy of $\mathbb{S}^d$ is explicitly given by (cf., e.g., \cite[Eq.s~(2.24), (2.26)]{Br2008})
\begin{equation} \label{eq:log.energy.sphere}
W_{\LOG}(\mathbb{S}^d) \DEF \{ \mathcal{I}_{\LOG}[\mu] : \mu \in \mathcal{M}(\mathbb{S}^d) \} = - \log 2 + \left[ \digammafcn(d) - \digammafcn(d/2) \right] / 2
\end{equation}
(here $\digammafcn(s) \DEF \gammafcn^\prime(s) / \gammafcn(s)$ denotes the digamma function), and the logarithmic potential of $\sigma_d$ can be represented as
\begin{equation} \label{Wlog.logPotential}
U_{\LOG}^{\sigma_d}( \PT{a} ) = \log \frac{1}{R+1} + \frac{1}{2} \sum_{k=1}^\infty \frac{\Pochhsymb{d/2}{k}}{\Pochhsymb{d}{k} \, k} \frac{\left( 4 R \right)^k}{\left( R + 1 \right)^{2k}}.
\end{equation}
On the other hand, direct computation of the weighted logarithmic potential at the North Pole gives
\begin{equation*}
G_{\mathbb{S}^d, Q_{\LOG}, \LOG} = \left( 1 + q \right) W_{\LOG}(\mathbb{S}^d) - q \int \frac{\left| R^2 - 1 \right|^d}{\left| \PT{x} - \PT{a} \right|^{2d}} \log \frac{1}{\left| \PT{x} - \PT{p} \right|} \dd \sigma_d( \PT{x} ) + q \log \frac{1}{\left| R - 1 \right|}. 
\end{equation*}

The density $\eta_{\LOG}^\prime$ assumes its minimum value at the North (South) Pole if $q>0$ (${q<0}$), 
\begin{equation*}
\eta_{\LOG}^\prime( \PT{p} ) = 1 + q - q \, \frac{\left( R + 1 \right)^d}{\left| R - 1 \right|^d}, \qquad \eta_{\LOG}^\prime( -\PT{p} ) = 1 + q - q \, \frac{\left| R - 1 \right|^d}{\left( R + 1 \right)^d}.
\end{equation*}
From that one can obtain necessary and sufficient conditions for when $\supp(\mu_{Q_{\LOG}}) = \mathbb{S}^d$.

\begin{thm} \label{logthm}
Let $\PT{a} \in \mathbb{R}^{d+1}$ with $R = | \PT{a} | \neq 1$ and $q \neq 0$. Set $\alpha \DEF [ (1 + q) / q]^{1/d}$. Then $\supp(\mu_{Q_{\LOG}}) = \mathbb{S}^d$ (and thus $\mu_Q = \eta_Q$) if and only if 
\begin{enumerate}[\bf (i)]
\item for positive fields ($q>0$)
\begin{equation} \label{log.inequality1}
\frac{1+q}{q} \geq \frac{\left( R + 1 \right)^d}{\left| R - 1 \right|^d}, \quad \text{i.e.} \quad 
\begin{cases}
\displaystyle R \geq \frac{\alpha+1}{\alpha-1} & \text{for $R > 1$ and $q > 0$,} \\[1em]
\displaystyle R \leq \frac{\alpha-1} {\alpha+1} & \text{for $0 < R < 1$ and $q > 0$,}
\end{cases}
\end{equation}
\item or for negative fields ($q<0$)
\begin{equation} \label{log.inequality2}
\frac{1+q}{q} \leq \frac{\left| R - 1 \right|^d}{\left( R + 1 \right)^d}, \quad \text{i.e.} \quad 
\begin{cases}
\displaystyle R \geq \frac{1+\alpha}{1-\alpha} & \text{for $R > 1$ and $q < -1$,} \\[1em]
\displaystyle R \leq \frac{1-\alpha}{1+\alpha} & \text{for $0 < R < 1$ and $q < - 1$,}
\end{cases}
\end{equation}
whereas the first inequality is trivially satisfied for all $R\neq1$ if $-1 \leq q < 0$. 
\end{enumerate}
\end{thm}

\begin{rmk}
While the critical distance $R_{q,\LOG}$ is given by equality in above relations for $q<-1$ or $q>0$, in a weak negative logarithmic field ($-1 \leq q < 0$) there exists no critical distance. Observe that for fixed charge $q\neq0$, the sequence of critical distances $\{ R_{q,\LOG} \}_{d\geq2}$ goes to $+\infty$ ($0$) monotonically as $d \to \infty$ for $q>0$ ($q<0$).
\end{rmk}

\begin{proof}[Proof of Theorem~\ref{logthm}]
By differentiating the expressions in \eqref{eq:W.s.S.d} and \eqref{eq:s.potential} with respect to $s$ and letting $s\to0$, we arrive at \eqref{eq:log.energy.sphere} and \eqref{Wlog.logPotential}. The equivalence in Theorem~\ref{logthm} (and $\mu_{Q_{\LOG}} = \eta_{Q_{\LOG}}$) follows in a similar way as in the proof of Proposition~\ref{prop}. The relations in \eqref{log.inequality1} and \eqref{log.inequality2} are obtained by simple algebraic manipulations.
\end{proof}

\subsection*{Beyond Gonchar's Problem} 

This problem arises in a natural way when studying the external field problem on the $d$-sphere in the presence of a single positive point source above $\mathbb{S}^d$ exerting the external field.
The answer to Gonchar's problem pinpoints the critical distance $R_q$ of a charge $q > 0$ from the center of $\mathbb{S}^d$ such that the support of the $s$-extremal measure $\mu_Q$ associated with $Q$ is all of $\mathbb{S}^d$ for $R \geq R_q$ but $\supp( \mu_Q )$ is a proper subset of $\mathbb{S}^d$ for $1 < R < R_q$. Finding the $s$-extremal measure $\mu_Q$ when $\supp( \mu_Q ) \subsetneq \mathbb{S}^d$ turns out to be much more difficult. Given $d - 2 \leq s < d$, a convexity argument (cf. \cite[Theorem~10]{BrDrSa2009}) shows that $S_Q \DEF \supp( \mu_Q )$ is connected and forms a spherical cap centered at the pole opposite to the charge $q$; in fact, $S_Q$ minimizes the $\mathcal{F}_s$-functional\footnote{It is the Riesz analog of the Mhaskar-Saff functional from classical logarithmic potential theory in the plane (see \cite{MhSa1985} and \cite[Ch.~IV, p. 194]{SaTo1997}).} 
\begin{equation} \label{eq:F.s.functional}
\mathcal{F}_s( A ) \DEF W_s( A ) + \int Q \dd \mu_A, \qquad \text{$A \subset \mathbb{S}^d$ compact with $\CAP_s(A) > 0$,}
\end{equation}
where $W_s(A)$ is the $s$-energy of $A$ and $\mu_A$ is the $s$-extremal measure on $A$. Remarkably, if the signed $s$-equilibrium on a compact set $A \subset \mathbb{S}^d$ associated with $Q$ exists, then $\mathcal{F}_s(A) = G_{A,Q,s}$ (cf. \eqref{signedeq}). This connection to signed equilibria is exploited in \cite{BrDrSa2009} when determining the support $S_Q$ and, subsequently, the $s$-extremal measure $\mu_Q$ on $\mathbb{S}^d$ associated with $Q$ as the signed equilibrium $\eta_Q$ on a spherical cap $\Sigma_t \DEF \{ ( \sqrt{1 - u^2}  \overline{\PT{x}}, u) : -1 \leq u \leq t, \overline{\PT{x}} \in \mathbb{S}^{d-1} \}$ with critical 'size' $t = t_0$. We remark that either the  variational inequality \eqref{geqineq} would be violated on $\mathbb{S}^d \setminus \Sigma_t$ in case of a too small $\Sigma_t$ ($t < t_0$) or the density of $\eta_Q$ would be negative near the boundary of $\Sigma_t$ in case of a too large $\Sigma_t$ ($t > t_0$). 
Somewhat surprisingly it turns out that for $s = d - 2$ the signed equilibrium $\eta_Q$ on $\Sigma_t$ has a component that is uniformly distributed on the boundary of $\Sigma_t$ which vanishes if $t = t_0$. It should be noted that $\eta_Q$ on $\Sigma_t$ can be expressed in terms of the $s$-balayage measures (onto $\Sigma_t$) $\bal_s( \sigma_d, \Sigma_t )$ and $\bal_s( \delta_{\PT{a}}, \Sigma_t )$ by means of $\eta_{Q} = c \bal_s( \sigma_d, \Sigma_t ) - q \bal_s( \delta_{\PT{a}}, \Sigma_t )$, where $c$ is chosen such that $\eta_{Q}$ has total charge $1$.\footnote{Given a measure $\nu$ and a compact set $A \subset \mathbb{S}^d$, the $s$-balayage measure $\hat{\nu} \DEF \bal_s( \nu, A )$ preserves the Riesz $s$-potential of $\nu$ onto the set $A$ and diminishes it elsewhere (on the sphere $\mathbb{S}^d$).} In fact, the balayage method and the (restricted if $d - 2 \leq s < d - 1$) principle of domination play a crucial role in the derivation of these results (cf. \cite{BrDrSa2009}). Those techniques break 
down when $0 < s < d - 2$ ($d \geq 3$). For this gap new ideas are needed.

\subsection*{Padovan Sequence and the Plastic Number} The Padovan sequence $1$, $1$, $1$, $2$, $2$, $3$, $4$, $5$, $7$, $9$, $12$, $16$, $\dots$ (sequence A000931 in Sloane's OEIS~\cite{OEIS2013}) is named after architect Richard Padovan (cf.~\cite{St1996}). These numbers satisfy the recurrence relation $P_n = P_{n-2} + P_{n-3}$ with initial values $P_0 = P_1 = P_2 = 1$. The ratio $P_{n+1}/P_n$ of two consecutive Padovan numbers approximates the {\em Plastic constant $P$}\footnote{One origin of the name is Dutch: ``plastische getal''.} as $N \to \infty$:
\begin{equation} \label{Plastic.constant}
P \DEF  \lim_{n\to\infty} \frac{P_{n+1}}{P_n} = \frac{( 9 - \sqrt{69} )^{1/3} + ( 9 + \sqrt{69} )^{1/3}}{2^{1/3} 3^{2/3}} = 1.3247179572\dots.
\end{equation}
Padovan attributed \cite{Pa2002} its discovery to Dutch architect Hans van der Laan, who introduced in \cite{vdLa1960} the number $P$ as the ideal ratio of the geometric scale for spatial objects. 
The Plastic number is one of only two numbers $x$ for which there exist integers $k, m > 0$ such that $x+1=x^k$ and $x-1=x^{-m}$ (see \cite{AaFoKr2001}). The other number is the Golden ratio $\phi$. Figure~\ref{fig0a} shows one way to visualize the Padovan sequence as cuboid spirals, where the dimensions of each cuboid made up by the previous ones are given by three consecutive numbers in the sequence. Further discussion of this sequence appears in \cite{Ro2011}.
\begin{figure}[htb]
\includegraphics[scale=.4]{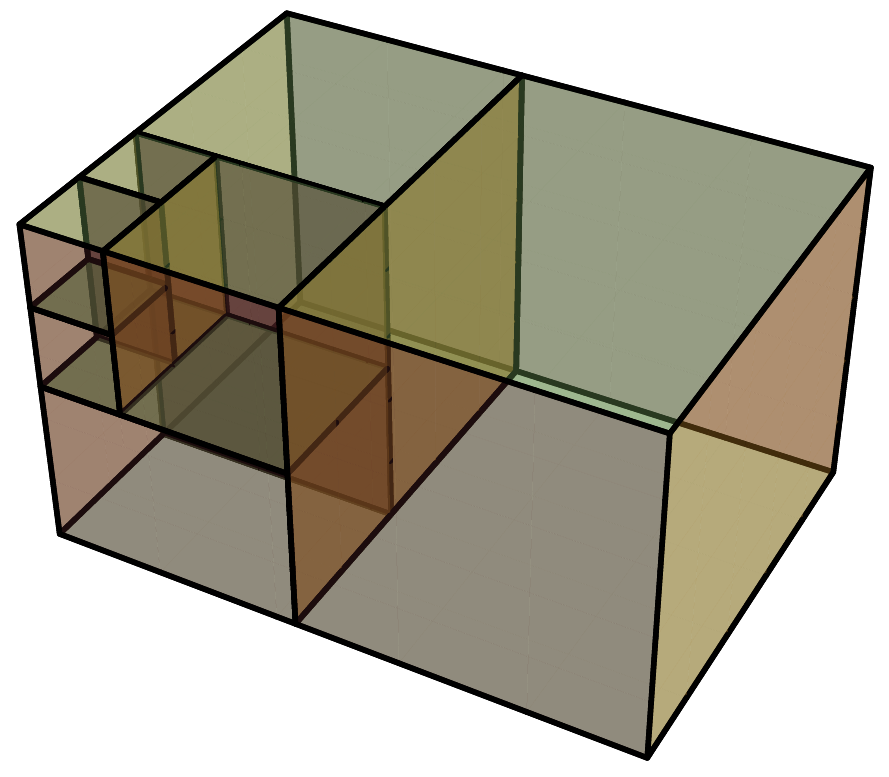}
\includegraphics[scale=.4]{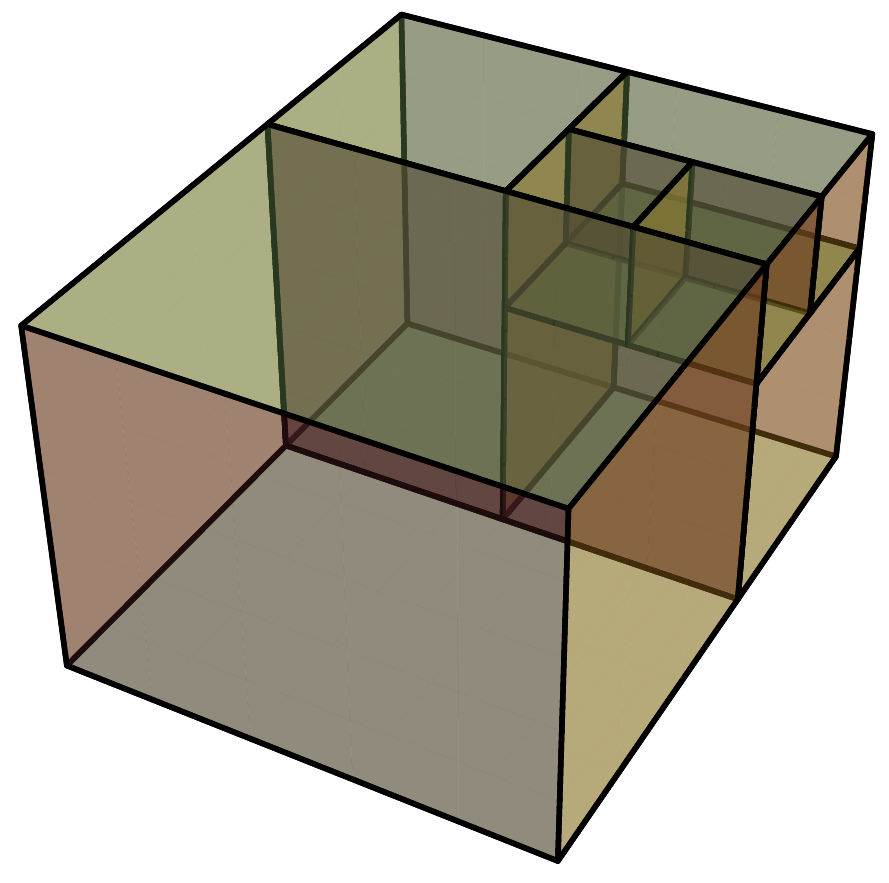}
\includegraphics[scale=.4]{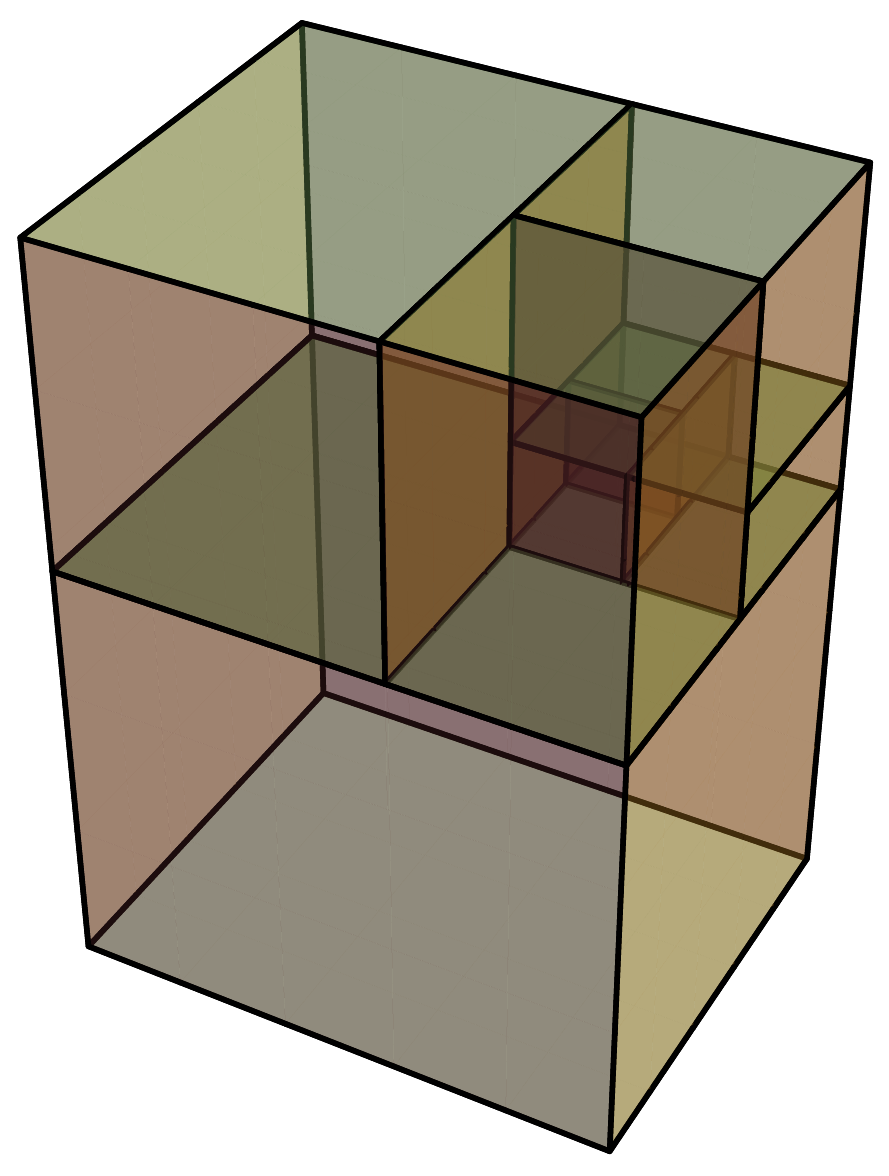}
\caption{\label{fig0a} Cuboid spiral visualizing the Padovan sequence.}
\end{figure}

\section{The Polynomials Arising from Gonchar's Problem}
\label{sec:polynomial.families}

The discussion of Gonchar's problem for positive/negative Riesz-$s$ external fields with interior/exterior point source leads to the introduction of four kinds of functions. These functions reduce to polynomials if the Riesz-$s$ parameter is given by $s = d - 1 - 2m$ for $m$ a non-negative integers. The last two displays in Figure~\ref{fig2} and Figure~\ref{fig2b} illustrate zero patterns of the four families of polynomials when $m = 1, 2$ which already indicate intriguing features. 
For example, one notices that the Gonchar polynomials of the second kind have an isolated group of $m$ zeros (indicated by $\bullet$ in Figures~\ref{fig2} and~\ref{fig2b}) inside of the left crescent-shaped region, whereas the other zeros are gathered near the right circle. Numerics indicate that the $m$ zeros coalesce at $-1$ as $d$ increases and the other zeros go to the right circle. In general, fixing $m$, the zeros of all four polynomials seem to approach the three circles and the two vertical line segments connecting the intersection points of the circles as $d$ increasing (cf. Figure~\ref{fig2} versus Figure~\ref{fig2b}). 

In the following we investigate the four families of Gonchar polynomials $\GoncharA(d,q; z)$, $\GoncharB(d,q; z)$, $\GoncharC(d,q^\prime; z)$, and $\GoncharD(d,q^\prime; z)$ given in \eqref{equation}, \eqref{equation.interior}, \eqref{equation.exterior.neg}, and \eqref{equation.interior.neg}. We shall assume that $q + q^\prime = 0$ and $q > 0$. Aside from the solution to Gonchar's problem, these polynomials are interesting in themselves and their distinctive properties merit further studies. We studied the family of polynomials $\GoncharA(d,q; z)$ for $q = 1$ in \cite{BrDrSa2012}. Two of the conjectures raised there have been answered in \cite{La2013}.
\begin{figure}[ht]
\begin{center}
\begin{minipage}{0.475\linewidth}
\centerline{\includegraphics[scale=.825]{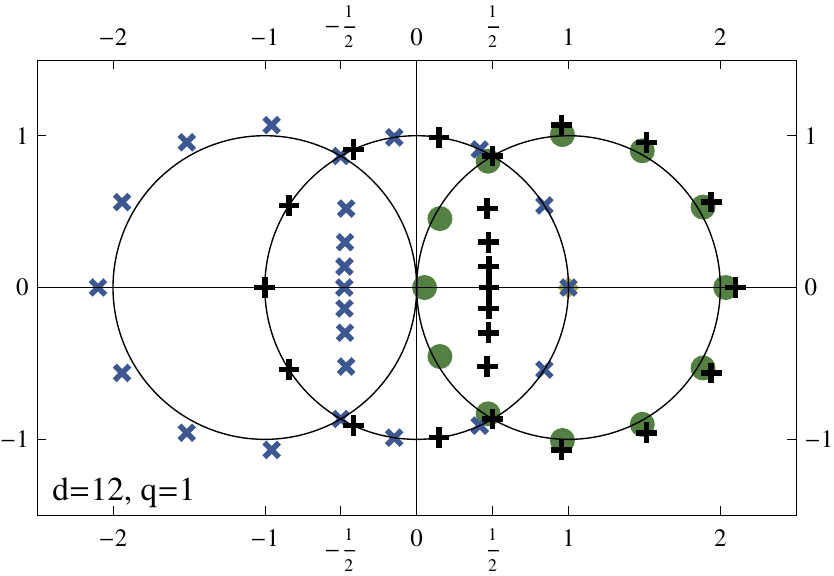}}
\end{minipage}
\begin{minipage}{0.475\linewidth}
\centerline{\includegraphics[scale=.825]{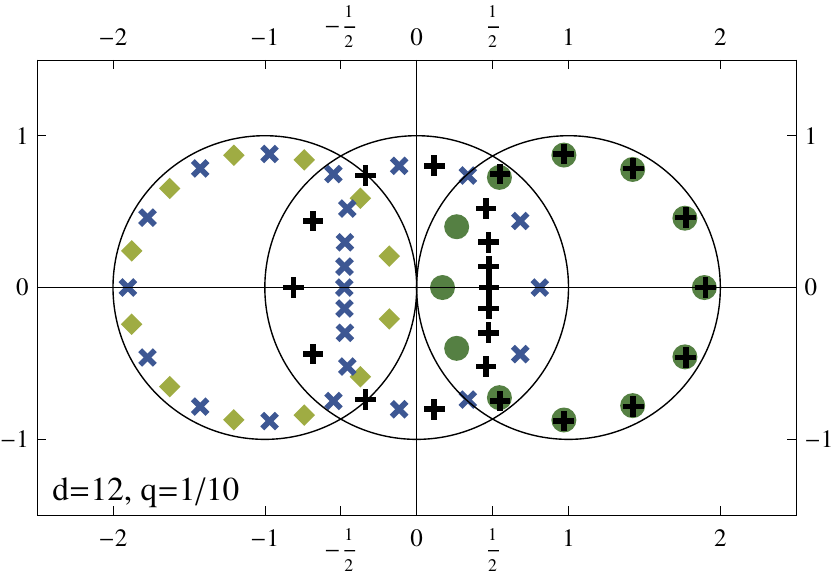}}
\end{minipage}
\begin{minipage}{0.475\linewidth}
\centerline{\includegraphics[scale=.825]{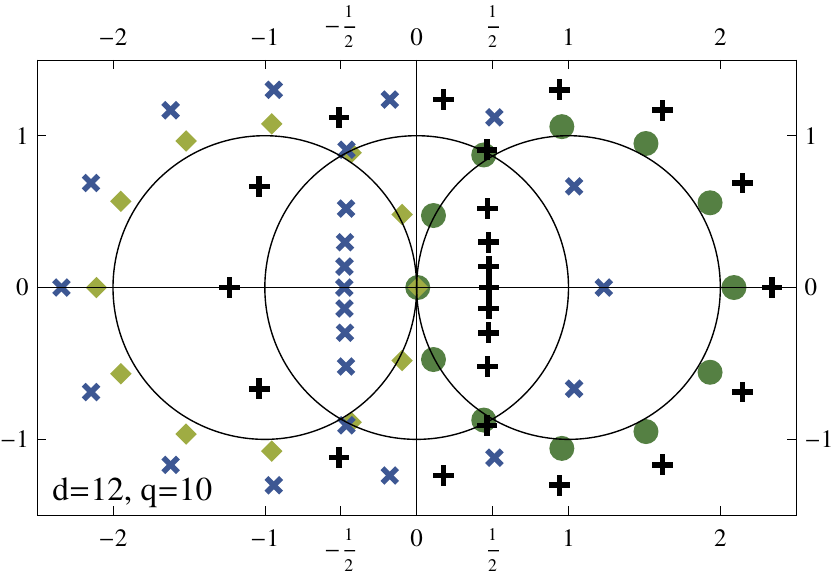}}
\end{minipage}
\begin{minipage}{0.475\linewidth}
\small
\begin{align*}
\left[ \left( z - 1 \right)^d / q - z - 1 \right] z^{d-1} + \left( z - 1 \right)^d &= 0 \;\; (+) \\
\left[ 1 + ( 1 / q ) \right] \left( 1 - z \right)^d - z - 1 &= 0 \;\; (\bullet) \\
\left[ \left( z + 1 \right)^d / q^\prime - z + 1 \right] z^{d-1} + \left( z + 1 \right)^d &= 0 \;\; (\times) \\
\left[ 1 + ( 1 / q^\prime ) \right] \left( 1 + z \right)^d + z - 1 &= 0 \;\; (\blacklozenge) \\
\end{align*}
\end{minipage}
\begin{minipage}{0.475\linewidth}
\centerline{\includegraphics[scale=.825]{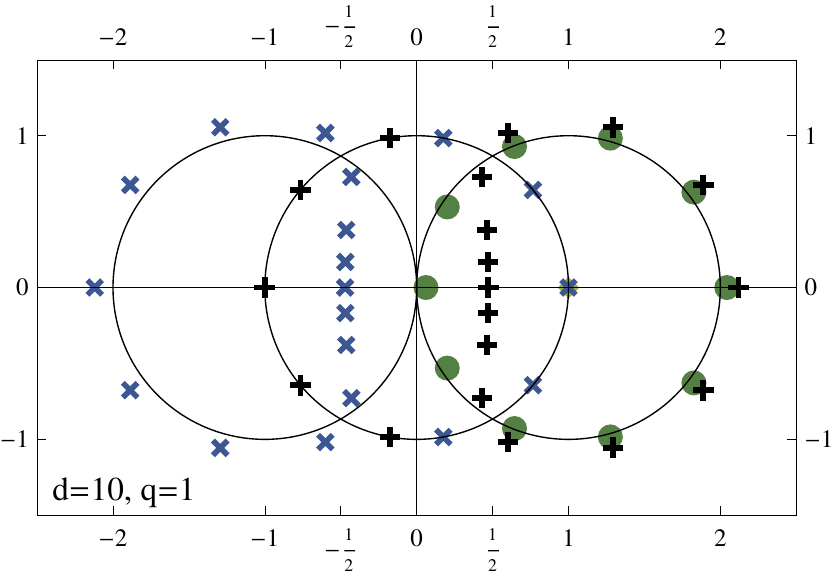}}
\end{minipage}
\begin{minipage}{0.475\linewidth}
\centerline{\includegraphics[scale=.825]{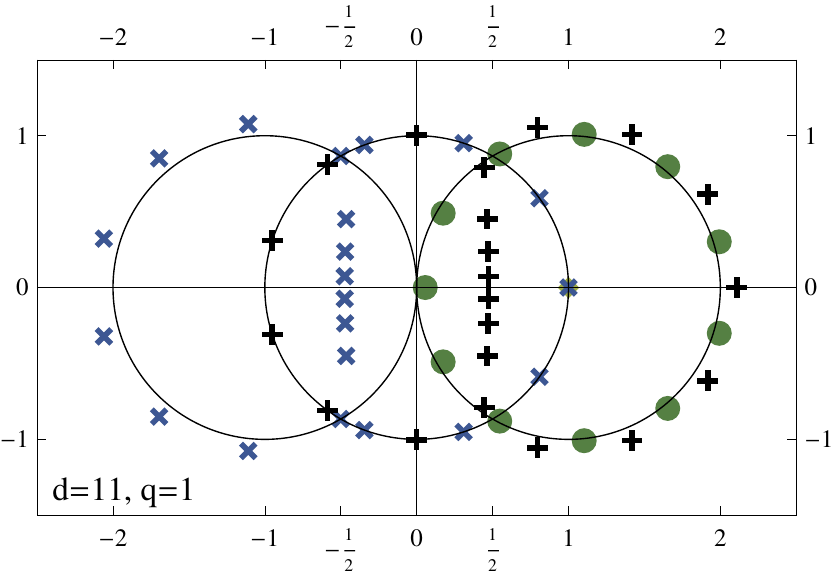}}
\end{minipage}
\begin{minipage}{0.475\linewidth}
\centerline{\includegraphics[scale=.75]{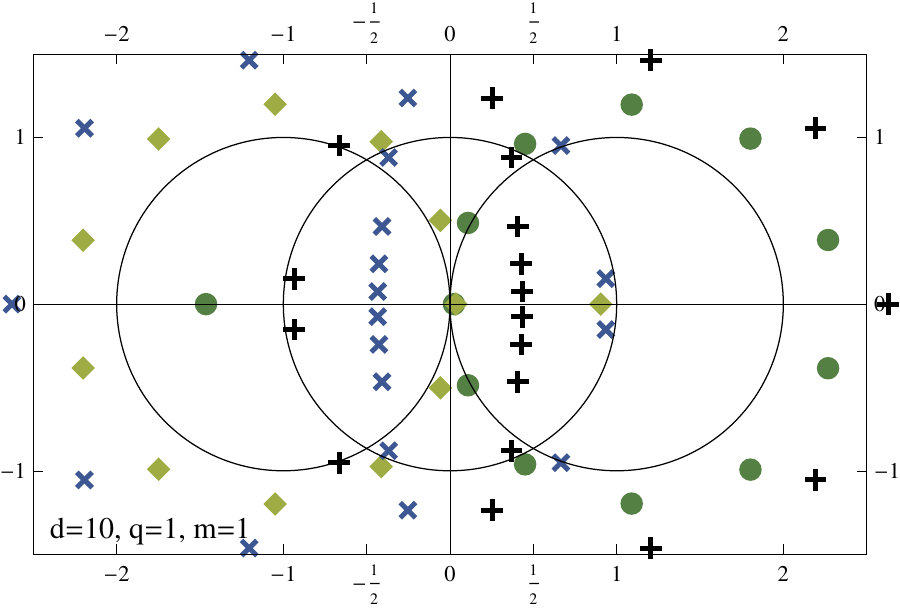}}
\end{minipage}
\begin{minipage}{0.475\linewidth}
\centerline{\includegraphics[scale=.75]{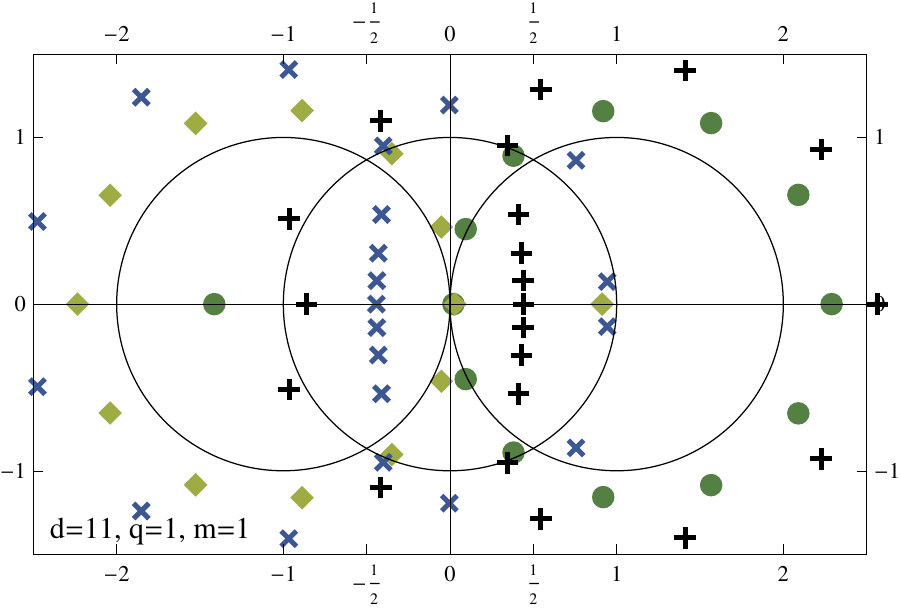}}
\end{minipage}
\end{center}
\caption{\label{fig2} Roots of the given polynomial equations for selected $d$ and $q$ ($q^\prime = - q$). The bottom row shows the roots for the case $s = d - 3$.} 
\end{figure}
\afterpage{\clearpage}
\begin{figure}[ht]
\begin{center}
\includegraphics[scale=0.80]{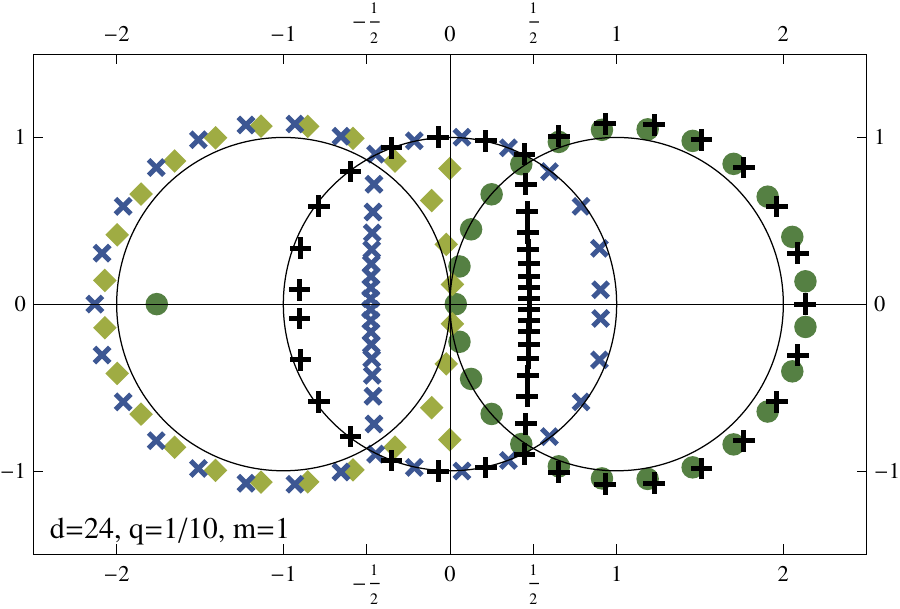} \includegraphics[scale=0.80]{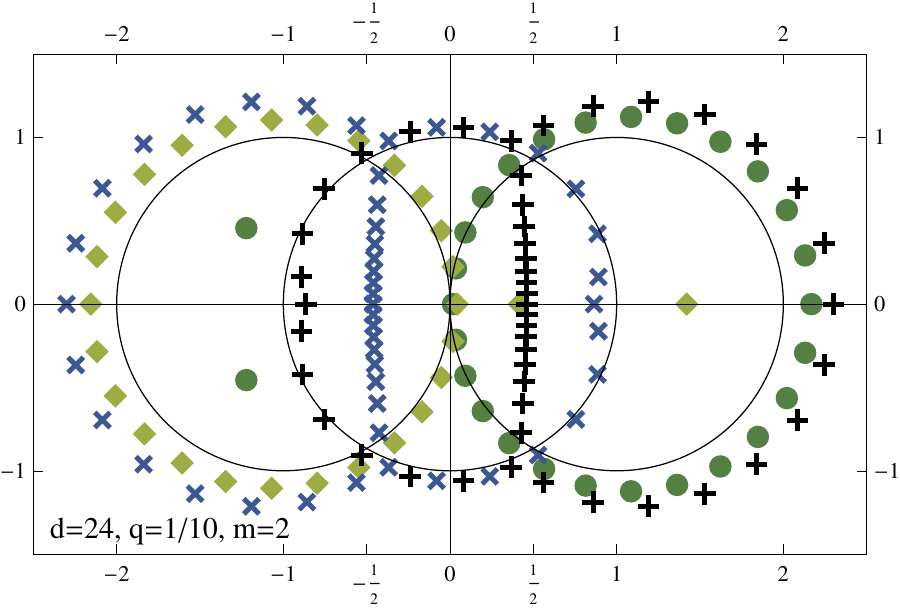} \\
\includegraphics[scale=0.80]{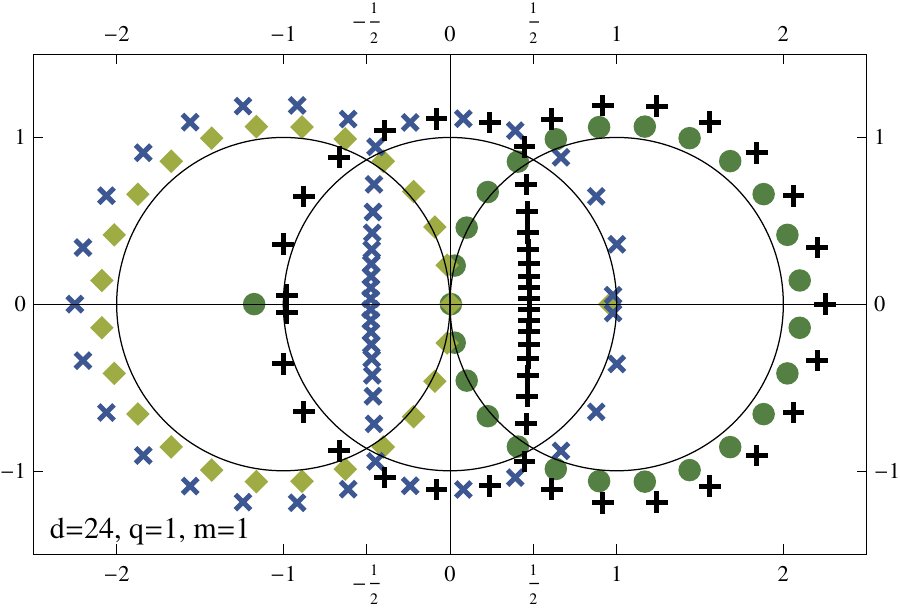} \includegraphics[scale=0.80]{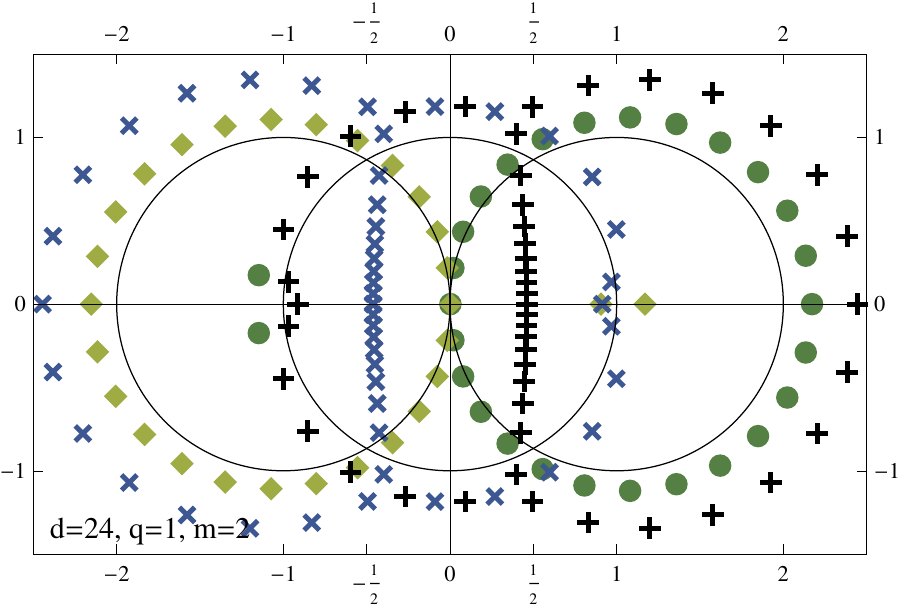} \\
\includegraphics[scale=0.80]{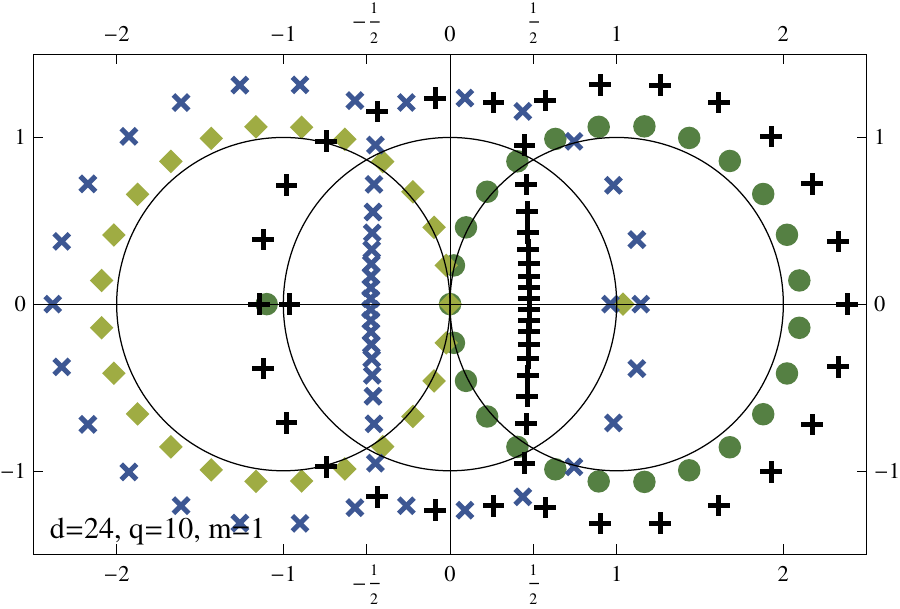} \includegraphics[scale=0.80]{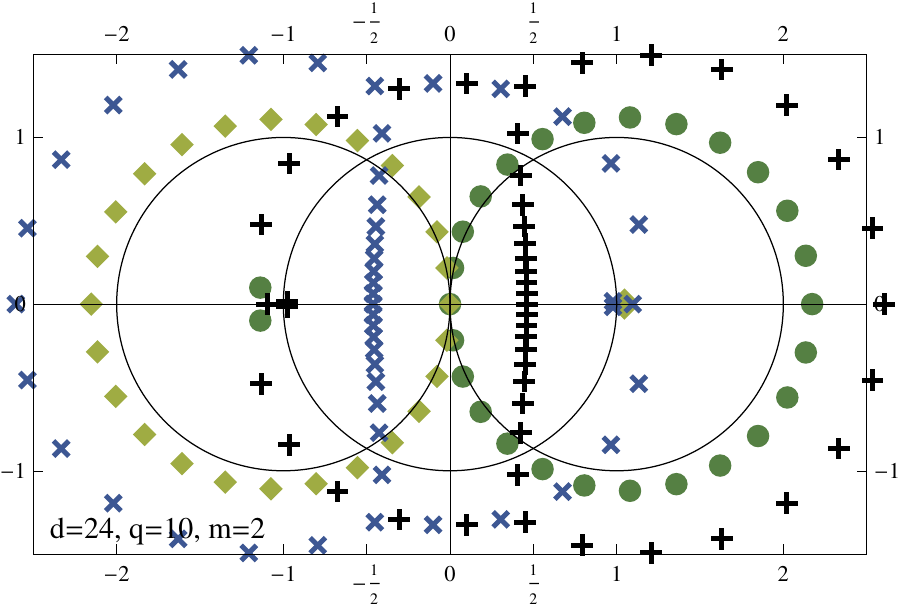}
\end{center}
\caption{\label{fig2b} Zeros of the polynomials analogue to the one given in Figure~\ref{fig2} but for the case $s = d - 1 - 2m$; that is, $s = d - 3$ (left column) and $s = d - 4$ (right column).}
\end{figure}

\subsection*{Interrelations and Self-Reciprocity}

The four kinds of Gonchar polynomials are connected. 
The Gonchar polynomials of the first and second kind are related through
\begin{align}
z^d \, \GoncharA( d, q \left( 1 / z \right)^{d-1}; 1 / z ) &= \GoncharB( d, q; z ), \\
z^d \, \GoncharB( d, q \left( 1 / z \right)^{d-1}; 1 / z ) &= \GoncharA( d, q; z ).
\end{align}
Similarly, the Gonchar polynomials of the third and fourth kind are related through
\begin{align}
z^d \, \GoncharC( d, q^\prime \left( 1 / z \right)^{d-1}; 1 / z ) &= \GoncharD( d, q^\prime; z ), \\
z^d \, \GoncharD( d, q^\prime \left( 1 / z \right)^{d-1}; 1 / z ) &= \GoncharC( d, q^\prime; z ).
\end{align}
The Gonchar polynomials of the second and fourth kind have in common that for $q \neq 1$, 
\begin{equation}
\GoncharB( d, -q; -z ) = \left( 1 + \frac{1}{-q} \right) \left( 1 + z \right)^d + z - 1 = \GoncharD( d, q^\prime; z ).
\end{equation}
For even dimension $d$ and canonical charges $q = 1$ and $q^\prime = -1$, the Gonchar polynomials of first and third kind are connected by means of
\begin{equation} \label{eq:correspondence.A.C}
\left( - z \right)^{2d-1} \GoncharA( d, 1; - 1 / z ) = ( - 1 )^{d-1} \left[ \left( 1 + z \right)^d + z - 1 \right] z^{d-1} + \left( 1 + z \right)^{d} = \GoncharC( d, - 1; z ),
\end{equation}
whereas for odd dimension $d$ and conical charges one has
\begin{equation}
\left( - z \right)^{2d-1} \GoncharA( d, 1; - 1 / z ) + \GoncharC( d, - 1; z ) = 2 \left( 1 + z \right)^{d}.
\end{equation}

A polynomial $P$ with real coefficients is called \emph{self-reciprocal} if its \emph{reciprocal polynomial $P^*( z ) \DEF z^{\deg P} \, P( 1/z )$} coincides with $P$ and it is called \emph{reciprocal} if $P^*( z ) = \pm P( z )$. That means, that the coefficients of $z^k$ and of $z^{\deg P - k}$ in $P(z)$ are the same. The polynomial $\GoncharA( d, 1; z)$ is self-reciprocal for even $d$, since
\begin{equation}
z^{2d-1} \, \GoncharA( d, 1; 1/z ) = \left[ \left( 1 - z \right)^d - z - 1 \right] z^{d-1} + \left( 1 - z \right)^{d}.
\end{equation}
The Gonchar polynomial of third kind is reciprocal for every dimension $d$; i.e.,
\begin{equation}
z^{2d-1} \, \GoncharC( d, -1; 1/z ) = \frac{1}{-1} \left( 1 + z \right)^{d} - z^{d-1} + z^d + z^d \left( 1 + z \right)^{d} = - \GoncharC( d, -1; z ).
\end{equation}
Consequently, if $\zeta$ is a zero of $\GoncharC( d, -1; z)$, then so is $1/\zeta$ for any $d$.

\subsection*{The Gonchar Polynomials of the First Kind} 

We refer the interested reader to~\cite{BrDrSa2012}.

\subsection*{The Gonchar Polynomials of the Third Kind}

The correspondence \eqref{eq:correspondence.A.C} implies that for even dimension $d$, the number $\zeta$ is a zero of $\GoncharC( d, -1; z )$ if and only if $-1/\zeta$ is a zero of $\GoncharA( d, 1; z )$.

\begin{conj}
Let $\Gamma$ be the set consisting of the boundary of the union of the two unit disks centered at $-1$ and $0$ and the line-segment connecting the intersection points 
Then, as $d \to \infty$, all the zeros of $\GoncharC( d, q^\prime; z)$ tend to $\Gamma$, and every point of $\Gamma$ attracts zeros of these polynomials.
\end{conj}

\subsection*{The Gonchar Polynomials of the Second and Fourth Kind} 

Both polynomials can be derived from the trinomial
\begin{equation} \label{eq:basic.polynomial}
P_d( q; w ) \DEF \left( 1 + \frac{1}{q} \right) w^d +  w - 2
\end{equation}
by means of a linear transformation of the argument; that is,
\begin{equation*}
\GoncharB( d, q; z ) = P_d( q; 1 - z ), \qquad \GoncharD( d, q^\prime; z ) = P_d( q^\prime; 1 + z ).
\end{equation*}
Applying results of the Hungarian mathematician Egerv{\'a}ry on the distribution of zeros of trinomials (delightfully summarized in~\cite{Sz2010} and otherwise difficult to come by in the English speaking literature), we derive several results for the Gonchar polynomials of the second and fourth kind. An interesting observation central to Egerv{\'a}ry's work is that \emph{the zeros of a trinomial polynomial can be characterized as the equilibrium points of an external field problem in the complex plane of unit point charges that are located at the vertices of two regular concentric polygons centered at the origin.} 

\begin{prop}
Let $d \geq 1$ and $q \neq -1$. Then $P_d( q; w )$ has simple zeros.
\end{prop}

\begin{proof}
This is clear for $d = 1$. Let $d \geq 2$. Then, according to \cite[Comment after Theorem~1]{Sz2010}, the polynomial equation $A z^{n+m} + B z^{m} + C = 0$ has a root with higher multiplicity if and only if
\begin{equation*}
(-1)^{n+m} \frac{A^m C^n}{B^{n+m}} = \frac{m^m n^n}{\left( n + m \right)^{n+m}}.
\end{equation*}
Taking $n = d - 1$, $m = 1$, $A = 1 + 1/q$, $B = 1$, and $C = -2$, it follows that the left-hand side is negative and the right-hand side positive. Hence $P_d( q; w)$ has no zero with higher multiplicity.
\end{proof}

\begin{cor}
All the zeros of the Gonchar polynomials of the second and fourth kind are simple.
\end{cor}

\begin{prop} \label{prop:force.field.interpretation}
Let $d \geq 2$ and $q \neq -1$. Assuming that the force is inversely proportional to the distance, the zeros of $P_d( q; w )$ are the equilibrium points of the force field of (positive) unit point charges at the vertices of two concentric polygons in the complex plane. The vertices are given by
\begin{align*}
\zeta_{1,k} &\DEF \left[ \left( 2 - \frac{1}{d} \right) \frac{1}{\big| 1 + \frac{1}{q} \big|} \right]^{1/(d-1)} e^{i \, (-\alpha + (2k + 1) \pi ) / (d-1)}, \qquad k = 1, \dots, d - 1,
\intertext{and}
\zeta_{2,k} &\DEF \left[ \left( 2 + \frac{1}{d-1} \right) \frac{2}{\big| 1 + \frac{1}{q} \big|} \right]^{1/d} e^{i \, (\pi -\alpha + (2k + 1) \pi ) / d}, \qquad k = 1, \dots, d,
\end{align*}
where $\alpha = 0$ if $q \in ( -\infty, -1) \cup ( 0, \infty )$ and $\alpha = \pi$ if $q \in (-1,0)$.
\end{prop}

\begin{proof}
This follows from \cite[Theorem~2]{Sz2010}.
\end{proof}

\begin{cor}
Let $d \geq 2$. If $q > 0$, then the zeros of $\GoncharB( d, q; z )$ are the equilibrium points of the force field of positive unit point charges at the vertices $1 - \zeta_{1,k}$ ($k = 1, \dots, d - 1$) and $1 - \zeta_{2,k}$ ($k = 1, \dots, d$). If $q^\prime < 0$ ($q^\prime \neq -1$), then the zeros of $\GoncharD( d, q^\prime; z )$ are the equilibrium points of the force field of positive unit point charges at the vertices $\zeta_{1,k} - 1$ ($k = 1, \dots, d - 1$) and $\zeta_{2,k} - 1$ ($k = 1, \dots, d$). 
\end{cor}

The paper \cite{Sz2010} and the more recent \cite{Me2012} discuss annular sectors as zero inclusion regions. 

\begin{prop} \label{prop:zero.inclusion.regions}
Let $d \geq 3$. Suppose $s$ and $t$ are the unique positive roots of $x^d - \frac{1}{2} \, x - 1 = 0$ and $x^d + \frac{1}{2} \, x - 1 = 0$, respectively. Choose radii $\rho_1$ and $\rho_2$ with $2/3 \leq \rho_1 \leq t < 1 < s \leq \rho_2 \leq (3/2)^{1/(d-1)}$. Then each of the annular sectors
\begin{equation} \label{eq:annular.sectors}
\left\{ z \in \mathbb{C} : \rho_1 \leq | z | \leq \rho_2, \left| \arg( z ) - \frac{2\pi k}{d} \right| \leq \frac{(3/2)^{1/(d-1)}}{2d} \right\}, \qquad k = 0, \dots, d-1,
\end{equation}
contains exactly one of the $d$ zeros of $P_d( 1; w ) = 2 w^d + w - 2$.
\end{prop}

\begin{proof}
The substitution $w = ( -C / A )^{1/d} \zeta$ for any $d$-th root of $-C/A$ transforms the equation $A w^d + B w + C = 0$ into $\zeta^d - a \zeta - 1 = 0$, where $a = ( B / C ) ( -C / A )^{1/d}$. Hence $P_d( q; w ) = 0$ if and only if 
\begin{equation} \label{eq:Me.transformation}
\zeta^d - a \zeta - 1 = 0, \qquad \text{where $a = - \frac{1}{2} \left( \frac{2}{1+\frac{1}{q}} \right)^{1/d}$ and $w = \left( \frac{2}{1+\frac{1}{q}} \right)^{1/d} \zeta$.}
\end{equation}
Let $s$ and $t$, with $0 < t < 1 < s$, be the unique positive roots of $x^d - | a | \, x - 1 = 0$ and $x^d + | a | \, x - 1 = 0$, respectively. Then, by \cite[Thm.~3.1]{Me2012}, every root of \eqref{eq:Me.transformation} lies in the annulus $\{ \zeta \in \mathbb{C} : t \leq | \zeta | \leq s \}$. It can be readily seen that $| a | ( 1 + | a | )^{1/(d-1)} < 1$ for $d \geq 3$ and $q > 0$. Let $d \geq 3$ and $q > 0$. Then each of the disjunct sectors 
\begin{equation*}
\left\{ \zeta \in \mathbb{C} : \left| \arg( \zeta ) - \frac{2\pi k}{d} \right| \leq \frac{\theta}{d} \right\}, \qquad  k = 0, \dots, d-1,
\end{equation*}
where $\sin \theta = | a | ( 1 + | a | )^{1/(d-1)}$, contains exactly one root of \eqref{eq:annular.sectors} by~\cite[Thm.~5.3]{Me2012}. For $q = 1$, the coefficient $a$ reduces to $-1/2$ and both $P_d( 1; w ) = 2 w^d + w - 2$ and $\zeta^d - a \zeta - 1$ have the same zeros. The result follows. The bounds for $\rho_1$ and $\rho_2$ follow from~\cite[Lem.~2.6]{Me2012}. 
\end{proof}

Proposition~\ref{prop:zero.inclusion.regions} can be obtained for $P_d( q; w )$ with general $q$. Figure~\ref{fig:zeroinclusionregions} illustrates the force field setting and the zero inclusion regions for the canonical case $q = 1$ (and $d = 6$). 
\begin{figure}[htb]
\begin{center}
\includegraphics[scale=1]{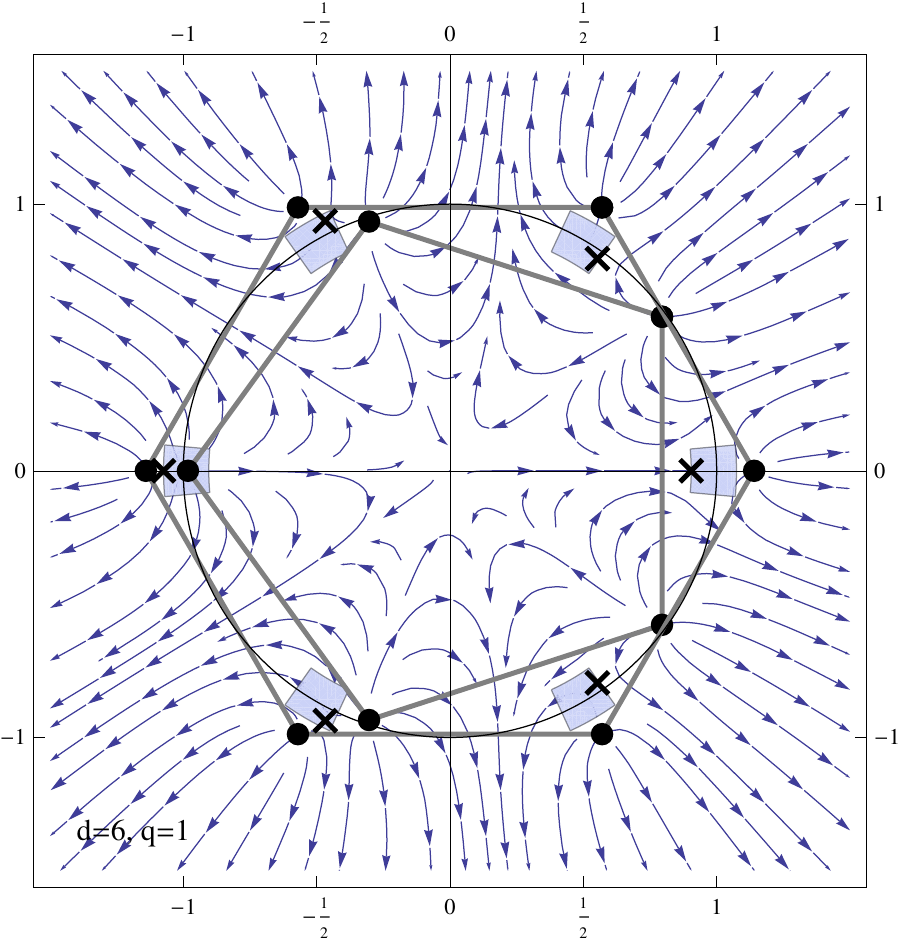}
\caption{\label{fig:zeroinclusionregions} The zeros ($\boldsymbol{\times}$) of $P_6(1; w) = 2 w^6 + w - 2$, lying in the annular sectors of \eqref{eq:annular.sectors} with $\rho_1 = t$ and $\rho_2 = s$, are equilibrium points of the force field due to the vertices of the concentric pentagon and hexagon (cf. Prop.~\ref{prop:force.field.interpretation}). The resultant of the force (stream plot) is $\sum_{k=1}^{5} (\overline{z} - \overline{\zeta_{1,k}})^{-1} + \sum_{k=1}^{6} (\overline{z} - \overline{\zeta_{2,k}})^{-1}$.}
\end{center}
\end{figure}

\pagebreak

Theorem~\ref{thm:Gonchar.B} and Theorem~\ref{thm:Gonchar.D.harmonic}, respectively, imply the following properties of the Gonchar polynomials of the second and fourth kind.

\begin{prop} \label{prop:Gonchar.B.real.zero}
If $q > 0$, then $\GoncharB( d, q; z )$ has a unique real zero in the interval $(0,1)$.
\end{prop}

\begin{prop} \label{prop:Gonchar.D.real.zero}
If $q^\prime \in (-1,0)$, then $\GoncharD( d, q^\prime; z )$ has no zero in $( 0, 1 )$. 
If~$q^\prime = -1$, then $\GoncharD( d, q^\prime; z ) = z - 1$. If $q^\prime < -1$, then $\GoncharD( d, q^\prime; z )$ has a unique real zero in the interval~$(0,1)$.
\end{prop}

From the fact that $P_d(q; \zeta) = 0$ implies 
\begin{equation*}
\left| 2 - \zeta \right| = \left| 1 + \frac{1}{q} \right| \left| \zeta \right|^d, 
\end{equation*}
where the right-hand side is changing exponentially fast as $d \to \infty$ when the zeros avoid an $\varepsilon$-neighborhood of the unit circle, we get that the zeros of $P_d(q; w)$ approach the unit circle as $d \to \infty$. This in turn implies that the zeros of $\GoncharB( d, q; z )$ (for $q > 0$) approach the circle $\Gamma_1 \DEF \{ z \in \mathbb{C} : | z - 1 | = 1 \}$ and the zeros of $\GoncharD( d, q^\prime; z )$ (for $q^\prime < 0$ and $q^\prime \neq -1$) approach the circle $\Gamma_{-1} \DEF \{ z \in \mathbb{C} : | z + 1 | = 1 \}$ as $d \to \infty$.


\section{Negatively Charged External Fields -- Signed Equilibrium on Spherical Caps}
\label{sec:negatively.charged.external field}

We are interested in the external field due to a negative charge below the South Pole,
\begin{equation} \label{eq:neg.external.field}
Q_{\PT{b},s}( \PT{x} ) = \frac{q}{| \PT{x} - \PT{b} |^{s}}, \quad \PT{x} \in \mathbb{S}^d, \qquad \PT{b} = - R \PT{p}, \quad R > 1, \quad q < 0, 
\end{equation}
that is sufficiently strong to give rise to an $s$-extremal measure on $\mathbb{S}^d$ that is {\bf not} supported on all of the sphere.
In \cite{BrDrSa2014} we outline how to derive the signed equilibrium on spherical caps centered at the South Pole and, ultimately, the s-extremal (positive) measure on $\mathbb{S}^d$ associated with the external field \eqref{eq:neg.external.field}. Here we present the details.

For the statement of the results we need to recall the following instrumental facts: We assume throughout this section that $s \geq d - 2$. \footnote{When $d = 2$, $s = d - 2$ is understood as the logarithmic case $s = \log$.}
The signed $s$-equilibrium on a spherical cap $\Sigma_t \DEF \{ \PT{x} \in \mathbb{S}^d : \PT{p} \cdot \PT{x} \leq 1 \}$ associated with $Q_{\PT{b},s}$ can be represented as the difference 
\begin{equation} \label{eq:eta} 
\eta_t = \frac{\Phi_s(t)}{W_s( \mathbb{S}^d )} \, \nu_t - q \varepsilon_t,
\end{equation}
in terms of the $s$-balayage measure\footnote{Given a measure $\nu$ and a compact set $K$ (of the sphere $\mathbb{S}^d$), the balayage measure $\hat{\nu}:=\bal_s(\nu,K)$ preserves the Riesz $s$-potential of $\nu$ onto the set $K$ and diminishes it elsewhere (on the sphere $\mathbb{S}^d$).} onto $\Sigma_t$ of the positive unit point charge at $\PT{b}$ and the uniform measure $\sigma_d$ on $\mathbb{S}^d$ given by
\begin{equation} \label{bal}
\epsilon_t = \epsilon_{t,s} \DEF \bal_s(\delta_{\PT{b}},\Sigma_t), \qquad \nu_t = \nu_{t,s} \DEF \bal_s(\sigma_d,\Sigma_t).
\end{equation}
Furthermore, the function 
\begin{equation} \label{eq:Phi}
\Phi_s(t) \DEF \frac{W_s(\mathbb{S}^d) \left( 1 + q \left\|\varepsilon_t\right\| \right)}{\left\|\nu_t\right\|}, \qquad d-2 < s < d, 
\end{equation}
where $\|\epsilon_t\| = \int_{\mathbb{S}^d} \dd \epsilon_t$ and $\|\nu_t\| = \int_{\mathbb{S}^d} \dd \nu_t$, plays an important role in the determination of the support of the $s$-extremal measure on~$\mathbb{S}^d$. In particular, one has
\begin{equation} \label{eq:weighted.potential}
U_s^{\eta_t}( \PT{x} ) + Q_{\PT{b},s}( \PT{x} ) = G_{\Sigma_t,Q_{\PT{b},s},s} = \Phi_s(t) = \mathcal{F}_s( \Sigma_t ).
\end{equation}
Here, $\mathcal{F}_s$ is the functional of \eqref{eq:F.s.functional} for the field $Q_{\PT{b},s}$.

First, we provide an extended version of \cite[Theorem~19]{BrDrSa2014} that includes asymptotic formulas for density and weighted potential valid near the boundary of the spherical cap.

\begin{prop} \label{prop:SignEq.general} Let $d-2<s<d$. The signed $s$-equilibrium $\eta_t$
on the spherical cap $\Sigma_t \subset \mathbb{S}^d$, $-1 < t < 1$, associated with $Q_{\PT{b},s}$ in \eqref{eq:neg.external.field} is
given by \eqref{eq:eta}.
It is absolutely continuous in the sense that for $\PT{x} = ( \sqrt{1-u^2} \, \overline{\PT{x}}, u) \in \Sigma_t$,
\begin{equation} \label{eq:eta.t}
\dd \eta_{t}(\PT{x}) = \eta_{t}^{\prime}(u) \frac{\omega_{d-1}}{\omega_{d}} \left( 1 - u^2 \right)^{d/2-1} \dd u \dd\sigma_{d-1}(\overline{\PT{x}}),
\end{equation}
where (with $R=|\PT{b}|$ and $r = \sqrt{R^2 + 2 R t + 1}$)
\begin{equation}
\begin{split} \label{eta.t.prime.1st.result}
\eta_{t}^{\prime}(u) &= \frac{1}{W_s(\mathbb{S}^d)} \frac{\gammafcn(d/2)}{\gammafcn(d-s/2)}
\left( \frac{1-t}{1-u} \right)^{d/2} \left( \frac{t-u}{1-t} \right)^{(s-d)/2} \\
&\phantom{=\times}\times \Bigg\{ \Phi_s (t)
\HypergeomReg{2}{1}{1,d/2}{1-(d-s)/2}{\frac{t-u}{1-u}}  \\
&\phantom{=\times\pm}-  \frac{q\left( R - 1 \right)^{d-s}}{r^{d}}
\HypergeomReg{2}{1}{1,d/2}{1-(d-s)/2}{\frac{\left(R+1\right)^{2}}{r^{2}}
\, \frac{t-u}{1-u}}  \Bigg\}.
\end{split}
\end{equation}
The density $\eta_{t}^{\prime}$ is expressed in terms of regularized Gauss hypergeometric functions. 
As $u$ approaches $t$ from below, we get
\begin{equation}
\begin{split} \label{eq:weighted.eta.density.near.t}
\eta_{t}^{\prime}(u) 
&= \frac{1}{W_s( \mathbb{S}^d )} \, \frac{\gammafcn(d/2)}{\gammafcn(d-s/2) \gammafcn( 1 - (d-s)/2 )} \left( \frac{t - u}{1 - t} \right)^{(s-d)/2} \\
&\phantom{=\pm}\times \left\{ \Phi_s( t ) - q \frac{\left( R - 1 \right)^{d-s}}{r^d} + \mathcal{R}(t ) \left( t - u \right) + \mathcal{O}( \left( t - u \right)^2 ) \right\} \quad \text{as $u \to t^-$,}
\end{split}
\end{equation}
where
\begin{equation} \label{eq:weighted.eta.density.near.t.middle.term}
\mathcal{R}( t ) = \frac{d}{2} \frac{\frac{d-s}{2}}{1 - \frac{d-s}{2}} \, \frac{\Phi_s( t ) - q \frac{\left( R - 1 \right)^{d-s}}{r^d}}{1 - t} - \frac{\frac{d}{2}}{1 - \frac{d-s}{2}} \, q \frac{\left( R - 1 \right)^{d-s}}{r^d} \, \frac{2 R}{r^2}.
\end{equation}

Furthermore, if $\PT{z} = ( \sqrt{1-\xi^2}\; \overline{\PT{z}},
\xi)\in \mathbb{S}^d$, the weighted $s$-potential is given by
\begin{align}
U_s^{\eta_t}(\PT{z})+Q_{\PT{b},s}(\PT{z}) &= \Phi_s(t), \qquad \PT{z} \in \Sigma_t, \label{eq:weighted.inside} \\
\begin{split}
U_s^{\eta_t}(\PT{z})+Q_{\PT{b},s}(\PT{z}) &= \Phi_s(t) + \frac{q}{\rho^s} \, \mathrm{I}\Big(\frac{(R-1)^2}{r^2} \frac{\xi-t}{1+\xi};
\frac{d-s}{2}, \frac{s}{2} \Big) \\
&\phantom{=\pm}- \Phi_s(t) \, \mathrm{I}\Big(\frac{\xi-t}{1+\xi};
\frac{d-s}{2}, \frac{s}{2}\Big), \qquad \PT{z} \in \mathbb{S}^d \setminus \Sigma_t, \label{eq:weighted.outside}
\end{split}
\end{align}
where $\rho=\sqrt{R^2+2R\xi+1}$ and $\mathrm{I}(x;a,b)$ is the regularized incomplete Beta function. 

As $\xi$ approaches $t$ from above, we get
\begin{equation}
\begin{split} \label{eq:weighted.near.t.plus}
U_s^{\eta_t}(\PT{z})+Q_{\PT{b},s}(\PT{z}) 
&= \Phi_s(t) + \left( \frac{\xi - t}{ 1 + t} \right)^{(d-s)/2} \Bigg\{ \frac{\gammafcn(d/2)}{\gammafcn( 1 + (d-s)/2 ) \, \gammafcn(s/2)} \\
&\phantom{=\pm}\times \left[ q \frac{\left( R - 1 \right)^{d-s}}{r^d} - \Phi_s(t) \right] + \mathcal{O}( \xi - t ) \Bigg\} \qquad \text{as $\xi \to t^+$.}
\end{split}
\end{equation}
\end{prop}

The next remarks, leading up to Proposition~\ref{prop:s.equilibrium.measure}, emphasize the special role of $\Phi_s( t )$.

\begin{rmk}
The behavior of the density $\eta_t^\prime$ near the boundary of $\Sigma_t$ \emph{inside} $\Sigma_t$ determines if $\eta_t$ is a positive measure; namely, the signed equilibrium $\eta_t$ on $\Sigma_t$ associated with $Q_{\PT{b},s}$ is a positive measure with support $\Sigma_t$ if and only if
\begin{equation} \label{eq:weighted.neccessary.density}
\Phi_s( t ) \geq q \frac{\left( R - 1 \right)^{d-s}}{\left( R^2 + 2 R t + 1 \right)^{d/2}}.
\end{equation}
Indeed, relations \eqref{eq:weighted.eta.density.near.t} and \eqref{eq:weighted.eta.density.near.t.middle.term} show that \eqref{eq:weighted.neccessary.density} is necessary and sufficient for $\eta_t^\prime( u ) > 0$ in a sufficiently small neighborhood $(t-\eps, t)$, $\eps > 0$, and this inequality extends to all of $[-1,t)$ as shown after the proof of Proposition~\ref{prop:SignEq.general}. Note that the density $\eta_t^\prime(u)$ has a singularity at $u = t$ if $\Phi_s(t) - q ( R - 1 )^{d-s} / ( R^2 + 2 R t + 1 )^{d/2} \neq 0$ and $\eta_t^\prime(u)$ approaches $0$ as $u \to t^-$ when equality holds in \eqref{eq:weighted.neccessary.density}. In that case, however, 
\begin{equation*}
\frac{\dd \eta_{t}^{\prime}}{\dd u} = \frac{\gammafcn(1+d/2)}{\gammafcn(d-s/2) \gammafcn( 1 - (d-s)/2 )} \left( \frac{t - u}{1 - t} \right)^{(s-d)/2} \left\{ \frac{\Phi_s(t)}{W_s( \mathbb{S}^d )}  + \mathcal{O}( t - u ) \right\}.
\end{equation*}
\end{rmk}

\begin{rmk}
The weighted $s$-potential of the signed equilibrium $\eta_t$ on $\Sigma_t$ associated with $Q_{\PT{b},s}$ exceeds the value $\Phi_s(t)$ assumed on $\Sigma_t$ strictly \emph{outside} of $\Sigma_t$ (but on $\mathbb{S}^d$) if and only if 
\begin{equation} \label{eq:weighted.neccessary}
\Phi_s( t ) \leq q \frac{\left( R - 1 \right)^{d-s}}{\left( R^2 + 2 R t + 1 \right)^{d/2}}.
\end{equation}
%
Indeed, the expansion \eqref{eq:weighted.near.t.plus} shows that $U_s^{\eta_t}(\PT{z})+Q_{\PT{b},s}(\PT{z}) > \Phi_s(t)$ in a small neighborhood of the boundary of $\mathbb{S}^d \setminus \Sigma_t$ if and only if \eqref{eq:weighted.neccessary} holds. In addition, if $\Phi_s(t)$ satisfies \eqref{eq:weighted.neccessary}, then 
\begin{equation} \label{eq:weighted.larger.than}
U_s^{\eta_t}(\PT{z})+Q_{\PT{b},s}(\PT{z}) > \Phi_s(t) \qquad \text{everywhere on $\mathbb{S}^d \setminus \Sigma_t$.}
\end{equation}
(This inequality is shown after the proof of Proposition~\ref{prop:SignEq.general}.) The weighted $s$-potential of $\eta_t$ tends to $\Phi_s( t )$ when the boundary of $\mathbb{S}^d \setminus \Sigma_t$ is approached from the outside ($\xi \to t^+$). There is a vertical tangent if $\Phi_s(t) - q ( R - 1 )^{d-s} / ( R^2 + 2 R t + 1 )^{d/2} \neq 0$ which turns into a horizontal one at a $t$ for which equality holds in \eqref{eq:weighted.larger.than}. In such a case
\begin{equation*}
\frac{\dd}{\dd \xi} \Big\{ U_s^{\eta_t}(\PT{z})+Q_{\PT{b},s}(\PT{z}) \Big\} = \left( \frac{\xi - t}{1 + t} \right)^{(d-s)/2} \left\{ \frac{\gammafcn( 1 + d/2 )}{\gammafcn( 1 + ( d - s ) / 2 ) \gammafcn( s/2 )} \frac{2R \, \Phi_s( t )}{R^2 + 2 R t + 1} + \mathcal{O}( \xi - t )  \right\}.
\end{equation*}
\end{rmk}

\begin{figure}[ht]
\begin{center}
\begin{minipage}{0.5\linewidth}
\centerline{\includegraphics[scale=.925]{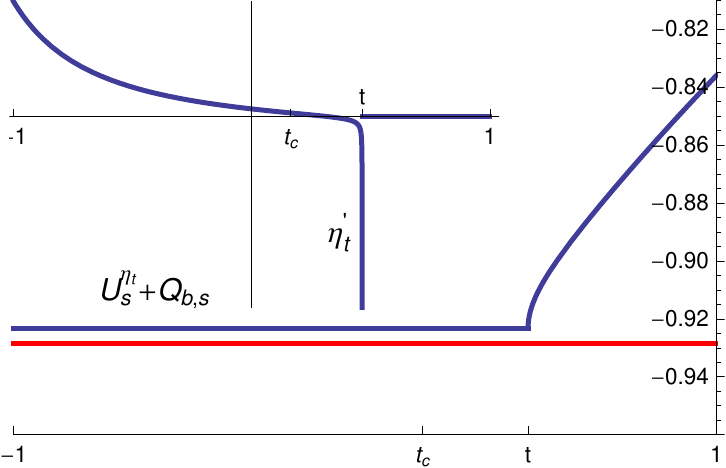}}
\end{minipage}
\begin{minipage}{0.4\linewidth}
\small
\begin{align*}
t &> t_c, \\
U_s^{\eta_t}(\PT{z}) + Q_{\PT{b},s}(\PT{z}) &\geq \mathcal{F}_s( \Sigma_t ) \quad \text{on $\mathbb{S}^d\setminus\Sigma_t$,} \\
U_s^{\eta_t}(\PT{z}) + Q_{\PT{b},s}(\PT{z}) &= \mathcal{F}_s( \Sigma_t ) \quad \text{on $\Sigma_t$,} \\
\eta_t^\prime &\ngeq 0 \quad \text{on $\Sigma_t$.} 
\end{align*}
\end{minipage}
\begin{minipage}{0.5\linewidth}
\centerline{\includegraphics[scale=.925]{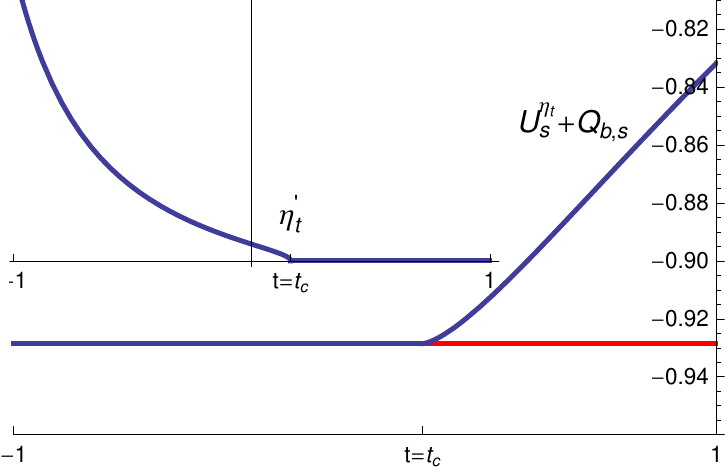}}
\end{minipage}
\begin{minipage}{0.4\linewidth}
\small
\begin{align*}
t &= t_c, \\
U_s^{\eta_t}(\PT{z}) + Q_{\PT{b},s}(\PT{z}) &\geq \mathcal{F}_s( \Sigma_t ) \quad \text{on $\mathbb{S}^d\setminus\Sigma_t$,} \\
U_s^{\eta_t}(\PT{z}) + Q_{\PT{b},s}(\PT{z}) &= \mathcal{F}_s( \Sigma_t ) \quad \text{on $\Sigma_t$,} \\
\eta_t^\prime &\geq 0 \quad \text{on $\Sigma_t$.} 
\end{align*}
\end{minipage}
\begin{minipage}{0.5\linewidth}
\centerline{\includegraphics[scale=.925]{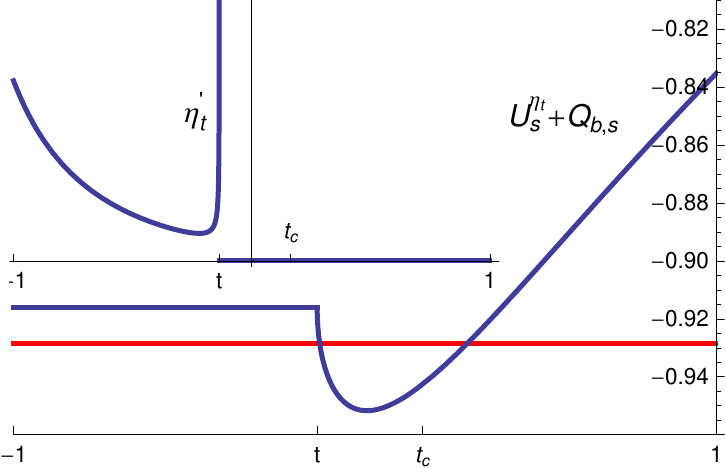}} 
\end{minipage}
\begin{minipage}{0.4\linewidth}
\small
\begin{align*}
t &< t_c, \\
U_s^{\eta_t}(\PT{z}) + Q_{\PT{b},s}(\PT{z}) &\ngeq \mathcal{F}_s( \Sigma_t ) \quad \text{on $\mathbb{S}^d\setminus\Sigma_t$,} \\
U_s^{\eta_t}(\PT{z}) + Q_{\PT{b},s}(\PT{z}) &= \mathcal{F}_s( \Sigma_t ) \quad \text{on $\Sigma_t$,} \\
\eta_t^\prime &\geq 0 \quad \text{on $\Sigma_t$.} 
\end{align*}
\end{minipage}
\caption{\label{fig5} The weighted $s$-potential of $\eta_t$ for $t>t_c$, $t=t_c$, and $t<t_c$ versus altitude $\xi$ of $\PT{z} = ( \sqrt{1-\xi^2} \, \overline{z}, \xi) \in \mathbb{S}^d$ for $d=2$, $s=1$, $q=-5$, and $R=1+\phi$ ($\phi$ the Golden ratio), cf. Propositions~\ref{prop:SignEq.general} and \ref{prop:s.equilibrium.measure}. Insets show the respective density $\eta_t^\prime$. The horizontal line (Red in colored version) indicates $\Phi_s( t_c ) = \mathcal{F}_s( \Sigma_{t_c} ) = G_{Q_{\PT{b},s},s}$. Observe the vertical tangent at the graph of $\eta_t^\prime$ as $t \to t_c^-$ in the middle display (see first remark after Proposition~\ref{prop:SignEq.general}).}
\end{center}
\end{figure}

For $\eta_t$ to coincide with the $s$-extremal measure on $\mathbb{S}^d$ associated with $Q_{\PT{b},s}$ with support $\Sigma_t$ both \eqref{eq:weighted.neccessary} and \eqref{eq:weighted.neccessary.density} have to hold. The next result is \cite[Theorem~20]{BrDrSa2014}.

\begin{prop} \label{prop:s.equilibrium.measure}
Let $d - 2 < s < d$. For the external field \eqref{eq:neg.external.field} the function $\Phi_s (t)$ given in \eqref{eq:Phi} has precisely one global minimum $t_c\in (-1,1]$. This minimum is either the unique solution $t_c\in (-1,1)$ of the equation 
\begin{equation*}
\Phi_s (t) = q \left( R - 1 \right)^{d-s} \big/ \left( R^2 + 2 R t + 1 \right)^{d/2},
\end{equation*}
or $t_c=1$ when such a solution does not exist. In addition, $\Phi_s(t)$ is greater than the right-hand side above if $t \in (-1,t_c)$ and is less than if $t \in (t_c,1)$.  Moreover, $t_c = \max \{ t : \eta_t \geq 0 \}$.
The extremal measure $\mu_{Q_{\PT{b},s}}$ on $\mathbb{S}^d$ is given by $\eta_{t_c}$ (see \eqref{eq:eta.t}), and $\supp( \mu_{Q_{\PT{b},s}} ) = \Sigma_{t_c}$. 
\end{prop}

Figure~\ref{fig5} illustrates the typical behavior of the signed equilibrium (using density and weighted potential) on spherical caps that are too large, too small, and have "just" the "right" size. The right column shows which conditions are violated when the spherical cap is too small or too large. This figure should be compared with \cite[Fig.~1]{BrDrSa2009}.

In the limiting case $s = d - 2$ with $s > 0$ it can be shown that the $s$-balayage measures
\begin{equation} \label{barBal}
\overline{\epsilon}_t \DEF \epsilon_{t,d-2} = \bal_{d-2}(\delta_{\PT{b}},\Sigma_t), \qquad  \overline{\nu}_t \DEF \nu_{t,d-2} = \bal_{d-2}(\sigma,\Sigma_t)
\end{equation}
exist and both have a component that is uniformly distributed on the boundary of $\Sigma_t$. 
Moreover, unlike the case $d-2 < s < d$, the density for $\mu_{Q_{\PT{b},s}}$, where $s = d - 2$, does not vanish on the boundary of its support.
We introduce the measure 
\begin{equation*}
\beta_t( \PT{x} ) \DEF \delta_t (u) \cdot \sigma_{d-1}(\overline{\PT{x}}), \qquad \PT{x}=(\sqrt{1-u^2} \,\overline{\PT{x}}, u) \in \mathbb{S}^d.
\end{equation*}
Using similar methods as in \cite{BrDrSa2009}, one can show the following.

\begin{thm} \label{ExcepThm} Let $d \geq 3$. The signed $s$-equilibrium $\overline{\eta}_t$ on the spherical cap $\Sigma_t$ associated with $\overline{Q}_{\PT{b},d-2}(\PT{x}) = q \, | \PT{x} - \PT{b} |^{2-d}$, $q < 0$ and $\PT{b} = ( \PT{0}, -R )$ ($R > 1$), is given by
\begin{equation*}
\overline{\eta}_t = \left[ \overline{\Phi}_{d-2}(t) / W_{d-2}(\mathbb{S}^d) \right] \overline{\nu}_t - q \overline{\epsilon}_t, \qquad \overline{\Phi}_{d-2}(t) \DEF  W_{d-2}(\mathbb{S}^d) \left( 1 + q \left\|\overline{\epsilon}_t\right\| \right) / \left\|\overline{\nu}_t\right\|,
\end{equation*}
where $\overline{\nu}_t$ and $\overline{\epsilon}_t$ are given in \eqref{barBal}. More explicitly, for $\PT{x} = ( \sqrt{1 - u^2} \, \overline{\PT{x}}, u ) \in \mathbb{S}^d$,
\begin{equation} \label{etabar}
\dd \overline{\eta}_t(\PT{x}) = \overline{\eta}_t^\prime(u) \, \dd \sigma_{d}\big|_{\Sigma_t}(\PT{x}) + \overline{q}_t \, \dd \beta_t(\PT{x}),
\end{equation}
where the density with respect to $\sigma_d$ restricted to $\Sigma_t$ takes the form
\begin{equation*}
\overline{\eta}_t^\prime(u) = \frac{\overline{\Phi}_{d-2}(t)}{W_{d-2}(\mathbb{S}^d)} - \frac{q}{W_{d-2}(\mathbb{S}^d)} \, \frac{\left( R^2 - 1 \right)^2}{\left( R^2 + 2 R u + 1 \right)^{d/2+1}}
\end{equation*}
and the boundary charge uniformly distributed over the boundary of $\Sigma_t$ is given by
\begin{equation*}
\overline{q}_t = \frac{1-t}{2} \left( 1 - t^2 \right)^{d/2-1} \left[ \overline{\Phi}_{d-2}(t) - \frac{q \left( R - 1 \right)^2}{\left( R^2 + 2 R t + 1 \right)^{d/2}} \right].
\end{equation*}

Furthermore, for any fixed $t\in(-1,1)$, the following weak$^*$ convergence holds:
\begin{equation} \label{weak.star.conv.exceptional}
\nu_{t,s}\stackrel{*}{\longrightarrow} \overline{\nu}_t , \qquad \epsilon_{t,s} \stackrel{*}{\longrightarrow}
\overline{\epsilon}_t, \qquad \text{as $s\to(d-2)^+$.}
\end{equation}

The function $\overline{\Phi}_{d-2}(t)$ has precisely one global minimum
$t_c\in (-1,1]$. This minimum is either the unique solution $t_c\in (-1,1)$ of the equation
\begin{equation*}
\overline{\Phi}_{d-2}(t) = q \left( R - 1 \right)^2 \big/ \left( R^2 + 2 R t + 1 \right)^{d/2},
\end{equation*}
or $t_c=1$ when such a solution does not exist. Moreover, $t_c=\max\{ t: \overline{\eta}_t \geq 0 \}$.

The extremal measure $\mu_{\overline{Q}_{\PT{b},d-2}}$ on $\mathbb{S}^d$ with $\supp(\mu_{\overline{Q}_{\PT{b},d-2}}) = \Sigma_{t_c}$
is given by
\begin{equation} \label{etabarzero.exceptional}
\dd \mu_{\overline{Q}_{\PT{b},d-2}}(\PT{x}) = \dd \overline{\eta}_{t_0}(\PT{x}) = \frac{\overline{\Phi}_{d-2}(t_0)}{W_{d-2}(\mathbb{S}^d)} \left[ 1 - \frac{\left( R + 1 \right)^2 \left( R^2 + 2 R t_c + 1 \right)^{d/2}}{\left( R^2 + 2 R u + 1 \right)^{d/2+1}}  \right] \dd \sigma_{d}\big|_{\Sigma_{t_0}}(\PT{x}).
\end{equation}

Furthermore, if $\PT{z} = ( \sqrt{1-\xi^2}\; \overline{\PT{z}}, \xi) \in \mathbb{S}^d$, the weighted $(d-2)$-potential is given by
\begin{align}
U_{d-2}^{\overline{\eta}_t}(\PT{z}) &+\overline{Q}_{\PT{b},d-2}(\PT{z}) = \overline{\Phi}_{d-2}(t), \quad \PT{z} \in \Sigma_t, \label{eq:weighted.inside.d-2} \\
\begin{split} \label{eq:weighted.outside.exceptional}
U_{d-2}^{\overline{\eta}_t}(\PT{z}) &+\overline{Q}_{\PT{b},d-2}(\PT{z}) = \overline{\Phi}_{d-2}(t) + \frac{q}{\rho^{d-2}} \left[ 1 - \left( 1 - \frac{\xi - t}{1 + \xi} \right)^{d/2-1} \right] \\
&\phantom{}- \overline{\Phi}_{d-2}(t) \, \left[ 1 - \left( 1 - \frac{( R - 1 )^2}{R^2 + 2 R t + 1} \, \frac{\xi - t}{1 + \xi} \right)^{d/2-1} \right], \quad \PT{z} \in \mathbb{S}^d \setminus \Sigma_t,
\end{split}
\end{align}
where $\rho=\sqrt{R^2+2R\xi+1}$. As $\xi$ approaches $t$ from above, we get 
\begin{equation} \label{eq:weighted.near.t.plus.exceptional}
U_{d-2}^{\overline{\eta}_t}(\PT{z})+\overline{Q}_{\PT{b},d-2}(\PT{z}) = \overline{\Phi}_{d-2}(t) + \frac{\xi - t}{ 1 + t} \Bigg\{ \frac{d-2}{2} \left[ q \frac{\left( R - 1 \right)^{2}}{r^d} - \overline{\Phi}_{d-2}(t) \right] + \mathcal{O}( \xi - t ) \Bigg\}.
\end{equation}
\end{thm}

A similar result holds for the logarithmic case $s = \log$ on $\mathbb{S}^2$.

\section{Proofs}
\label{sec:proofs}

\subsection{Proofs and Discussions for Section~\ref{sec:Gonchar.s.Question}}

First, we show the result for the signed $s$-equilibrium on $\mathbb{S}^d$. 

\begin{proof}[Proof of Theorem~\ref{thm:signed.equilibrium.sphere}]
For $R > 1$ this result has been proven in \cite{BrDrSa2009} (cf. \cite{BrDrSa2012} for the harmonic case). Let $0 \leq R < 1$. 
By linearity of the $s$-potential, we can write
\begin{equation*}
U_s^{\eta_{Q}}(\PT{x}) = \left[ 1 + \frac{q U_s^{\sigma_d} (\PT{a})}{W_s(\mathbb{S}^d)} \right] U_s^{\sigma_d}(\PT{x}) - \frac{q}{W_s(\mathbb{S}^d)} \int \frac{\left| \PT{y} - \PT{a} \right|^s}{\left( 1 - R^2 \right)^s \left| \PT{y} - \PT{x}\right|^{s}} \,  \frac{\left( 1 - R^2 \right)^d}{\left|\PT{y}-\PT{a}\right|^{d}} \frac{\dd \sigma_d( \PT{y} )}{\left| \PT{y} - \PT{a} \right|^d}.
\end{equation*}
Using Imaginary inversion (i.e., utilizing $| \PT{y} - \PT{a} | \, | \PT{y}^* - \PT{a} | = 1 - R^2$ and $\dd \sigma_d( \PT{y} ) / | \PT{y} - \PT{a} |^d = \dd \sigma_d( \PT{y}^* ) / | \PT{y}^* - \PT{a} |^d$), the integral reduces to $W_s(\mathbb{S}^d) / | \PT{x} - \PT{a} |^s$. Since $U_s^{\sigma_d}(\PT{x}) = W_s(\mathbb{S}^d)$ on $\mathbb{S}^d$, one gets that the weighted potential of $\eta_{Q}$ is constant,
\begin{equation*}
U_s^{\eta_{Q}}(\PT{x}) + \frac{q}{\left| \PT{x} - \PT{a} \right|^s} = W_s(\mathbb{S}^d) + q \, U_s^{\sigma_d} (\PT{a}) \qquad \text{everywehre on $\mathbb{S}^d$.}
\end{equation*}
This shows \eqref{eq:G.sphere.Q.s}. In a similar way,
\begin{equation*}
\int \dd \eta_{Q} = 1 + \frac{q U_s^{\sigma_d} (\PT{a})}{W_s(\mathbb{S}^d)} - \frac{q}{W_s(\mathbb{S}^d)} \int \frac{\left| \PT{y} - \PT{a} \right|^s}{\left( 1 - R^2 \right)^s} \,  \frac{\left( 1 - R^2 \right)^d}{\left|\PT{y}-\PT{a}\right|^{d}} \frac{\dd \sigma_d( \PT{y} )}{\left| \PT{y} - \PT{a} \right|^d} = 1,
\end{equation*}
where the integral reduces to $U_s^{\sigma_d}( a )$ under Imaginary inversion. Thus $\eta_{Q}$, indeed, satisfies Definition~\ref{def:signed.equilibrium} for $A = \mathbb{S}^d$. Furthermore, the signed measure $\eta_{Q}$ is absolutely continuous with respect to $\sigma_d$. The result follows.
\end{proof}

Next, we provide the technical details for the discussion in the two remarks after Proposition~\ref{prop}.

\begin{proof}[First remark after Proposition~\ref{prop}] 
For $R > 1$, we can use term-wise differentiation in the following series expansion of the right-hand side of \eqref{eqsigned} to show monotonicity,
\begin{equation*}
\sum_{k=0}^\infty \left[ 1 - \frac{\Pochhsymb{s/2}{k}}{\Pochhsymb{d}{k}} \right] \frac{\Pochhsymb{d/2}{k}}{k!} \frac{\left( 4 R \right)^k}{\left(R+1\right)^{s+2k}}.
\end{equation*}
For $R \in (0,1)$, we show that the following representation of the right-hand side of \eqref{eqsigned}
\begin{equation*}
\frac{\left(1+R\right)^{d-s}}{\left(1-R\right)^d} \left[ 1 - \Hypergeom{2}{1}{d-s/2,d/2}{d}{-\frac{4R}{(1-R)^2}} \right],
\end{equation*}
obtained by applying the linear transformation \cite[last of Eq.s~15.8.1]{NIST:DLMF} to~\eqref{eq:s.potential}, is strictly monotonically increasing on $(0,1)$ for each $s \in (0,d)$. It is easy to see that the ratio $(1+R)^{d-s} / (1-R)^d$ has this property for each $s \in (0,d)$ and using the differentiation formula~\cite[Eq.~15.5.1]{NIST:DLMF}, the square-bracketed expression above has a positive derivative on $(0,1)$ for each $s \in (0,d)$. The result follows.
\end{proof}

\begin{proof}[Second remark after Proposition~\ref{prop}]
The continuity of $f$ is evident by continuity of the $s$-potential $U_s^{\sigma_d}(\PT{a})$, which is a radial function depending on $R = | \PT{a} |$, in the potential-theoretical case $0 < s < d$. Note that, since $\sigma_d$ is the $s$-equilibrium measure on $\mathbb{S}^d$,
\begin{equation*}
f( 1 ) = - U_s^{\sigma_d}( 1 ) = - W_s(\mathbb{S}^d) < 0.
\end{equation*}
The negativity of $f$ in $( 0, 1 ) \cup ( 1, \infty )$ follows from
\begin{equation*}
U_s^{\sigma_d}( R ) = \frac{\big[ ( R - 1 )^2 \big]^{(d-s)/2}}{\left( R + 1 \right)^d} \Hypergeom{2}{1}{d-s/2,d/2}{d}{\frac{4R}{( R + 1 )^2}} > \frac{\left| R - 1 \right|^{d-s}}{\left( R + 1 \right)^d},
\end{equation*}
where the right-hand side is derived from \eqref{eq:s.potential} using the last transformation in \cite[Eq.~15.8.1]{NIST:DLMF}.
Clearly, $f(R) \geq - U_s^{\sigma_d}( R )$. From \eqref{eq:s.potential} it follows that~$U_s^{\sigma_d}( R )$ is strictly monotonically decreasing on $( 1, \infty )$. 
Since $\sigma_d$ is the $s$-equilibrium measure on $\mathbb{S}^d$, 
\begin{equation*}
f(R) \geq - U_s^{\sigma_d}( 1 ) = - W_s(\mathbb{S}^d) = f( 1 ) \qquad \text{on $[1, \infty)$.}
\end{equation*}
Since the continuous function $f$ is bounded on the compact interval $[0,1]$, $f$ is bounded from below on $[0, \infty)$.  
Verification of monotonicity of~$f$ on $(1,\infty)$ by direct calculation of~$f^\prime$ seems to be futile, but by using \eqref{eq:s.potential} and the differentiation formula \cite[Eq.~15.5.1]{NIST:DLMF}, we get
\begin{align*}
&s \left( R + 1 \right)^{s-1} f(R) + \left( R + 1 \right)^s f^\prime(R) = \left\{ \left( R + 1 \right)^s f(R) \right\}^{\prime} \\
&\phantom{equ}= \left( d - s \right) \left( \frac{R-1}{R+1} \right)^{d-1-s} \frac{2}{\left(R + 1 \right)^2} + s \frac{R-1}{(R+1)^3} \Hypergeom{2}{1}{1+s/2,1+d/2}{1+d}{\frac{4R}{(R+1)^2}} > 0.
\end{align*}
Since it has already been established that $f(R)<0$, it follows that ${f^\prime(R) > 0}$ on $(1, \infty)$.
It is easy to see that $f(R) \to 0$ as $R \to \infty$, thus $f$ has a horizontal asymptote at level~$0$. 
\end{proof}

Next, we prove results for Gonchar's problem for {\bf negative external fields}. The following proof concerns in particular Theorem~\ref{thm:Gonchar.D.superharmonic}.

\begin{proof}
We consider interior sources, that is $0 \leq R < 1$. The right-hand side in \eqref{neg.eqsigned} is
\begin{equation*}
f(R) \DEF g(R) - U_s^{\sigma_d}( R ), \qquad \text{where} \quad g(R) \DEF \left( 1 - R \right)^{d-s} \left( 1 + R \right)^{-d}.
\end{equation*}
We show that the equation $f^\prime(R)=0$, or equivalently, $g^\prime(R) / R = \{ U_s^{\sigma_d}(R) \}^\prime / R$ has only one solution in the interval $(0,1)$ in the superharmonic regime $0< s < d - 1$. Expanding the last equation using formula~\eqref{eq:s.potential.B} for $U_s^{\sigma_d}(R)$, we get
\begin{equation*}
\left[ 2d - s \left( 1 + R \right) \right] \frac{( 1 - R )^{d-s-1}}{( 1 + R )^{d+1}} \frac{1}{R} = \frac{s ( d - 1 - s )}{d + 1} \Hypergeom{2}{1}{1-(d-1-s)/2,1+s/2}{1+(d+1)/2}{R^2}.
\end{equation*}
The function $h(R) \DEF - g^\prime(R) / R$ at the left-hand side satisfies $h(R) \to +\infty$ as $R \to 0$ and $h(1) = 0$. Since $\{ R ( 1 + R )^{d+1} h(R) \}^\prime = \{ \left[ 2d - s \left( 1 + R \right) \right] ( 1 - R )^{d-s-1} \}^\prime < 0$,
it follows that $h^\prime(R) < 0$. The right-hand side of the equation assumes the positive value $s ( d - 1 - s ) / ( d + 1 )$ at $R = 0$ and is strictly decreasing if $0 < s < d - 3$, constant if $s = d - 3$, or strictly increasing if $d - 3 < s < d - 1$. In either case for $0< s < d - 1$ there is exactly one solution in the interval $(0,1)$. That is, $f(R)$ has a single minimum in the interval $(0,1)$ because $f^\prime(0) = s - 2d < 0$ and (\cite[Eq.~15.4.20]{NIST:DLMF})
\begin{align*}
f^\prime(1^-) 
&= \frac{s ( d - 1 - s )}{d + 1} \Hypergeom{2}{1}{1-(d-1-s)/2,1+s/2}{1+(d+1)/2}{1} \\
&= \frac{s ( d - 1 - s )}{d + 1} \frac{\gammafcn(1+(d+1)/2)\gammafcn(d-1-s)}{\gammafcn(d-s/2) \gammafcn((d+1-s)/2)} \\
&= \frac{s}{2} \, \frac{\gammafcn((d+1)/2) \gammafcn(d-s)}{\gammafcn(d-s/2) \gammafcn((d+1-s)/2)} \\
&= \frac{s}{2} \, W_s(\mathbb{S}^d) > 0, \qquad 0 < s < d - 1.
\end{align*}
The last step follows by applying the duplication formula for the gamma function and \eqref{eq:W.s.S.d}. In particular, the minimum from above is strictly less than $- f(1^-) = - W_s(\mathbb{S}^d)$. 

The function $g(R)$ is strictly monotonically decreasing on $(0,1)$ for all $0 < s < d$. From
\begin{equation*}
g^{\prime\prime}(R) = \frac{(1 - R)^{d-s-2}}{(R + 1)^{d+2}} \left\{ \left[ 2 ( d + 1 ) - ( s + 2 ) ( R + 1 ) \right] \left[ 2 d - s ( R + 1) \right] + s ( 1 - R^2 ) \right\}
\end{equation*}
one infers that $g(R)$ changes from convex to concave as $R \to 1$ and $g^{\prime\prime}(R) \to - \infty$ as $R \to 1^-$ if $d - 1 < s < d$. Since the $s$-potential $U_s^{\sigma_d}(R)$ is a strictly increasing and convex function on $(0,1)$ in the subharmonic range $d - 1 < s < d$, 
it follows that $f(R)$ is strictly decreasing on $(0,1)$ and neither convex nor concave on all of $(0,1)$ for $d - 1 < s < d$. Moreover, $f^{\prime\prime}(0)>0$ and $f^{\prime\prime}(R) \to - \infty$ as $R \to 1^-$ for $d - 1 < s < d$.
\end{proof}

\subsection{Proofs and Discussions for Section~\ref{sec:negatively.charged.external field}}
\label{subsec:proofs.discussions.neg.external.field}

In the following we make use of a Kelvin transformation (spherical inversion) of points and measures that maps $\mathbb{S}^d$ to $\mathbb{S}^d$. 
Let $\kelvin_R$ denote the Kelvin transformation (stereographic projection) with center $\PT{b}=(\PT{0},-R)$ and radius $\sqrt{R^2-1}$; that is, for any point $\PT{x}\in \mathbb{R}^{d+1}$ the image $\PT{x}^*\DEF\kelvin_R(\PT{x})$ lies on a ray stemming from $\PT{b}$, and passing through $\PT{x}$ such that 
\begin{equation} \label{eq:KelTr}
\left| \PT{x} - \PT{b} \right| \cdot \left| \PT{x}^* - \PT{b} \right| = R^2 - 1.
\end{equation}
\begin{figure}[htb]
\includegraphics[scale=.075]{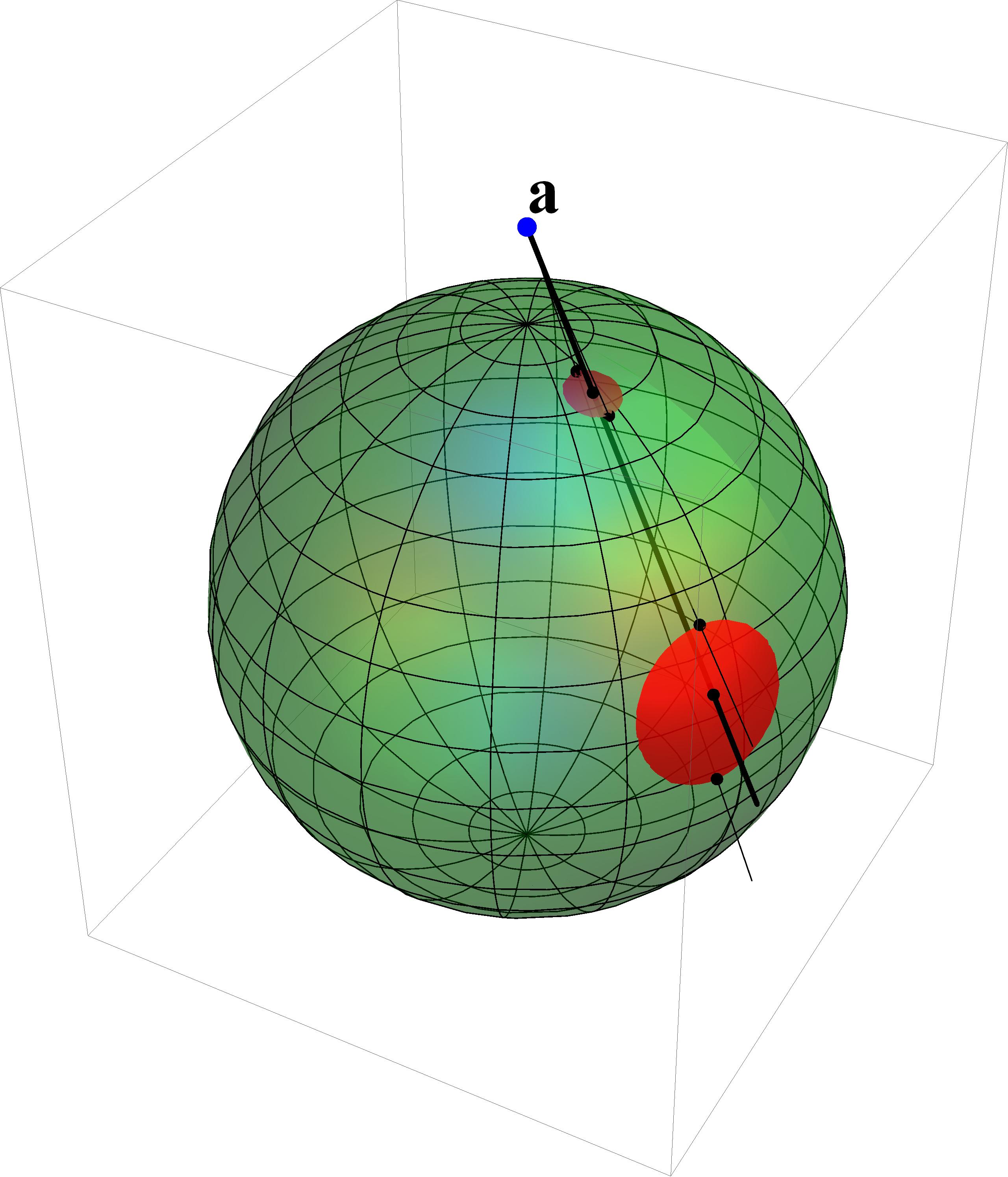}
\caption{\label{fig:inversion} Kelvin transform with center at $\PT{a}$ and radius $\sqrt{R^2 - 1}$.}
\end{figure}
The image of $\PT{x} = ( \sqrt{1-u^2} \, \overline{\PT{x}}, u ) \in \mathbb{S}^d$ is again a point $\PT{x}^* = ( \sqrt{1-({u^*})^2} \,\overline{\PT{x}}, u^* ) \in \mathbb{S}^d$, where the formulas
\begin{equation} \label{eq:u_rel}
1 + u^* = \frac{\left( R - 1 \right)^2}{R^2 + 2 R u + 1} \left( 1 - u \right), \qquad 1 - u^* = \frac{\left( R + 1 \right)^2}{R^2 + 2 R u + 1} \left( 1 + u \right),
\end{equation}
relating the heights $u$ and $u^*$, follow from similar triangle proportions. From this and the formula $| \PT{x} - \PT{y} |^2 = 2 - 2 \, \PT{x} \cdot \PT{y}$ for $\PT{x}, \PT{y} \in \mathbb{S}^d$ it follows that the Euclidean distance of two points on the sphere transforms like
\begin{equation} \label{eq:KelTrDist}
\left| \PT{x}^{*} - \PT{y}^{*} \right| = \left( R^{2} - 1 \right)\frac{\left| \PT{x} - \PT{y} \right|}{\left| \PT{x} - \PT{a} \right| \left| \PT{y} - \PT{a} \right|}, \qquad \PT{x},\PT{y} \in \mathbb{S}^{d}.
\end{equation}
Geometric properties include that the Kelvin transformation maps the North Pole $\PT{p}$ to the South Pole $\PT{q}$ and vice versa, $\kelvin_R (\mathbb{S}^d) =\mathbb{S}^d$, and $\kelvin_R$ sends the spherical cap $A_R \DEF \{ (\sqrt{1-u^2} \, \overline{\PT{x}}, u ) : -1/R \leq u \leq 1, \overline{\PT{x}} \in \mathbb{S}^{d-1} \}$ to $B_R \DEF \{ ( \sqrt{1-u^2} \, \overline{\PT{x}}, u ) : -1 \leq u \leq -1/R, \overline{\PT{x}} \in \mathbb{S}^{d-1} \}$ and vice versa, with the points on the boundary being fixed. 

We note that the uniform measure $\sigma_d$ on $\mathbb{S}^d$ transforms like
\begin{equation} \label{eq:sigma.d.transformation}
\left| \PT{x}^* - \PT{b} \right|^{-d} \dd \sigma_d( \PT{x}^* ) = \left| \PT{x} - \PT{b} \right|^{-d} \dd \sigma_d( \PT{x} ).
\end{equation} 
Further, given a measure $\lambda$ with no point mass at $\PT{b}$, its Kelvin transformation (associated with a fixed $s$)
$\lambda^*=\kelvinMEAS_{\PT{b},s}(\lambda)$ is a measure defined by
\begin{equation} \label{eq:KelMeas}
\dd\lambda^* (\PT{x}^*) \DEF \left(R^2-1\right)^{s/2} \left|\PT{x}-\PT{b}\right|^{-s} \dd\lambda(\PT{x}),
\end{equation}
where the $s$-potentials of the two measures are related as follows (e.g. \cite[Eq. (5.1)]{DrSa2007})
\begin{equation} \label{eq:KelPot}
U_s^{\lambda^*}(\PT{x}^*) = \int \frac{\dd\lambda^*(\PT{y}^*)}{\left|\PT{x}^*-\PT{y}^*\right|^s} = \int \frac{\left|\PT{x}-\PT{b}\right|^s \dd\lambda(\PT{y})}{\left(R^2-1\right)^{s/2} \left|\PT{x}-\PT{y} \right|^s} = \frac{\left|\PT{x}-\PT{b}\right|^s} {\left(R^2-1\right)^{s/2}} \, U_s^\lambda(\PT{x}).
\end{equation}
(The Kelvin transformation has the duality property $\kelvinMEAS_{\PT{b},s} (\lambda^* (\PT{x}^*))=\lambda(\PT{x})$.)

For convenience, we recall the specific form of the $s$-balayage onto $\Sigma_t$ of the uniform measure (see \cite[Lemma~24]{BrDrSa2009}) and its norm (see \cite[Lemma~30]{BrDrSa2009}).
\begin{prop} \label{prop:nu}
Let $d-2 < s < d$. The measure $\nu_t = \bal_s( \sigma_d, \Sigma_t )$ is given by
\begin{equation} \label{eq:NuR}
\dd\nu_{t}(\PT{x}) = \nu_{t}^{\prime}(u) \frac{\omega_{d-1}}{\omega_{d}} \left( 1 - u^{2} \right)^{d/2-1} \dd u \dd \sigma_{d-1}(\overline{\PT{x}}), \qquad \PT{x} \in \Sigma_t,
\end{equation}
where the density $\nu_{t}^{\prime}(u)$ is given by
\begin{equation} \label{eq:nu.dens}
\nu_{t}^{\prime}(u) \DEF \frac{\gammafcn(d/2)}{\gammafcn(d-s/2)} \left( \frac{1-t}{1-u} \right)^{d/2} \left( \frac{t-u}{1-t} \right)^{(s-d)/2} \HypergeomReg{2}{1}{1,d/2}{1-(d-s)/2}{\frac{t-u}{1-u}}.
\end{equation}

Furthermore,
\begin{align}
\left\|\nu_{t}\right\|
&= \frac{2^{1-d} \gammafcn(d)}{\gammafcn(d-s/2)\gammafcn(s/2)} \int_{-1}^t \left( 1 + u \right)^{s/2-1} \left( 1 - u \right)^{d-s/2-1} \dd u  \label{NuNormB} \\
&= 1 - \mathrm{I}\left((1-t)/2;d-s/2,s/2\right). 
\label{NuNormA} 
\end{align}
\end{prop}

Suppose 
\begin{equation*}
\epsilon_t = ( R^2 - 1 )^{-s/2} \kelvinMEAS_{\PT{b},s}( \lambda^*( \PT{x}^* ) ),
\end{equation*}
where $\lambda^*( \PT{x}^*)$ is the $s$-extremal measure on the image $\Sigma_t^* = \kelvin_R( \Sigma_t )$. Then by \eqref{eq:KelPot}
\begin{equation*}
U_s^{\epsilon_t}( \PT{x} ) = \left( R^2 - 1 \right)^{-s/2} U_s^{\kelvinMEAS_{\PT{b},s}( \lambda^*( \PT{x}^* ) )}( \PT{x} ) = \frac{1}{\left| \PT{x} - \PT{b} \right|^s} \, U_s^{\lambda^*}( \PT{x}^* ) = \frac{1}{\left| \PT{x} - \PT{b} \right|^s}, \qquad \PT{x} \in \Sigma_t.
\end{equation*}
With this idea in mind we can easily prove the analogous of \cite[Lemma~25 and Lemma~29]{BrDrSa2009}.

\begin{lem} \label{lem:epsilon} Let $d - 2 < s < d$. The measure $\epsilon_t = \bal_s( \delta_{\PT{b}}, \Sigma_t )$ is given by
\begin{equation} \label{eq:dd.eps.t}
\dd\epsilon_{t}(\PT{x}) = \epsilon_{t}^{\prime}(u) \frac{\omega_{d-1}}{\omega_{d}} \left( 1 - u^2 \right)^{d/2-1} \dd u \dd\sigma_{d-1}(\overline{\PT{x}}), \qquad \PT{x} \in \Sigma_t,
\end{equation}
and setting $r^2 \DEF R^2 + 2 R t + 1$, the density is given by
\begin{equation}
\begin{split} \label{eq:eps.dens} 
\epsilon_{t}^{\prime}(u) &\DEF \frac{1}{W_{s}(\mathbb{S}^{d})} \, \frac{\gammafcn(d/2)}{\gammafcn(d-s/2)} \frac{\left( R - 1 \right)^{d-s}}{r^{d}} \left(  \frac{1-t}{1-u} \right)^{d/2} \\
&\phantom{=\times}\times \left( \frac{t-u}{1-t} \right)^{(s-d)/2} \HypergeomReg{2}{1}{1,d/2}{1-(d-s)/2}{\frac{\left(R+1\right)^{2}}{r^{2}} \, \frac{t-u}{1-u}}.
\end{split}
\end{equation}

The norm of $\epsilon_t$ is given by
\begin{equation}
\begin{split} \label{eq:NormEpsB}
\left\| \epsilon_{t} \right\|
&= \frac{2^{1-d} \gammafcn(d)}{\gammafcn(d-s/2)\gammafcn(s/2)} \frac{\left( R - 1 \right)^{d-s}}{W_s(\mathbb{S}^d)} \int_{-1}^{t} \frac{ \left( 1 + u \right)^{s/2-1} \left( 1 - u \right)^{d-s/2-1} }{ \left( R^2 + 2 R u + 1 \right)^{d/2} } \dd u.
\end{split}
\end{equation}
\end{lem}

\begin{proof}
For $\lambda^* = \bal_s( \sigma_d, \Sigma_t^* ) / W_s( \mathbb{S}^d )$ (see \cite[Eq.~(40) and following equation]{BrDrSa2009})
\begin{equation*} 
\dd\lambda^*(\PT{x}^*) = (\lambda^*)^{\prime}(u^*) \frac{\omega_{d-1}}{\omega_{d}} \left[ 1 - (u^*)^{2} \right]^{d/2-1} \dd u^* \dd\sigma_{d-1} (\overline{\PT{x}}^*), 
\end{equation*}
where the density is given by
\begin{equation*}
\begin{split}
(\lambda^*)^{\prime}(u^*) 
&\DEF \frac{\gammafcn(d/2) / W_{s}(\mathbb{S}^{d})}{\gammafcn(d-s/2)} \left( \frac{1+t^*}{1+u^*} \right)^{d/2} \left( \frac{u^*-t^*}{1+t^*} \right)^{(s-d)/2} \HypergeomReg{2}{1}{1,d/2}{1-(d-s)/2}{\frac{u^*-t^*}{1+u^*}}.
\end{split}
\end{equation*}
From \eqref{eq:u_rel} we obtain
\begin{equation*}
\frac{1+t^*}{1+u^*} = \frac{R^2 + 2 R u + 1}{R^2 + 2 R t + 1} \, \frac{1-t}{1-u} 
\end{equation*}
and 
\begin{equation*}
u^* - t^* = \left( R^2 - 1 \right)^2 \frac{t-u}{\left( R^2 + 2 R u + 1 \right)\left( R^2 + 2 R t + 1 \right)}. 
\end{equation*}
From the latter we get
\begin{equation*}
\frac{u^*-t^*}{1+t^*} = \frac{\left( R + 1 \right)^2}{R^2 + 2 R u + 1} \, \frac{t-u}{1-t}, \qquad \frac{u^*-t^*}{1+u^*} = \frac{\left( R + 1 \right)^2}{R^2 + 2 R t + 1} \, \frac{t-u}{1-u}.
\end{equation*}
Furthermore, by \eqref{eq:sigma.d.transformation}
\begin{equation*}
\begin{split}
\frac{\omega_{d-1}}{\omega_{d}} \left[ 1 - (u^*)^{2} \right]^{d/2-1} \dd u^* \dd\sigma_{d-1} (\overline{\PT{x}}^*) 
&= \dd \sigma_d\big|_{\Sigma_t^*}( \PT{x}^* ) = \frac{\left| \PT{x}^* - \PT{b} \right|^d}{\left| \PT{x} - \PT{b} \right|^d} \dd \sigma_d\big|_{\Sigma_t}( \PT{x} ) \\
&= \frac{\left( R^2 - 1 \right)^d}{\left( R^2 + 2 R u + 1 \right)^{2d}} \, \frac{\omega_{d-1}}{\omega_{d}} \left( 1 - u^{2} \right)^{d/2-1} \dd u \dd\sigma_{d-1} (\overline{\PT{x}}).
\end{split}
\end{equation*}
Hence, substituting these relations into
\begin{equation*}
\dd \epsilon_t( \PT{x} ) = \left( R^2 - 1 \right)^{-s/2} \dd \kelvinMEAS_{\PT{b},s}( \lambda^*( \PT{x}^* ) ) = \left( R^2 - 1 \right)^{-s/2} \frac{\left| \PT{x} - \PT{b} \right|^s}{\left( R^2 - 1 \right)^{s/2}} \dd \lambda^*( \PT{x}^* )
\end{equation*}
we arrive after some simplifications at the desired results \eqref{eq:dd.eps.t} and \eqref{eq:eps.dens}. 

Proceeding as in the proof of Lemma~29, we substitute \eqref{eq:dd.eps.t} and \eqref{eq:eps.dens} into
\begin{equation*}
\left\| \epsilon_t \right\| = \int_{\mathbb{S}^d} \dd \epsilon_t = \frac{\omega_{d-1}}{\omega_d} \int_{-1}^t \epsilon_{t}^{\prime}(u) \left( 1 - u^2 \right)^{d/2-1} \dd u.
\end{equation*}
Applying \cite[Lemma~A.1]{BDSarXiv} (which is also valid for $| x y | < 1$), we get ($r^2 = R^2 + 2 R t + 1$)
\begin{equation*}
\begin{split}
\left\| \epsilon_{t} \right\| 
&= 2^{(d-s)/2-1} \frac{\gammafcn(d/2)}{\gammafcn(d-s/2)} \frac{\gammafcn(d/2)}{\gammafcn(s/2)} \frac{\omega_{d-1}}{\omega_{d}} \frac{\left( R - 1 \right)^{d-s}}{W_{s}(\mathbb{S}^{d}) r^{d}} \left( 1 - t \right)^{d/2} \left( 1 + t \right)^{s/2} \\
&\phantom{=\pm}\times \left( 1 - x y \right)^{-d/2} \int_0^1 v^{s/2-1} \left( 1 - x v \right)^{d-s/2-1} \left( 1 - \frac{x \left( 1 - y \right)}{1 - x y} v \right)^{-d/2} \dd v,
\end{split}
\end{equation*}
where
\begin{equation*}
x = \frac{1+t}{2}, \qquad y = \frac{\left( R + 1 \right)^2}{R^2 + 2 R t + 1}
\end{equation*}
and thus
\begin{equation*}
1 - x y = \frac{\left( R - 1 \right)^2}{R^2 + 2 R t + 1} \, \frac{1 - t}{2}, \qquad \frac{x \left( 1 - y \right)}{1 - x y} = - \frac{4 R}{\left( R - 1 \right)^2} \, \frac{1 + t}{2}.
\end{equation*}
Simplification gives the Euler-type integral of an Appell function \cite[Eq.~16.15.1]{NIST:DLMF}
\begin{equation*}
\begin{split}
\left\| \epsilon_{t} \right\| 
&= \frac{2^{-s/2} \gammafcn(d)}{\gammafcn(d-s/2)\gammafcn(s/2)} \frac{1}{W_s(\mathbb{S}^d)} \left( R - 1 \right)^{-s} \left( 1 + t \right)^{s/2} \\
&\phantom{=\pm}\times \int_{0}^{1} u^{s/2-1} \left( 1 - \frac{1+t}{2} u \right)^{d-s/2-1} \left( 1 + \frac{4R}{\left(R-1\right)^{2}} \frac{1+t}{2} u \right)^{-d/2} \dd u.
\end{split}
\end{equation*}
A change of variable $1+v=(1+t) u$ yields \eqref{eq:NormEpsB}.
\end{proof}

With this preparations we are able to prove Proposition~\ref{prop:SignEq.general}. 

\begin{proof}[Proof of Proposition~\ref{prop:SignEq.general}]
Let $-1 < t < 1$. The representation of the signed equilibrium $\eta_t$ follows by substituting the representations of $\nu_t$ (Proposition~\ref{prop:nu}) and $\epsilon_t$ (Lemma~\ref{lem:epsilon}) into \eqref{eq:eta}. 
For the analysis of the behavior of the density $\eta_t^\prime$ near $t^-$ we write \eqref{eta.t.prime.1st.result} as
\begin{equation*}
\eta_{t}^{\prime}(u) = \frac{1}{W_s(\mathbb{S}^d)} \frac{\gammafcn(d/2)}{\gammafcn(d-s/2) \gammafcn( 1 - (d-s)/2 )} \left( \frac{t-u}{1-t} \right)^{(s-d)/2} f( u ),
\end{equation*}
where the function 
\begin{equation*}
\begin{split}
f( u ) \DEF \left( \frac{1-t}{1-u} \right)^{d/2} &\Bigg\{ \Phi_s (t) \Hypergeom{2}{1}{1,d/2}{1-(d-s)/2}{\frac{t-u}{1-u}} \\
&\phantom{=\times\pm}- \frac{q\left( R - 1 \right)^{d-s}}{r^{d}} \Hypergeom{2}{1}{1,d/2}{1-(d-s)/2}{\frac{\left(R+1\right)^{2}}{r^{2}} \, \frac{t-u}{1-u}} \Bigg\}
\end{split}
\end{equation*}
is analytic at $u = t$. (Note that the argument of either hypergeometric function is in the interval $(0,1-\eps)$ if $t \in (-1, 1-\eps]$.) We consider the Taylor expansion 
\begin{equation*}
f( u ) = f( t ) - f^\prime( t ) \left( t - u \right) + \frac{f^{\prime\prime}( \tau )}{1!} \left( t - u \right)^2  
\end{equation*}
for some $u < \tau < t$ whenever $-1 < u < t$. As the Gauss hypergeometric functions evaluate to $1$ at $u = t$, we get
\begin{equation*}
f( t ) = \Phi_s (t) - \frac{q\left( R - 1 \right)^{d-s}}{r^{d}}, 
\end{equation*}
the differentiation formula for Gauss hypergeometric functions (\cite[Eq.~15.5.1]{NIST:DLMF}) gives
\begin{equation*}
\begin{split}
f^\prime( t ) 
&= \frac{d}{2} \frac{1}{1-t} \left\{ \Phi_s (t) - \frac{q\left( R - 1 \right)^{d-s}}{r^{d}} \right\} + \Bigg\{ \Phi_s (t) \frac{\frac{d}{2}}{1 - \frac{d-s}{2}}  \Hypergeom{2}{1}{2,1+d/2}{2-(d-s)/2}{\frac{t-u}{1-u}} \frac{t - 1}{\left( 1 - u \right)^2} \\
&\phantom{=\pm}- \frac{q\left( R - 1 \right)^{d-s}}{r^{d}} \frac{\frac{d}{2}}{1 - \frac{d-s}{2}} \Hypergeom{2}{1}{2,1+d/2}{2-(d-s)/2}{\frac{\left(R+1\right)^{2}}{r^{2}} \, \frac{t-u}{1-u}} \frac{\left(R+1\right)^{2}}{r^{2}} \, \frac{t - 1}{\left( 1 - u \right)^2} \Bigg\}\Bigg|_{u = t}
\end{split}
\end{equation*}
and one can verify that $| f^{\prime\prime}( u ) |$ is uniformly bounded on $[-1,t]$. In particular
\begin{align*}
\mathcal{R}( t ) 
&= - f^\prime( t ) = - \frac{d}{2} \, \frac{f(t)}{1-t} + \frac{\frac{d}{2}}{1 - \frac{d-s}{2}} \frac{1}{1-t} \left\{ \Phi_s (t) - \frac{q\left( R - 1 \right)^{d-s}}{r^{d}} \frac{\left(R+1\right)^{2}}{r^{2}} \right\} \\
&= - \frac{d}{2} \left( 1 - \frac{1}{1 - \frac{d-s}{2}} \right) \frac{f(t)}{1-t} - \frac{\frac{d}{2}}{1 - \frac{d-s}{2}} \frac{1}{1-t} \left( \frac{R^2 + 2 R + 1}{R^2 + 2 R t + 1} - 1 \right).
\end{align*}
Putting everything together and simplification gives \eqref{eq:weighted.eta.density.near.t} and \eqref{eq:weighted.eta.density.near.t.middle.term}.

The constance of the weighted potential on $\Sigma_t$ follows from \eqref{eq:weighted.potential}. It remains to show \eqref{eq:weighted.outside}. We can proceed as in \cite[Section~5]{BrDrSa2009} but using $r^2 = R^2 + 2 R t + 1$ and $\rho^2 = R^2 + 2 R \xi + 1$ and
\begin{equation*}
c_t^2 = \frac{(R+1)^2}{r^2}, \qquad C = \frac{1}{W_s( \mathbb{S}^d )} \frac{\gammafcn( d/2 )}{\gammafcn( d - s/2 )} \, \frac{\left( R - 1 \right)^{d-2}}{r^d}.
\end{equation*}
This leads to the relations (cf. \cite[Section~5]{BrDrSa2009})
\begin{equation*}
1 - z = 1 - c_t^2 \frac{1+t}{2} = \frac{\left( R - 1 \right)^2}{r^2} \, \frac{1 - t}{2}, \qquad w + z - w z = \frac{\rho^2}{r^2} \, \frac{1 + t}{1 + \xi}
\end{equation*}
and, subsequently, to the desired result \eqref{eq:weighted.outside}.

With the help of {\sc MATHEMATICA} we derive the representation for the weighted potential near $t^+$.
\end{proof}

\begin{proof}[Positivity of $\eta_t^\prime$ (Remark after Proposition~\ref{prop:SignEq.general})]
We substitute the series expansion of the regularized hypergeometric function into \eqref{eta.t.prime.1st.result} to obtain
\begin{equation*}
\begin{split} 
\eta_{t}^{\prime}(u) &= \frac{1}{W_s(\mathbb{S}^d)} \frac{\gammafcn(d/2)}{\gammafcn(d-s/2)}
\left( \frac{1-t}{1-u} \right)^{d/2} \left( \frac{t-u}{1-t} \right)^{(s-d)/2} \\
&\phantom{=\times}\times \sum_{n=0}^\infty \frac{\Pochhsymb{d/2}{n}}{\gammafcn(n + 1 - (d-s)/2)} \left( \frac{t-u}{1-u} \right)^n \Bigg\{ \Phi_s (t)
- \frac{q\left( R - 1 \right)^{d-s}}{r^{d}} \left[ \frac{\left(R+1\right)^{2}}{r^{2}} \right]^n  \Bigg\}.
\end{split}
\end{equation*}
Assuming \eqref{eq:weighted.neccessary.density}, it can be readily verified that the expression in braces above is $0$ if $n = 0$ and postive if $n \geq 1$. Hence $\eta_{t}^{\prime}(u) > 0$ for $u \in [-1,t)$.
\end{proof}

\begin{proof}[Proof of Relation~\eqref{eq:weighted.larger.than}]
The (series) expansion
\begin{equation*}
\mathrm{I}(z; a, b) = \left[ \gammafcn(a+b) / \gammafcn(b) \right] z^a \left(
1 - z \right)^b \HypergeomReg{2}{1}{1,a+b}{a+1}{z},
\end{equation*}
applied to \eqref{eq:weighted.outside} yields for $\xi>t>-1$
\begin{equation*}
\begin{split}
&U_s^{\eta_t}(\PT{z}) + Q_{\PT{b},s}(\PT{z}) = \Phi_s(t) + \frac{\gammafcn(d/2)}{\gammafcn(s/2)} \left( \frac{\xi-t}{1+\xi} \right)^{(d-s)/2} \left( \frac{1+t}{1+\xi} \right)^{s/2} \\
&\phantom{=\times}\times \sum_{n=0}^\infty \frac{\Pochhsymb{d/2}{n}}{\gammafcn(n+1+(d-s)/2)} \left( \frac{\xi-t}{1+\xi} \right)^n \left\{ \frac{q \left( R - 1 \right)^{d-s}}{r^d} \left[ \frac{R^2 - 2 R + 1}{R^2 + 2 R t + 1} \right]^n - \Phi_s(t) \right\}.
\end{split}
\end{equation*}
If $q (R+1)^{d-s}/r^d \geq \Phi_s(t)$, then the infinite series above is positive for ever $\xi \in (t, 1]$. 
\end{proof}

\begin{proof}[Proof of Proposition~\ref{prop:s.equilibrium.measure}]
Set $\Delta( t ) \DEF \Phi_s( t ) - q ( R - 1 )^{d-s} / r^d$, where $r = r(t) = \sqrt{R^2 + 2 R t + 1}$. Proceeding as in the Proof \cite[Theorem~13]{BrDrSa2009}, we show that $\Delta(t)$ having a unique solution in $(-1,1]$ is intimately connected with $\Phi_s$ having a unique minimum in $(-1, 1]$. 

Note that $\Phi_s$ and therefore $\Delta( t )$ tend to $+\infty$ as $t \to -1^+$. Hence there is a largest $t_c \in (-1,1]$ such that $\Delta( t ) > 0$ on $(-1, t_c)$ (and $\Delta( t_c ) = 0$ by continuity if $t_c < 1$). From 
\begin{equation} \label{eq:Phi.s.derivative}
\frac{\dd \Phi_{s}}{\dd t} = - \frac{\left\| \nu_t \right\|^\prime}{\left\| \nu_t \right\|} \left[ \Phi_{s}(t) - q \, W_s(\mathbb{S}^d) \frac{\left\| \epsilon_t \right\|^\prime}{\left\| \nu_t \right\|^\prime} \right] = - \frac{\left\| \nu_t \right\|^\prime}{\left\| \nu_t \right\|} \Delta(t),
\end{equation}
where $\| \nu_t \|^\prime / \| \nu_t \| > 0$ on $(-1,1)$, we see that $\Phi_s$ is strictly decreasing on $(-1,t_c)$ (and thus on all of $(-1,1)$ if $t_c = 1$). Hence, in the case $t_c = 1$, the function $\Phi_s$ attains its unique minimum at $1$. (Then $\Delta( t ) \geq 0$ is equivalent with the condition \eqref{neg.eqsigned} with $\PT{a}$ changed to~$\PT{b}$.) Suppose $t_c < 1$. Then any zero $\tau$ of $\Phi_s^\prime$ in $(-1,1)$ is a minimum of $\Phi_s$ as the in $(-1,1)$ twice continuously differentiable function $\Phi_{s}$ satisfies
\begin{equation*}
\frac{\dd^2 \Phi_{s}}{\dd t^2} (\tau) = - \frac{\left\| \nu_t \right\|^\prime}{\left\| \nu_t \right\|}  \frac{d\,q \left( R - 1 \right)^{d-s} R}{r^{d+2}} \Bigg|_{t=\tau} > 0.
\end{equation*}
This implies that $\Phi_s$ has a unique minimum in $(-1,1)$. Relation \eqref{eq:Phi.s.derivative} also implies that $\Delta(t) > 0$ on $(-1,t_c)$ and $\Delta(t) < 0$ on $(t_c, 1)$. Hence, by the first remark after Proposition~\ref{prop:SignEq.general}, $t_c = \max\{ t : \eta_t \geq 0 \}$.

By the remarks after Proposition~\ref{prop:SignEq.general}, the signed measure $\eta_{t_c}$ is a positive probability measure with support $\Sigma_{t_c}$ that has constant weighted $s$-potential on $\supp( \eta_{t_c} )$ and exceeds this constant (away from the support) on $\mathbb{S}^d \setminus \supp( \eta_{t_c} )$. By the variational inequalities (cf. \eqref{geqineq} and \eqref{leqineq}), $\eta_{t_c}$ is the $s$-extremal measure on $\mathbb{S}^d$ associated with $Q_{\PT{b}, s}$. 
\end{proof}

\appendix

\section{The $s$-Potential of the $s$-Extremal Measure $\sigma_d$} 
\label{sec:s.potential}

The $s$-potential $U_s^{\sigma_d}( \PT{a} ) = \int \frac{1}{\left| \PT{x} - \PT{a} \right|^s} \, \dd \sigma_d( \PT{x} )$ is well-defined for all $\PT{a} \in \mathbb{R}^{d+1}$ and all $s \in \mathbb{R}$. The following discussion will be restricted to the potential-theoretical regime $0 < s < d$.\footnote{In the hyper-singular regime $s \geq d$, $U_s^{\sigma_d}$ assumes the value $+\infty$ on $\mathbb{S}^d$ and is finite in~$\mathbb{R}^{d+1} \setminus \mathbb{S}^d$.} 
\hspace{1em}The $s$-potential of $\sigma_d$ is then continuous and uniformly bounded in $\mathbb{R}^{d+1}$. As we shall see, it attains the maximum value $W_s( \mathbb{S}^d ) > 1$ on $\mathbb{S}^d$ in the strictly subharmonic case $d - 1 < s < d$, equals~$1$ on the closed unit ball and is diminishing outside $\mathbb{S}^d$ in the harmonic case $s = d - 1$, and assumes the maximal value $1$ at the center of the sphere in the strictly superharmonic case $0 < s < d - 1$. This is one example of how the potential-theoretical regime governs the behavior of $U_s^{\sigma_d}$. Figure~\ref{fig:typical.s.potentials} illustrates the typical form of $U_s^{\sigma_d}$ utilizing that $U_s^{\sigma_d}( \PT{a} )$ is a radial function depending on~$R = | \PT{a} |$ only. (By abuse of notation we shall write $U_s^{\sigma_d}(R)$.) This dependence on $R$ can be easily seen from the integral representation
\begin{equation} \label{eq:s.potential.integral}
U_s^{\sigma_d}( R ) = \frac{\omega_{d-1}}{\omega_d} \, \int_{-1}^1 \frac{( 1 - t^2 )^{d/2-1}}{( R^2 - 2 R \, t + 1 )^{s/2}} \, \dd t, \qquad R \geq 0,
\end{equation}
which is obtained by using the identity 
\begin{equation*}
| \PT{x} - \PT{a} |^2 = | \PT{a} |^2 - 2 \, \PT{a} \cdot \PT{x} + 1, \qquad \PT{x} \in \mathbb{S}^d, \quad \PT{a} \in \mathbb{R}^{d+1},
\end{equation*}
and the Funk-Hecke formula (cf. \cite{Mu1966}). The standard substitution ${2 u = 1 + t}$ gives the symmetric representation \eqref{eq:s.potential}, valid for $R \geq 0$, in terms of a Gauss hypergeometric function. The quadratic transformation for such functions \cite[Eq.~15.3.17]{AbSt1992} yields the formulas \eqref{eq:s.potential.A} and \eqref{eq:s.potential.B} suitable for the domain $[0,1]$ and~$[1,\infty)$, respectively. The first upper parameter $-( d - 1 - s) / 2$ ``measures'' how far $s$ is away from the harmonic case $d - 1$. If this parameter is a negative integer (that is, $d - s$ is an odd positive integer), then the series expansion of the Gauss hypergeometric function reduces to a polynomial. 
For $d - s$ is an even integer, the $s$-potential $U_s^{\sigma_d}$ reduces to a linear combination of complete elliptic integrals of the first and second kind with coefficients that are rational functions of $R$ if $d$ is odd, whereas for even sphere dimension $d$, the $s$-potential is a sum of a rational function in $R$ and another rational function in $R$ times a logarithmic term in $R$. For $( d - s ) / 2$ not an integer and $d$ even, the $s$-potential is a linear combination of $(R + 1)^{d-2}$ and $| R - 1 |^{d-s}$ with coefficients that are rational functions in $R$. The derivation of these representations of $U_s^{\sigma_d}$ are sketched out after Theorem~\ref{thm:Gonchar.A}. 

We shall assume that $d \geq 2$.

\subsection*{Monotonicity Properties} 
The $s$-potential $U_s^{\sigma_d}(R)$ is strictly monotonically decreasing on the interval $(1,\infty)$ for every $s \in (0,d)$ as can be seen by differentiating~\eqref{eq:s.potential.integral}, also cf. \eqref{eq:differentiation.formula.B} below. Next, we consider $U_s^{\sigma_d}(R)$ on $(0,1)$. Differentiating \eqref{eq:s.potential.B} using \cite[Eq.~15.5.1]{NIST:DLMF}, we obtain
\begin{equation*}
\frac{\dd U_s^{\sigma_d}(R)}{\dd R} = - \frac{( d - 1 - s ) s}{d + 1} \, R \, \Hypergeom{2}{1}{1-(d-1-s)/2,1+s/2}{1+(d+1)/2}{R^2}.
\end{equation*}
The hypergeometric function is positive for every $s \in (0,d)$ which follows, e.g., from its integral representation (cf. \cite[Eq.~15.6.1]{NIST:DLMF}). Therefore, $U_s^{\sigma_d}(R)$ is strictly monotonically decreasing on $(0,1)$ in the strictly superharmonic case, constant on $(0,1)$ in the harmonic case, and strictly monotonically increasing on $(0,1)$ in the strictly subharmonic case. We further infer that $U_s^{\sigma_d}(R)$ has a unique maximum at $R = 0$ with value $1$ when $0 < s < d - 1$, assumes the maximum value $1$ everywhere on $[0,1]$ if $s = d - 1$, and has a unique maximum at $R = 1$ with value $W_s( \mathbb{S}^d ) > 1$ (since $W_s( \mathbb{S}^d ) = U_s^{\sigma_d}( 1 ) > U_s^{\sigma_d}( 0 ) = 1$) if~$d - 1 < s < d$.

\subsection*{The Critical Point $R = 1$} 
In the strictly subharmonic case, $U_s^{\sigma_d}(R)$ has a cusp at $R = 1$ with
\begin{equation*}
\frac{\dd U_s^{\sigma_d}(R)}{\dd R} \to +\infty \quad \text{as $R \to 1^-$}, \qquad \frac{\dd U_s^{\sigma_d}( R )}{\dd R} \to -\infty \quad \text{as $R \to 1^+$.}
\end{equation*}
(Informally, this makes the whole sphere $\mathbb{S}^d$ into a ``cusp'' if $d - 1 < s < d$.) This can be seen from the following differentiation formula (derived from \eqref{eq:s.potential} using \cite[Eq.~15.5.1 and Eq.~15.8.1]{NIST:DLMF}) 
\begin{equation} \label{eq:differentiation.formula.A}
\begin{split}
\frac{\dd U_s^{\sigma_d}(R)}{\dd R} 
&= -s \frac{U_s^{\sigma_d}(R)}{R+1} - s \frac{( R - 1 ) | R - 1 |^{d-s-2}}{( R + 1 )^{d+1}} \\
&\phantom{=\pm}\times \Hypergeom{2}{1}{d-s/2,d/2}{1+d}{\frac{4R}{(R+1)^2}},
\end{split}
\end{equation}
where the hypergeometric functions remains finite for all $R \geq 0$ if $d - 2 < s < d$. In the harmonic case one clearly has
\begin{equation*}
\frac{\dd U_s^{\sigma_d}(R)}{\dd R} \to 0 \quad \text{as $R \to 1^-$}, \qquad \frac{\dd U_s^{\sigma_d}( R )}{\dd R} \to 1 - d \quad \text{as $R \to 1^+$.}
\end{equation*}
In the strictly superharmonic case one has 
\begin{equation} \label{eq:critical.R.EQ.1.limit}
\lim_{R \to 1} \frac{\dd U_s^{\sigma_d}(R)}{\dd R} = - \frac{s}{2} \, W_s( \mathbb{S}^d ),
\end{equation}
as can be seen from the differentiation formula (using only \cite[Eq.~15.5.1]{NIST:DLMF} on~\eqref{eq:s.potential})
\begin{equation} \label{eq:differentiation.formula.B}
\begin{split}
\frac{\dd U_s^{\sigma_d}(R)}{\dd R} 
&= -s \frac{U_s^{\sigma_d}(R)}{R+1} - s \frac{R - 1}{( R + 1 )^{s+3}} \\
&\phantom{=\pm}\times \Hypergeom{2}{1}{1+d/2,1+s/2}{1+d}{\frac{4R}{(R+1)^2}},
\end{split}
\end{equation}
where the hypergeometric functions remains finite for all $R \geq 0$ if $0 < s < d - 2$. For $s = d - 2$, both hypergeometric functions, above and in \eqref{eq:differentiation.formula.A}, are zero-balanced (that is, the sum of the upper parameters equals the lower parameter). After application of the linear transformation \cite[Eq.~15.8.10]{NIST:DLMF} this leads to a term $( R - 1) \, \log [ ( R - 1 )^2 / ( R + 1 )^2 ]$ which goes to $0$ as $R \to 1$.
Next, we consider the behavior of the second derivative of $U_s^{\sigma_d}(R)$ as $R \to 1$ in the strictly superharmonic case. From \eqref{eq:critical.R.EQ.1.limit} and \eqref{eq:differentiation.formula.B} we obtain that
\begin{equation*}
\begin{split}
\frac{\frac{\dd U_s^{\sigma_d}(R)}{\dd R} - ( - \frac{s}{2} \, W_s( \mathbb{S}^d ) )}{R - 1} 
&= \frac{-\frac{s}{R+1} \, U_s^{\sigma_d}(R) + \frac{s}{2} \, W_s( \mathbb{S}^d )}{R-1} \\
&\phantom{=}- \frac{s}{( R + 1 )^{s+3}} \, \Hypergeom{2}{1}{1+d/2,1+s/2}{1+d}{\frac{4R}{(R+1)^2}}.
\end{split}
\end{equation*}
For $s \geq d - 2$, the hypergeometric function goes to $+\infty$ as $R \to 1$ (using, e.g., integral formula) and for $0 < s < d - 2$ it assumes the following value (after applying \cite[Eq.~15.4.20]{NIST:DLMF} and \eqref{eq:W.s.S.d}) 
\begin{equation*}
\Hypergeom{2}{1}{1+d/2,1+s/2}{1+d}{1} = \frac{\gammafcn( 1 + d ) \gammafcn( (d - s)/2 - 1 )}{\gammafcn( d/2 ) \gammafcn( d - s/2 )} = \frac{2^{s+1} d}{d - 2 - s} \, W_s( \mathbb{S}^d ). 
\end{equation*}
On observing that $2 / (R+1) = 1 - (R - 1) / (R + 1)$, we also have that (as~$R \to 1$)
\begin{align*}
\frac{-\frac{s}{R+1} \, U_s^{\sigma_d}(R) + \frac{s}{2} \, W_s( \mathbb{S}^d )}{R-1} 
&= - \frac{s}{2} \Big( \frac{U_s^{\sigma_d}(R) - W_s( \mathbb{S}^d )}{R - 1} - \frac{U_s^{\sigma_d}(R)}{R + 1} \Big) \\
&\to - \frac{s}{2} \Big( \frac{\dd U_s^{\sigma_d}}{\dd R}( 1 ) - \frac{1}{2} \, U_s^{\sigma_d}( 1 ) \Big) = \frac{s}{2} \, \frac{s+1}{2} \, W_s( \mathbb{S}^d ).
\end{align*}
We conclude that for $d - 2 \leq s < d - 1$ (and $s > 0$),
\begin{equation} \label{eq:2nd.derivative.at.1.limit.A}
\frac{\dd^2 U_s^{\sigma_d}(R)}{\dd R^2} \to -\infty \quad \text{as $R \to 1$}
\end{equation}
and for $0 < s < d - 2$, 
\begin{equation} \label{eq:2nd.derivative.at.1.limit.B}
\lim_{R \to 1} \frac{\dd^2 U_s^{\sigma_d}(R)}{\dd R^2} = \frac{s}{4} \Big[ s + 1 - \frac{d}{d-2-s} \Big] W_s( \mathbb{S}^d ).
\end{equation}
Observe that the latter is negative for $s$ sufficiently close to $0^+$ or $(d - 2)^-$, provided $d \geq 3$.

\subsection*{Convexity Properties} 
The $s$-potential $U_s^{\sigma_d}(R)$ is strictly convex on the interval $(R_s,\infty)$ for every ${s \in (0,d)}$. The existence of $R_s \geq 1$ follows from the fact that $U_s^{\sigma_d}(R)$ is convex on $( ( 2 + s ) / ( 1 + s ), \infty )$ as the integrand of the second derivative of \eqref{eq:s.potential.integral} is positive for $R$ in this interval as can be seen from
\begin{equation*}
\frac{\dd^2 }{\dd R^2} \Big\{ ( R^2 - 2 R \, t + 1 )^{-s/2} \Big\} = s \, \frac{( s + 1 ) ( R - t )^2 - ( 1 - t^2 )}{( R^2 - 2 R \, t + 1 )^{s/2 + 2}}, \qquad t \in [-1,1].
\end{equation*}
In the strictly subharmonic case an inspection of the signs in the second derivative of \eqref{eq:s.potential.A},
\begin{equation*}
\begin{split}
\frac{\dd^2 U_s^{\sigma_d}( R )}{\dd R^2} 
&= +s \, \frac{U_s^{\sigma_d}(R)}{R^2} - \frac{s}{R} \, \frac{\dd U_s^{\sigma_d}( R )}{\dd R} + \frac{\dd}{\dd R} \Big\{  \frac{( d - 1 - s ) s}{d+1} \, R^{-s-3} \\
&\phantom{=\pm}\times \Hypergeom{2}{1}{1-(d-1-s)/2,1+s/2}{1+(d+1)/2}{\frac{1}{R^2}} \Big\},
\end{split}
\end{equation*}
shows that $U_s^{\sigma_d}( R )$ is convex on $( 1, \infty )$. In the harmonic case, $U_{d-1}^{\sigma_d}( R ) = R^{1-d}$ is strictly convex on $( 1, \infty )$ provided $d \geq 2$. In the strictly superharmonic regime convexity is a more subtle property. For $d - 2 \leq s <  d - 1$ and $s > 0$, the second derivative of $U_s^{\sigma_d}(R)$ is negative near $1$ (cf. \eqref{eq:2nd.derivative.at.1.limit.A}). Similarly, for $s$ sufficiently close to $0^+$ or $(d - 2)^-$, this derivative is also negative (cf. \eqref{eq:2nd.derivative.at.1.limit.B}). However, $U_s^{\sigma_d}(R)$ can be strictly convex on $(1,\infty)$ even for $s \in (0, d - 2)$, $d \geq 3$, as the following example for $d = 7$ and $s = 2$ demonstrates:
\begin{equation*}
U_{2}^{\sigma_7}(R) = \frac{R^4 - \frac{1}{2} \, R^2 + \frac{1}{10}}{R^6}, \quad  \frac{\dd^2 U_{2}^{\sigma_7}(R)}{\dd R^2} = 6 \frac{( R^2 - \frac{5}{6} )^2 + \frac{1}{180}}{R^8}, \qquad R \geq 1.
\end{equation*}

The $s$-potential $U_s^{\sigma_d}(R)$ is strictly convex on $(0, 1)$ in the strictly subharmonic regime. This follows from differentiating the series representation of \eqref{eq:s.potential.B} twice and rewrite it as
\begin{equation*}
\frac{\dd^2 U_s^{\sigma_d}( R )}{\dd R^2} = - \frac{( d - 1 - s ) s}{d+1} \, \Hypergeom{3}{2}{1-(d-1-s)/2,1+s/2,3/2}{1+(d+1)/2,1/2}{R^2},
\end{equation*}
where convergence is assured for $0 \leq R < 1$. We can also infer that the $s$-potential $U_s^{\sigma_d}(R)$ is strictly concave on $(0,1)$ in the strictly superharmonic case $s \in ( \max\{ 0, d - 3\}, d - 1 )$ and this property extends to $s = d - 3$ if $d \geq 4$. Indeed, as
\begin{equation*}
\lim_{R \to 0} \frac{\dd^2 U_s^{\sigma_d}( R )}{\dd R^2} = - \frac{( d - 1 - s ) s}{d+1}
\end{equation*}
is negative in the strictly superharmonic regime, the $s$-potential $U_s^{\sigma_d}(R)$ is always strictly concave on $(0,R_s)$ for some $0 < R_s \leq 1$ for $s \in (0, d - 1)$.


\begin{thebibliography}{10}

\bibitem{AaFoKr2001}
J.~Aarts, R.~Fokkink, and G.~Kruijtzer.
\newblock Morphic numbers.
\newblock {\em Nieuw Arch. Wiskd. (5)}, 2(1):56--58, 2001.

\bibitem{AbSt1992}
M.~Abramowitz and I.~A. Stegun, editors.
\newblock {\em Handbook of mathematical functions with formulas, graphs, and
  mathematical tables}.
\newblock Dover Publications Inc., New York, 1992.
\newblock Reprint of the 1972 edition.

\bibitem{Br2008}
J.~S. Brauchart.
\newblock Optimal logarithmic energy points on the unit sphere.
\newblock {\em Math. Comp.}, 77(263):1599--1613, 2008.

\bibitem{BDSarXiv}
J.~S. Brauchart, P.~D. Dragnev, and E.~B. Saff.
\newblock Minimal {R}iesz energy on the sphere for axis-supported external
  fields.
\newblock arXiv:0902.1558 [math-ph], Feb 2009.

\bibitem{BrDrSa2009}
J.~S. Brauchart, P.~D. Dragnev, and E.~B. Saff.
\newblock Riesz extremal measures on the sphere for axis-supported external
  fields.
\newblock {\em J. Math. Anal. Appl.}, 356(2):769--792, 2009.

\bibitem{BrDrSa2014}
J.~S. Brauchart, P.~D. Dragnev, and E.~B. Saff.
\newblock Riesz external field problems on the hypersphere and optimal point
  separation.
\newblock {\em Potential Anal.}, 2014 (accepted).

\bibitem{BrDrSa2012}
J.~S. Brauchart, P.~D. Dragnev, E.~B. Saff, and C.~E. van~de Woestijne.
\newblock A fascinating polynomial sequence arising from an electrostatics
  problem on the sphere.
\newblock {\em Acta Math. Hungar.}, 137(1-2):10--26, 2012.

\bibitem{NIST:DLMF}
{NIST Digital Library of Mathematical Functions}.
\newblock http://dlmf.nist.gov/, Release 1.0.6 of 2013-05-06.
\newblock Online companion to \cite{Olver:2010:NHMF}.

\bibitem{Dr2007a}
P.~D. Dragnev.
\newblock On the separation of logarithmic points on the sphere.
\newblock In {\em Approximation theory, {X} ({S}t. {L}ouis, {MO}, 2001)},
  Innov. Appl. Math., pages 137--144. Vanderbilt Univ. Press, Nashville, TN,
  2002.

\bibitem{Dr2007}
P.~D. Dragnev.
\newblock On an energy problem with riesz external ﬁeld.
\newblock {\em Oberwolfach reports}, 4(2):1042--1044, 2007.

\bibitem{DrSa2007}
P.~D. Dragnev and E.~B. Saff.
\newblock Riesz spherical potentials with external fields and minimal energy
  points separation.
\newblock {\em Potential Anal.}, 26(2):139--162, 2007.

\bibitem{HaWeZo2012}
H.~Harbrecht, W.~L. Wendland, and N.~Zorii.
\newblock On {R}iesz minimal energy problems.
\newblock {\em J. Math. Anal. Appl.}, 393(2):397--412, 2012.

\bibitem{HaSa2004}
D.~P. Hardin and E.~B. Saff.
\newblock Discretizing manifolds via minimum energy points.
\newblock {\em Notices Amer. Math. Soc.}, 51(10):1186--1194, 2004.

\bibitem{Ja1998}
J.~D. Jackson.
\newblock {\em Classical electrodynamics}.
\newblock John Wiley \& Sons Inc., New York, third edition, 1998.

\bibitem{LaSaVa1979}
M.~Lachance, E.~B. Saff, and R.~S. Varga.
\newblock Inequalities for polynomials with a prescribed zero.
\newblock {\em Math. Z.}, 168(2):105--116, 1979.

\bibitem{La2013}
M.~Lamprecht.
\newblock On the zeros of {G}onchar polynomials.
\newblock {\em Proc. Amer. Math. Soc.}, 141(8):2763--2766, 2013.

\bibitem{La1972}
N.~S. Landkof.
\newblock {\em Foundations of modern potential theory}.
\newblock Springer-Verlag, New York, 1972.
\newblock Translated from the Russian by A. P. Doohovskoy, Die Grundlehren der
  mathematischen Wissenschaften, Band 180.

\bibitem{LoMaNeEtal2013}
G.~L{\'o}pez~Lagomasino, A.~Mart{\'{\i}}nez~Finkelshtein, P.~Nevai, and E.~B.
  Saff.
\newblock Andrei {A}leksandrovich {G}onchar {N}ovember 21, 1931--{O}ctober 10,
  2012.
\newblock {\em J. Approx. Theory}, 172:A1--A13, 2013.

\bibitem{Me2012}
A.~Melman.
\newblock Geometry of trinomials.
\newblock {\em Pacific J. Math.}, 259(1):141--159, 2012.

\bibitem{MhSa1985}
H.~N. Mhaskar and E.~B. Saff.
\newblock Where does the sup norm of a weighted polynomial live? {A}
  {G}eneralization of {I}ncomplete {P}olynomials.
\newblock {\em Constr. Approx.}, 1(1):71--91, 1985.

\bibitem{Mu1966}
C.~M{\"u}ller.
\newblock {\em Spherical harmonics}, volume~17 of {\em Lecture Notes in
  Mathematics}.
\newblock Springer-Verlag, Berlin, 1966.

\bibitem{OEIS2013}
{OEIS Foundation Inc.}
\newblock {The On-Line Encyclopedia of Integer Sequences}.
\newblock published electronically at http://oeis.org, 2013.

\bibitem{Olver:2010:NHMF}
F.~W.~J. Olver, D.~W. Lozier, R.~F. Boisvert, and C.~W. Clark, editors.
\newblock {\em {NIST Handbook of Mathematical Functions}}.
\newblock Cambridge University Press, New York, NY, 2010.
\newblock Print companion to \cite{NIST:DLMF}.

\bibitem{Pa2002}
R.~Padovan.
\newblock {Dom Hans Van Der Laan and the Plastic Number}.
\newblock In K.~Williams and J.~F. Rodrigues, editors, {\em Nexus IV:
  Architecture and Mathematics}, pages 181--193. Kim Williams Books, Fucecchio
  (Florence), 2002.
\newblock http://www.nexusjournal.com/conferences/N2002-Padovan.html.

\bibitem{Ro2011}
F.~R\o{}nning.
\newblock Gyllent snitt og plastisk tall.
\newblock {\em Tangenten : tidsskrift for matematikk i grunnskolen}, 22(4),
  2011.

\bibitem{SaTo1997}
E.~B. Saff and V.~Totik.
\newblock {\em Logarithmic potentials with external fields}, volume 316 of {\em
  Grundlehren der Mathematischen Wissenschaften [Fundamental Principles of
  Mathematical Sciences]}.
\newblock Springer-Verlag, Berlin, 1997.
\newblock Appendix B by Thomas Bloom.

\bibitem{Si2005}
P.~Simeonov.
\newblock A weighted energy problem for a class of admissible weights.
\newblock {\em Houston J. Math.}, 31(4):1245--1260, 2005.

\bibitem{St1996}
I.~Stewart.
\newblock Tales of a neglected number.
\newblock {\em Sci. Amer.}, 274(June):102--103, 1996.

\bibitem{Sz2010}
P.~G. Szab{\'o}.
\newblock On the roots of the trinomial equation.
\newblock {\em CEJOR Cent. Eur. J. Oper. Res.}, 18(1):97--104, 2010.

\bibitem{vdLa1960}
H.~van~der Laan.
\newblock {\em Le Nombre Plastique: quinze Le\c{c}ons sur l'Ordonnance
  architectonique}.
\newblock Brill, Leiden, 1960.

\bibitem{Zo2003}
N.~V. Zori{\u\i}.
\newblock Equilibrium potentials with external fields.
\newblock {\em Ukra\"\i n. Mat. Zh.}, 55(9):1178--1195, 2003.

\bibitem{Zo2003a}
N.~V. Zori{\u\i}.
\newblock Equilibrium problems for potentials with external fields.
\newblock {\em Ukra\"\i n. Mat. Zh.}, 55(10):1315--1339, 2003.

\bibitem{Zo2004}
N.~V. Zori{\u\i}.
\newblock Potential theory with respect to consistent kernels: a completeness
  theorem, and sequences of potentials.
\newblock {\em Ukra\"\i n. Mat. Zh.}, 56(11):1513--1526, 2004.

\end{thebibliography}
\end{document}